%% file: main.tex
\documentclass[acmsmall]{acmart}

\AtBeginDocument{%
  }

\setcopyright{rightsretained}
\acmDOI{10.1145/3689749}
\acmYear{2024}
\acmJournal{PACMPL}
\acmVolume{8}
\acmNumber{OOPSLA2}
\acmArticle{309}
\acmMonth{10}
\acmSubmissionID{oopslab24main-p392-p}
\ifx\short\undefined
\else
  \received{2024-04-05}
  \received[accepted]{2024-08-18}
\fi

\setcitestyle{nosort}

\include{setup}

\include{macros}

\ifx\short\undefined
\else
\fi

\hypersetup{
  pdftitle   = {Extending the C/C++ Memory Model with Inline Assembly},
  pdfsubject = {Technical Report}
}

\begin{document}

\title{Extending the C/\cpp Memory Model with Inline Assembly}

\author{Paulo Em\'{i}lio de Vilhena}
\email{p.de-vilhena@imperial.ac.uk}
\orcid{0000-0001-7379-310X}
\affiliation{%
  \institution{Imperial College London}
  \city{London}
  \country{United Kingdom}
}

\author{Ori Lahav}
\email{orilahav@tau.ac.il}
\orcid{0000-0003-4305-6998}
\affiliation{%
  \institution{Tel Aviv University}
  \city{Tel Aviv}
  \country{Israel}
}

\author{Viktor Vafeiadis}
\email{viktor@mpi-sws.org}
\orcid{0000-0001-8436-0334}
\affiliation{%
  \institution{MPI-SWS}
  \city{Kaiserslautern}
  \country{Germany}
}

\author{Azalea Raad}
\email{azalea.raad@imperial.ac.uk}
\orcid{0000-0002-2319-3242}
\affiliation{%
  \institution{Imperial College London}
  \city{London}
  \country{United Kingdom}
}

\include{abstract}



\maketitle

\input{intro}

\input{overview}

\input{model}

\input{metatheory}

\input{related}

\input{conclusion}

\input{acknowledgments}

\bibliography{english,local}

\include{appendix}

\end{document}

%% file: setup.tex
\usepackage[utf8]{inputenc}
\usepackage[T1]{fontenc}
\usepackage{fontawesome}
\usepackage[english]{babel}
\usepackage{hhline}
\usepackage{colortbl} %
\usepackage{amsmath}
\usepackage{mathtools}
\usepackage{array}
\usepackage{xcolor} 
\usepackage{comment}
\usepackage{xspace}
\usepackage{microtype}
\usepackage{centernot}
\usepackage{tikz}
\usepackage{tikz-cd}
\usetikzlibrary{positioning,backgrounds,fit,shapes}
\usetikzlibrary{decorations.pathmorphing} 
\usetikzlibrary{external}
\usepackage{adjustbox}
\usepackage{epigraph}
\usepackage[most]{tcolorbox}
\usepackage{float}
\usepackage{wrapfig}
\usepackage{caption}
\usepackage{subcaption}
\usepackage{xparse} %
\usepackage{hyperref}
\usepackage{cleveref}

\usepackage{stmaryrd} %

\input{tango-colors}

\definecolor{ltblue}   {rgb} {  0, 0.4, 0.4}

\definecolor{dkblue}   {rgb} {  0, 0.1, 0.6}
\definecolor{dkgreen}  {rgb} {  0, 0.2,   0}
\definecolor{dkviolet} {rgb} {0.3,   0, 0.5}
\definecolor{dkred}    {rgb} {0.5,   0,   0}

\hypersetup{
  bookmarksopen  = true,
  colorlinks     = true,
  linkcolor      = blue,
  citecolor      = blue,
  urlcolor       = blue,
  hypertexnames  = false,
}

\usepackage{fmtcount}
\usepackage{locallabel} %

\let \softlink   \hyperlink
\let \softtarget \hypertarget

\renewcommand{\LlabelTypeset}[1]{(\textrm{#1})}

\renewcommand{\appearance}[1]{\padzeroes[2]{\decimalnum{#1}}}

\usepackage{mathpartir}

\let\TirName\textsc
\renewcommand{\TirNameStyle}[1] {\hypertarget{#1}{\TirName{#1}}}
\renewcommand{\RefTirName}  [1] {\hypertarget{#1}{\TirName{#1}}}
\renewcommand{\DefTirName}  [1] {\hyperlink{#1}{\TirName{#1}}}

\newcommand  {\RULE}        [1] {{\DefTirName{#1}}\xspace}

\makeatletter
  \def\arrcr{\@arraycr}
\makeatother

\usepackage{verbatim} %
\usepackage[inline]{enumitem} %

\tikzset{
  math array/.style={
    execute at begin node=$\def\arraypreamble{#1},
    execute at end node=$,
  },
  node text/.style={node contents=\makearray{#1}},
}
\newcommand{\makearray}[1]{%
  \expandafter\array\expandafter{\arraypreamble}
  #1
  \endarray
}

\ifx\arxiv\undefined
\else
\fi

%% file: tango-colors.tex
\definecolor{butter1}{HTML}{FCE94F}
\definecolor{butter2}{HTML}{EDD400}
\definecolor{butter3}{HTML}{C4A000}

\definecolor{orange1}{HTML}{FCAF3E}
\definecolor{orange2}{HTML}{F57900}
\definecolor{orange3}{HTML}{CE5C00}

\definecolor{chocolate1}{HTML}{E9B96E}
\definecolor{chocolate2}{HTML}{C17D11}
\definecolor{chocolate3}{HTML}{8F5902}

\definecolor{chameleon1}{HTML}{8AE234}
\definecolor{chameleon2}{HTML}{73D216}
\definecolor{chameleon3}{HTML}{4E9A06}

\definecolor{skyblue1}{HTML}{729FCF}
\definecolor{skyblue2}{HTML}{3465A4}
\definecolor{skyblue3}{HTML}{204A87}

\definecolor{plum1}{HTML}{AD7FA8}
\definecolor{plum2}{HTML}{75507B}
\definecolor{plum3}{HTML}{5C3566}

\definecolor{red1}{HTML}{EF2929}
\definecolor{red2}{HTML}{CC0000}
\definecolor{red3}{HTML}{A40000}

\definecolor{aluminium1}{HTML}{EEEEEC}
\definecolor{aluminium2}{HTML}{D3D7CF}
\definecolor{aluminium3}{HTML}{BABDB6}

\definecolor{slate1}{HTML}{888A85}
\definecolor{slate2}{HTML}{555753}
\definecolor{slate3}{HTML}{2E3436}

%% file: macros.tex
\newcommand\appendixRef[1]{(\cref{#1})}
\newcommand\appendixCite[1]{\cite[\cref*{#1}]{appendix}}
\ifx\short\undefined
  \newcommand\appendixref\appendixRef
\else
  \newcommand\appendixref\appendixCite
\fi
\newcommand{\sref}[1]{\cref{#1}}
\newcommand{\fref}[1]{Figure~\ref{#1}}
\newcommand{\lemmaref}[1]{Lemma~\ref{#1}}
\newcommand{\cororef}[1]{Corollary~\ref{#1}}
\newcommand{\thmref}[1]{Theorem~\ref{#1}}
\newcommand{\defref}[1]{Definition~\ref{#1}}

\newcommand{\cpp}{C\texttt{++}\xspace}

\newcommand{\ccpp}{C/C\texttt{++}\xspace}

\newcommand{\intel}{\textsf{x86}\xspace}
\newcommand{\intelext}{\textsf{Ex86}\xspace}

\NewDocumentCommand{\arm}{s s}{%
\IfBooleanTF{#1}{%
\IfBooleanTF{#2}{%
\textnormal{Armv8}}{%
\textnormal{Armv7}}}{%
\textnormal{Arm}}%
\xspace
}

\newcommand{\power}{\textnormal{POWER}\xspace}
\newcommand{\SC}{\textsf{SC}}
\newcommand{\rcelevenname}{\textsf{RC11}\xspace}
\newcommand{\rceleven}{\rcelevenname}
\newcommand{\rcelevenext}{\ensuremath{\rceleven^{\intelext}}\xspace}

\newcommand{\rcelevenlang}{\rceleven-lang\xspace}
\newcommand{\intelextlang}{\intelext-lang\xspace}

\newcommand{\rcelevenextlang}{\rcelevenext-lang\xspace}

\newcommand{\asm}{\textsf{asm}\xspace}
\newcommand{\rc}{\scalebox{.70}{\rcelevenname}}

\newcommand{\boite}[1]{\begin{array}{@{}l@{}}#1\end{array}}
\newcommand{\crqqb}[1]{\\\kern1.5em\boite{#1}}

\newcommand{\judgment}[2][]{\hspace{.2cm}\textit{#1}\hspace{\stretch{1}}\fbox{$#2$}}
\newcommand{\judgmentpar}[1]{\hspace{.2cm}\textit{#1}}

\NewDocumentCommand{\todo}{o}{
  \IfNoValueTF{#1}{\textcolor{red1}{TODO}}
                  {\textcolor{red1}{TODO:~#1}}
}

\NewDocumentCommand{\absurd}{o}{\textit{%
  \IfNoValueTF{#1}{\textcolor{red1}{(Absurd)}}
                  {\textcolor{red1}{(Absurd #1)}}
}}

\newcommand{\greenhl}[1]{%
  \begingroup
  \setlength{\fboxsep}{0pt}%
  \colorbox{green!40!black!08}{\strut#1\/}%
  \endgroup
}

\newcommand{\mathhl}[1]{\greenhl{\ensuremath{#1}}}

\NewDocumentCommand{\msmall}{O{.85} m}{\scalebox{#1}{\ensuremath{#2}}}

\newcommand{\eqdef}{\;\triangleq\;}

\newcommand{\pair}[2]{\left(#1,\,#2\right)}

\newcommand{\fst}[1]{{#1}.1}
\newcommand{\snd}[1]{{#1}.2}

\newcommand{\maximum}[1]{\mathit{max}\,(#1)}
\newcommand{\minimum}[1]{\mathit{min}\,(#1)}

\ExplSyntaxOn
\NewDocumentCommand{\tuple}{ m >{\SplitList{,}} m}
  {(#1\ProcessList{#2}{\davs__tuple_elem:n})}
\cs_new_protected:Nn \davs__tuple_elem:n{,\,#1} %
\ExplSyntaxOff

\newcommand{\dom}[1]{\mathit{dom}(#1)}
\newcommand{\codom}[1]{\mathit{codom}(#1)}

\newcommand{\metaIfThenElse}[3]
{\mathit{if}\;#1\;\mathit{then}\;#2\;\mathit{else}\;#3}

\newcommand{\wtt}[2]{#1 : #2}

\newcommand{\inti}{i}
\newcommand{\intj}{j}

\newcommand{\intn}{n}
\newcommand{\intm}{m}

\newcommand{\kw}[1]{{\normalfont\texttt{#1}}}

\newcommand{\scol}{\kw{;}}

\newcommand{\Skip}{\kw{skip}}
\newcommand{\Seq}[2]{#1\scol\,#2}

\newcommand{\While}[2]{\kw{while}\;#1\;\kw{\{}\,#2\,\kw{\}}}

\NewDocumentCommand{\Asm}{s o}{
  \IfBooleanTF{#1}{\kw{\textbf{asm}}}
                  {\kw{\textbf{asm}}\,\kw{\{}#2\kw{\}}}
}

\newcommand{\ReadPL}[3]{#2\kw{:=}\,\kw{[}#3\kw{]}^{#1}}
\newcommand{\AsmReadPL}[2]{\Asm[\ReadPL{}{#1}{#2}]}
\newcommand{\WritePL}[3]{\kw{[}#2\kw{]}^{#1}\kw{:=}\,#3}
\newcommand{\AsmWritePL}[2]{\Asm[\WritePL{}{#1}{#2}]}
\newcommand{\NTWritePL}[2]{[#1]\kw{:=}_{\kw{nt}}\,#2}
\newcommand{\AsmNTWritePL}[2]{\Asm[\NTWritePL{#1}{#2}]}

\newcommand{\FencePL}[1]{\kw{fence}_{#1}}
\newcommand{\MFencePL}{\kw{mfence}\xspace}
\newcommand{\AsmMFencePL}{\Asm[\MFencePL]}
\newcommand{\SFencePL}{\kw{sfence}\xspace}
\newcommand{\AsmSFencePL}{\Asm[\SFencePL]}
\newcommand{\RMWPL}[5]{%
\def\tmpinput{#1}\ifx\temp\empty
  \def\tmpkeyword{\kw{rmw}}
\else
  \def\tmpkeyword{\kw{rmw}_{\tmpinput}}
\fi
#2\,\kw{:=}\,\tmpkeyword\,(\kw{[}#3\kw{]},\,#4,\,#5)
}
\newcommand{\AsmRMWPL}[4]{\Asm[\RMWPL{}{#1}{#2}{#3}{#4}]}
\newcommand{\na}{\kw{na}\xspace}
\newcommand{\rlx}{\kw{rlx}\xspace}
\newcommand{\rel}{\kw{rel}\xspace}

\newcommand{\acq}{\kw{acq}\xspace}
\newcommand{\acqrel}{\kw{acqrel}\xspace}
\renewcommand{\sc}{\kw{sc}\xspace}
\newcommand{\nt}{\kw{nt}\xspace}
\newcommand{\tso}{\kw{tso}\xspace}
\newcommand{\stf}{\kw{sf}\xspace}
\newcommand{\mf}{\kw{mf}\xspace}

\newcommand{\ld}{\kw{ld}}

\newcommand{\sy}{\kw{sy}}

\newcommand{\PlusS}{\kw{+}}
\newcommand{\TimesS}{\kw{*}}
\newcommand{\SubS}{\kw{-}}

\newcommand{\Plus}[2]{#1 \;\PlusS\; #2}
\newcommand{\Times}[2]{#1 \;\TimesS\; #2}
\newcommand{\Sub}[2]{#1 \;\SubS\; #2}

\newcommand{\IfThen}[2]{\kw{if}~#1~\kw{\{}~#2~\kw{\}}}
\newcommand{\IfThenVertical}[2]{%
\begin{array}{@{}l@{}}
\kw{if}~\kw{(}#1\kw{)}~\kw{\{}\\
\quad#2\\
\kw{\}}
\end{array}
}

\newcommand{\toPool}[1]{\mathit{toPool}(#1)}

\newcommand{\UB}{\mathit{UB}}

\newcommand{\headstep}{\;\longrightarrow\;}
\newcommand{\multiheadstep}{\;\longrightarrow^{*}\;}

\NewDocumentCommand{\poolstep}{s m m m m}{
  \config{#2}{#3}
    \IfBooleanTF{#1}{\multiheadstep}{\headstep}
      \config{#4}{#5}
}
\newcommand{\progbeh}[2]{#1\;\longrightarrow\;#2}

\newcommand{\nameIfStep}{IfStep}
\newcommand{\nameReadStep}{ReadStep}
\newcommand{\nameWriteStep}{WriteStep}
\newcommand{\nameTerminateStep}{TerminateStep}

\newcommand{\ruleReadStep}{\RULE\nameReadStep}

\newcommand{\ruleTerminateStep}{\RULE\nameTerminateStep}

\newcommand{\config}[2]{#1\;/\;#2}

\NewDocumentCommand{\interp}{o m}{%
  \llbracket\hspace{2pt}{#2}\hspace{1pt}
  \IfNoValueTF{#1}{\rrbracket}{\rrbracket_{#1}}
}

\newcommand{\compile}[1]{%
  \llparenthesis\hspace{2pt}#1\hspace{1pt}\rrparenthesis
}
\newcommand{\compilealt}[1]{%
\llparenthesis\hspace{2pt}#1\hspace{1pt}\rrparenthesis\textnormal{-alt}%
}

\newcommand{\mcompile}[1]{%
  \llparenthesis\hspace{2pt}#1\hspace{1pt}\rrparenthesis
}

\newcommand{\prog}{p}
\newcommand{\progA}{p}
\newcommand{\progB}{p'}

\newcommand{\thread}{t}

\newcommand{\pool}{P}
\newcommand{\poolA}{P}
\newcommand{\poolB}{P'}

\newcommand{\expr}{e}
\newcommand{\exprA}{e}
\newcommand{\exprB}{e'}

\newcommand{\exprOne}{e_1}
\newcommand{\exprTwo}{e_2}
\newcommand{\exprThree}{e_3}
\newcommand{\val}{v}

\newcommand{\loc}{\ensuremath{\ell}}
\newcommand{\locx}{x}
\newcommand{\locy}{y}
\newcommand{\locz}{z}

\newcommand{\reg}{r}
\newcommand{\regA}{a}
\newcommand{\regB}{b}
\newcommand{\regC}{c}
\newcommand{\regD}{d}
\newcommand{\cmd}{s}
\newcommand{\cmdA}{s}
\newcommand{\cmdB}{s'}
\newcommand{\md}{\mathit{md}}
\newcommand{\mdA}{\md_1}
\newcommand{\mdB}{\md_2}

\newcommand{\typeReg}{\mathit{Reg}}   %
\newcommand{\typeLoc}{\mathit{Loc}}   %

\newcommand{\typeProg}{\mathit{Prog}}    %
\newcommand{\typeExpr}{\mathit{Expr}}    %
\newcommand{\typeCmd}{\mathit{Cmd}}      %
\newcommand{\typeMode}{\mathit{Mode}}    %

\newcommand{\regst}{\kw{reg\_st}}

\newcommand{\evcounter}{\kw{ev\_counter}}
\newcommand{\nextcmd}{\kw{next\_cmd}}
\newcommand{\typeThread}{\mathit{Thread}}
\newcommand{\typeThreadDef}{
\left\{\begin{array}{@{}l@{}}
  \wtt\regst{\typeArrow\typeReg\NN};\\
  \wtt\evcounter\NN;\\
  \wtt\nextcmd\typeCmd
\end{array}\right\}
}

\newcommand{\typePool}{\mathit{Pool}}
\newcommand{\typePoolDef}{\typeFMap\NN\typeThread}

\newcommand{\typeArrow}[2]{{#1} \rightarrow {#2}}
\newcommand{\typeProd}[2]{{#1} \times {#2}}

\newcommand\fpfn{%
\smash{\raisebox{-.3ex}{%
\ensuremath{%
\xrightharpoonup{%
\smash{\raisebox{-.3ex}{%
\ensuremath{\scriptstyle\kern-0.25ex\,\mathit{fin}\,\kern-0.1ex}}}}
}}}}
\newcommand{\typeFMap}[2]{{#1}\;\fpfn\,{#2}}

\newcommand{\NN}{\mathbb{N}} %
\newcommand{\store}{\sigma} %

\newcommand{\finalSt}{\mathit{finalSt}}
\newcommand{\finalStTypical}{\finalSt\tuple\graphA{\mo}}

\newcommand{\typeLabel}[1]{
  \mathit{Lab}_{\,#1}
}

\newcommand{\EventE}{\kw{E}\xspace}
\NewDocumentCommand{\EventExtE}{o}{\IfNoValueTF{#1}{\EventE}{\EventE_{#1}}}
\NewDocumentCommand{\EventExtAtMost}{o m}{\EventExtE[#1]^{\sqsubseteq #2}}
\NewDocumentCommand{\EventExtAtLeast}{o m}{\EventExtE[#1]^{\sqsupseteq #2}}
\NewDocumentCommand{\EventExtMd}{o m}{\EventExtE[#1]^{#2}}

\newcommand{\EventAtLeast}[1]{\EventExtAtLeast[]{#1}}
\newcommand{\EventMd}[1]{\EventExtMd[]{#1}}

\NewDocumentCommand{\MemOpE}{m o}{\IfNoValueTF{#2}{#1}{#1_{#2}}}
\NewDocumentCommand{\MemOp}{m o o m m}{%
  \MemOpE{#1}[#2]\IfNoValueTF{#3}{}{^{#3}}(#4,\,#5)
}
\NewDocumentCommand{\MemOpAtMost}{m o m}{
  \MemOpE{#1}[#2]^{\sqsubseteq #3}
}
\NewDocumentCommand{\MemOpAtLeast}{m o m}{
  \MemOpE{#1}[#2]^{\sqsupseteq #3}
}
\newcommand{\ReadE}{\kw{R}}
\NewDocumentCommand{\ReadExtE}{o}{\MemOpE\ReadE[#1]}
\NewDocumentCommand{\ReadExt}{o o m m}{\MemOp\ReadE[#1][#2]{#3}{#4}}
\NewDocumentCommand{\ReadExtAtMost}{o m}{\MemOpAtMost\ReadE[#1]{#2}}
\NewDocumentCommand{\ReadExtAtLeast}{o m}{\MemOpAtLeast\ReadE[#1]{#2}}

\newcommand{\Read}[3]{\ReadExt[][#1]{#2}{#3}}

\newcommand{\ReadMd}[1]{\ReadE^{#1}}
\newcommand{\ReadNoMd}[2]{\ReadExt[][]{#1}{#2}}

\newcommand{\ReadAtLeast}[1]{\ReadExtAtLeast[]{#1}}
\newcommand{\WriteE}{\kw{W}}
\NewDocumentCommand{\WriteExtE}{o}{\MemOpE\WriteE[#1]}
\NewDocumentCommand{\WriteExt}{o o m m}{\MemOp\WriteE[#1][#2]{#3}{#4}}
\NewDocumentCommand{\WriteExtAtMost}{o m}{\MemOpAtMost\WriteE[#1]{#2}}
\NewDocumentCommand{\WriteExtAtLeast}{o m}{\MemOpAtLeast\WriteE[#1]{#2}}

\newcommand{\Write}[3]{\WriteExt[][#1]{#2}{#3}}

\newcommand{\WriteMd}[1]{\WriteE^{#1}}
\newcommand{\WriteNoMd}[2]{\WriteExt[][]{#1}{#2}}

\newcommand{\WriteAtLeast}[1]{\WriteExtAtLeast[]{#1}}
\newcommand{\NTWriteE}{\kw{NT}}
\NewDocumentCommand{\NTWriteExtE}{o}{\MemOpE\NTWriteE[#1]}
\newcommand{\NTWrite}[2]{\NTWriteE(#1,\, #2)}

\newcommand{\FenceE}{\kw{F}}
\NewDocumentCommand{\FenceExtE}{o}{\MemOpE\FenceE[#1]}
\newcommand{\Fence}[1]{\kw{F}^{#1}}

\newcommand{\FenceAtLeast}[1]{\Fence{\sqsupseteq #1}}
\newcommand{\FenceMd}[1]{\FenceE^{#1}}
\newcommand{\FenceNoMd}{\FenceE}
\newcommand{\MFenceE}{\kw{MF}}
\NewDocumentCommand{\MFenceExtE}{o}{\MemOpE\MFenceE[#1]}
\newcommand{\MFence}{\MFenceE}
\newcommand{\SFenceE}{\kw{SF}}
\NewDocumentCommand{\SFenceExtE}{o}{\MemOpE\SFenceE[#1]}
\newcommand{\SFence}{\SFenceE}
\newcommand{\RMWE}{\kw{RMW}}
\newcommand{\RMWSE}{\RMWE{\textnormal{-\texttt{s}}}}
\newcommand{\RMWFE}{\RMWE{\textnormal{-\texttt{f}}}}
\NewDocumentCommand{\RMWExtE}{o}{\MemOpE\RMWE[#1]}
\NewDocumentCommand{\RMWSExtE}{o}{\MemOpE\RMWSE[#1]}
\NewDocumentCommand{\RMWFExtE}{o}{\MemOpE\RMWFE[#1]}
\NewDocumentCommand{\RMWExt}{o o m m m}{%
  \RMWExtE[#1]\IfNoValueTF{#2}{}{^{#2}}(#3,\,#4,\,#5)
}
\NewDocumentCommand{\RMWExtMd}{o m}{\RMWExtE[#1]^{#2}}
\NewDocumentCommand{\RMW}{o m m m}{\RMWExt[][#1]{#2}{#3}{#4}}
\NewDocumentCommand{\RMWFail}{o m m}{\RMWExt[#1]{#2}{#3}{\bot}}
\newcommand{\RMWMd}[1]{\RMWExtMd[]{#1}}
\newcommand{\RMWNoMd}[3]{\RMWExt[][]{#1}{#2}{#3}}

\newcommand{\raceS}[1]{\kw{race}(#1)}
\newcommand{\race}[1]{\ensuremath{#1}-race\xspace}
\newcommand{\racy}[1]{\ensuremath{#1}-racy\xspace}
\newcommand{\racefree}[1]{\ensuremath{#1}-race free\xspace}

\newcommand{\dotLoc}[1]{#1.\texttt{loc}}

\newcommand{\dotMode}[1]{#1.\texttt{md}}

\newcommand{\refC}[1]{{#1}^{?}}
\newcommand{\transC}[1]{{#1}^{+}}
\newcommand{\reftransC}[1]{{#1}^{*}}

\newcommand{\inv}[1]{#1^{-1}}   %
\newcommand{\settorel}[1]{[#1]}
\newcommand{\seq}{;\,}
\newcommand{\union}[2]{#1\,\cup\,#2}
\newcommand{\disj}{\,\cup\,}
\newcommand{\coerce}[2]{#1 |_{#2}}
\newcommand{\imm}[1]{\coerce{#1}{\textnormal{\textsf{imm}}}}

\newcommand{\perloc}[1]{\coerce{#1}{\textnormal{\textsf{loc}}}}

\newcommand{\neqloc}[1]{\coerce{#1}{\neq\textnormal{\textsf{loc}}}}
\newcommand{\atloc}[2]{{#1}_{#2}}
\newcommand{\internal}[1]{{#1}_{i}}
\newcommand{\external}[1]{{#1}_{e}}
\newcommand{\onsc}[1]{\coerce{#1}{\sc}}

 \newcommand{\makerel}[2]{{\color{#1}\sf{#2}}\xspace}

\newcommand{\REventMd}[1]{\settorel{\EventMd{#1}}}

\newcommand{\colorPO}{black}
\newcommand{\po}{\makerel\colorPO{po}}
\newcommand{\colorPPO}{black}
\newcommand{\ppo}{\makerel\colorPPO{ppo}}

\newcommand{\ppointel}{\intelext.\ppo}

\newcommand{\porceleven}{\makerel\colorPPO{po}_{\textnormal{\rc}}}
\newcommand{\ppoasm}{\makerel\colorPPO{ppo}_{\textnormal{\textsf{asm}}}}

\newcommand{\colorOB}{blue!80}
\newcommand{\ob}{\makerel\colorOB{ob}}

\newcommand{\colorRF}{green!60!black}
\newcommand{\rf}{\makerel\colorRF{rf}}
\newcommand{\rfe}{\external\rf}
\newcommand{\rfi}{\internal\rf}
\newcommand{\colorRMW}{olive!70!black}
\newcommand{\rmw}{\makerel\colorRMW{rmw}}
\newcommand{\colorMO}{orange}
\newcommand{\mo}{\makerel\colorMO{mo}}
\newcommand{\moe}{\external\mo}
\newcommand{\moi}{\internal\mo}
\newcommand{\colorRB}{red!70!black}
\newcommand{\rb}{\makerel\colorRB{rb}}
\newcommand{\rbe}{\external\rb}
\newcommand{\rbi}{\internal\rb}
\newcommand{\colorPB}{red!40!black}
\newcommand{\pb}{\makerel\colorPB{pb}}

\newcommand{\colorHB}{blue}
\newcommand{\hb}{\makerel\colorHB{hb}}

\newcommand{\colorSW}{green!30!black}
\newcommand{\sw}{\makerel\colorSW{sw}}

\newcommand{\colorECO}{red!90!black}
\newcommand{\eco}{\makerel\colorECO{eco}}
\newcommand{\colorPSC}{violet}
\newcommand{\psc}{\makerel\colorPSC{psc}}
\newcommand{\pscBase}{\psc_{\textnormal{\textsf{base}}}}
\newcommand{\pscFence}{\psc_{\textnormal{\textsf{fence}}}}
\newcommand{\colorSCB}{violet}
\newcommand{\scb}{\makerel\colorSCB{scb}}

\NewDocumentCommand{\cyc}{s m}{
\mathit{cyclic}%
\IfBooleanTF{#1}{\left(\begin{array}{@{}l@{}}#2\end{array}\right)}{(#2)}
}
\newcommand{\acyc}[1]{\mathit{acyclic}{(#1)}}
\NewDocumentCommand{\notirr}{s m}{
\neg\mathit{irreflexive}%
\IfBooleanTF{#1}{\left(\begin{array}{@{}l@{}}#2\end{array}\right)}{(#2)}
}
\newcommand{\irr}[1]{\mathit{irreflexive}{(#1)}}

\newcommand{\nameSCCoherence}{SC-Coherence}
\newcommand{\axiomSCCoherence}{\RULE\nameSCCoherence}
\newcommand{\axiomSCCoherenceDef}{$\acyc{\po\disj\rf\disj\mo\disj\rb}$}

\newcommand{\nameCoherence}{Coherence}
\newcommand{\axiomCoherence}{\RULE\nameCoherence}
\newcommand{\axiomCoherenceDef}{$\irr{\hb\seq\refC\eco}$}
\newcommand{\nameSC}{SC}
\newcommand{\axiomSC}{\RULE\nameSC}
\newcommand{\axiomSCDef}{$\acyc\psc$}
\newcommand{\nameAtomicity}{Atomicity}
\newcommand{\axiomAtomicity}{\RULE\nameAtomicity}
\newcommand{\axiomAtomicityDef}{$\irr{\rb\seq\mo}$}
\newcommand{\nameNTA}{No-Thin-Air}
\newcommand{\axiomNTA}{\RULE\nameNTA}
\newcommand{\axiomNTADef}{$\acyc{\union\po\rf}$}

\newcommand{\nameCoherenceI}{Coherence-I}
\newcommand{\axiomCoherenceI}{\RULE\nameCoherenceI}

\newcommand{\nameCoherenceII}{Coherence-II}
\newcommand{\axiomCoherenceII}{\RULE\nameCoherenceII}
\newcommand{\axiomCoherenceIIDef}{$\acyc{\ppoasm\disj\eco}$}
\newcommand{\nameCoherenceIII}{Coherence-III}
\newcommand{\axiomCoherenceIII}{\RULE\nameCoherenceIII}
\newcommand{\axiomCoherenceIIIDef}{$\irr{\settorel{\WriteMd\nt}\seq\po\seq(\rb\disj\mo)}$}

\newcommand{\nameInternal}{Internal}
\newcommand{\axiomInternal}{\RULE\nameInternal}
\newcommand{\axiomInternalDef}{$\irr{\po\seq(\rfi\disj\moi\disj\rbi)}$}
\newcommand{\nameExternal}{External}
\newcommand{\axiomExternal}{\RULE\nameExternal}
\newcommand{\relExternalDef}{%
  \union\ppo{\union\rfe{\union\moe\rbe}}
}

\newcommand{\initEv}[1]{\mathit{I}(#1)}
\newcommand{\initGraph}{\mathit{Init}}

\newcommand{\gsim}{\sim}

\newcommand{\lab}{\kw{lab}\xspace}

\newcommand{\graph}{G}
\newcommand{\graphA}{G}
\newcommand{\graphB}{G'}

\newcommand{\evA}{a}
\newcommand{\evB}{b}
\newcommand{\evC}{c} %
\newcommand{\evD}{d}
\newcommand{\evE}{e}
\newcommand{\cycle}{C}

\tikzstyle{mixed}=[
  rounded corners,
  draw,
  inner sep=5pt,
  thick,
  rectangle split,
  rectangle split horizontal,
  rectangle split parts=2,
  rectangle split part fill={white,black!03}
]

\tikzstyle{internal-node}=[
  inner sep=0pt,
  text opacity=1,
  fill opacity=0
]

\tikzstyle{arrow}=[->,shorten >= 1pt,>={Stealth[round]},semithick]
\tikzstyle{gr-edge}=[arrow,shorten >=2pt,shorten <=2pt,thick]
\tikzstyle{po-edge}=[\colorPO,gr-edge]
\tikzstyle{rf-edge}=[\colorRF,gr-edge,densely dashed]
\tikzstyle{sw-edge}=[\colorSW,gr-edge,densely dashed]
\tikzstyle{hb-edge}=[\colorHB,gr-edge,densely dashed]
\tikzstyle{rb-edge}=[\colorRB,gr-edge,densely dashed]
\tikzstyle{rmw-edge}=[\colorRMW,gr-edge]
\tikzstyle{mo-edge}=[\colorMO,gr-edge,densely dotted]

\newcommand{\mnode}[4][]{
  \node(#2)[mixed,#1]{{#3}\nodepart{second}{#4}};
}
\newcommand{\nameRR}{R-R}
\newcommand{\nodeRR}{\RULE\nameRR}
\NewDocumentCommand{\mnodeRR}{O{} O{} o O{\md} O{\loc} O{\intn}}{
  \mnode[#1]{#2}
    {$\Read{#4}{#5}{#6}$}
    {$\ReadNoMd{#5}{#6}$}
}
\newcommand{\nameWNTW}{W-NT}
\newcommand{\nodeWNTW}{\RULE\nameWNTW}
\NewDocumentCommand{\mnodeWNTW}{O{} O{} O{\nt} O{\loc} O{\intn}}{
  \mnode[#1]{#2}
    {$\Write{#3}{#4}{#5}$}
    {$\NTWrite{#4}{#5}$}
}
\newcommand{\nameWNTWalt}{W-NT-alt}
\newcommand{\nodeWNTWalt}{\RULE\nameWNTWalt}
\newcommand{\nameWSFW}{W-SFW}
\newcommand{\nodeWSFW}{\RULE\nameWSFW}
\NewDocumentCommand{\mnodeWSFW}{O{} O{} O{\loc} O{\intn}}{
  \mnode[#1]{#2}{$\Write\rel{#3}{#4}$}{\tikz{
    \node(A)[internal-node]{\;\;$\SFence$};
    \node(B)[internal-node,below of=A,yshift=4pt]{$\WriteNoMd{#3}{#4}$};
    \draw[thick,shorten >=2pt,shorten <=2pt] (A) -- (B.north);
  }}
}
\newcommand{\nameWW}{W-W}
\newcommand{\nodeWW}{\RULE\nameWW}
\NewDocumentCommand{\mnodeWW}{O{} O{} o O{\neq\sc,\nt} O{\loc} O{\intn}}{
  \mnode[#1]{#2}
    {$\Write{#4}{#5}{#6}$}
    {$\WriteNoMd{#5}{#6}$}
}

\NewDocumentCommand{\mnodeRRMW}
{O{} O{} O{\neq\na} O{\loc} O{\intn}}
{
  \mnode[#1]{#2}
    {$\Read{#3}{#4}{#5}$}
    {$\RMWNoMd{#4}{#5}\bot$}
}

\NewDocumentCommand{\mnodeRWRMW}
{O{} O{}
 O{\fst{\mcompile\md}} O{\snd{\mcompile\md}}
 O{\loc} O{\intn} O{\intm}}
{
  \mnode[#1]{#2}{\tikz{
    \node(A)[internal-node]{$\Read{#3}{#5}{#6}$};
    \node(B)[internal-node,below of=A,yshift=4pt]{$\Write{#4}{#5}{#7}$};
    \draw[thick,shorten <=2pt,shorten >=0pt] (A) -- (B.north);
  }}
  {$\RMWNoMd{#5}{#6}{#7}$}
}

\NewDocumentCommand{\mnodeRMWRMW}
{O{} O{} o O{\md} O{\loc} O{\intn} O{\intm^{?}}}
{
  \mnode[#1]{#2}
    {$\RMW[#4]{#5}{#6}{#7}$}
    {$\RMWNoMd{#5}{#6}{#7}$}
}
\newcommand{\nameRMWRMWS}{RMW-RMW-S}
\newcommand{\nodeRMWRMWS}{\RULE\nameRMWRMWS}
\NewDocumentCommand{\mnodeRMWRMWS}
{O{} O{} o O{\md} O{\loc} O{\intn} O{\intm}}
{
  \mnode[#1]{#2}
    {$\RMW[#4]{#5}{#6}{#7}$}
    {$\RMWNoMd{#5}{#6}{#7}$}
}
\newcommand{\nameRMWRMWF}{RMW-RMW-F}
\newcommand{\nodeRMWRMWF}{\RULE\nameRMWRMWF}
\NewDocumentCommand{\mnodeRMWRMWF}
{O{} O{} o O{\md} O{\loc} O{\intn}}
{
  \mnode[#1]{#2}
    {$\RMW[#4]{#5}{#6}{\bot}$}
    {$\RMWNoMd{#5}{#6}{\bot}$}
}

\NewDocumentCommand{\mnodeWSFWF}{O{} O{} O{\loc} O{\intn}}{
  \mnode[#1]{#2}{$\Write\sc{#3}{#4}$}{\tikz{
    \node(A)[internal-node]{\;\;$\SFence$};
    \node(B)[internal-node,below of=A,yshift=4pt]{$\WriteNoMd{#3}{#4}$};
    \node(C)[internal-node,below of=B,yshift=4pt]{\;\;$\FenceNoMd$};
    \draw[thick,shorten >=2pt,shorten <=2pt] (A) -- (B.north);
    \draw[thick,shorten >=2pt,shorten <=2pt] (B) -- (C.north);
  }}
}
\newcommand{\nameWSFWMF}{W-SFWMF}
\newcommand{\nodeWSFWMF}{\RULE\nameWSFWMF}
\NewDocumentCommand{\mnodeWSFWMF}{O{} O{} O{\loc} O{\intn}}{
  \mnode[#1]{#2}{$\Write\sc{#3}{#4}$}{\tikz{
    \node(A)[internal-node]{\;\;$\SFence$};
    \node(B)[internal-node,below of=A,yshift=4pt]{$\WriteNoMd{#3}{#4}$};
    \node(C)[internal-node,below of=B,yshift=4pt]{\;\;$\MFence$};
    \draw[thick,shorten >=2pt,shorten <=2pt] (A) -- (B.north);
    \draw[thick,shorten >=2pt,shorten <=2pt] (B) -- (C.north);
  }}
}

\NewDocumentCommand{\mnodeWWF}{O{} O{} O{\loc} O{\intn}}{
  \mnode[#1]{#2}{$\Write\sc{#3}{#4}$}{\tikz{
    \node(A)[internal-node]{$\WriteNoMd{#3}{#4}$};
    \node(B)[internal-node,below of=A,yshift=4pt]{\;\;$\FenceNoMd$};
    \draw[thick,shorten >=2pt,shorten <=2pt] (A) -- (B.north);
  }}
}
\newcommand{\nameWWMF}{W-WMF}
\newcommand{\nodeWWMF}{\RULE\nameWWMF}
\NewDocumentCommand{\mnodeWWMF}{O{} O{} O{\loc} O{\intn}}{
  \mnode[#1]{#2}{$\Write\sc{#3}{#4}$}{\tikz{
    \node(A)[internal-node]{$\WriteNoMd{#3}{#4}$};
    \node(B)[internal-node,below of=A,yshift=4pt]{\;\;$\MFence$};
    \draw[thick,shorten >=2pt,shorten <=2pt] (A) -- (B.north);
  }}
}

\NewDocumentCommand{\mnodeFF}{O{} O{}}{
  \mnode[#1]{#2}{$\Fence\sc$}{$\FenceNoMd$}
}
\newcommand{\nameFMF}{F-MF}
\newcommand{\nodeFMF}{\RULE\nameFMF}
\NewDocumentCommand{\mnodeFMF}{O{} O{}}{
  \mnode[#1]{#2}{$\Fence\sc$}{$\MFence$}
}
\newcommand{\nameFB}{F-$\bot$}
\newcommand{\nodeFB}{\RULE\nameFB}
\NewDocumentCommand{\mnodeFB}{O{} O{} O{\sqsubseteq\acqrel}}{
  \mnode[#1]{#2}{$\Fence{#3}$}{$\bot$}
}

\NewDocumentCommand{\mnodeMFMF}{O{} O{}}{
  \mnode[#1]{#2}{$\MFence$}{$\MFence$};
}
\newcommand{\nameFSF}{F-SF}
\newcommand{\nodeFSF}{\RULE\nameFSF}
\NewDocumentCommand{\mnodeFSF}{O{} O{} O{\stf}}{
  \mnode[#1]{#2}{$\FenceMd{#3}$}{$\SFence$};
}
\newcommand{\nameFSFalt}{F-SF-alt}
\newcommand{\nodeFSFalt}{\RULE\nameFSFalt}

\newcommand{\lproj}[1]{\downharpoonleft\hspace{-3pt}{#1}}
\newcommand{\rproj}[1]{{#1}\hspace{-3pt}\downharpoonright}

\newcommand{\graphM}{G_m}

\NewDocumentCommand{\mnodeRDMBR}{O{} O{} O{\loc} O{\intn}}{
  \mnode[#1]{#2}{$\Read\sc{#3}{#4}$}{\tikz{
    \node(A)[internal-node]{\;\;$\DMB\sy$};
    \node(B)[internal-node,below of=A,yshift=4pt]{$\Read\acq{#3}{#4}$};
    \draw[thick,shorten >=2pt,shorten <=2pt] (A) -- (B.north);
  }}
}

\NewDocumentCommand{\mnodeRRacq}{O{} O{} O{\loc} O{\intn}}{
  \mnode[#1]{#2}
    {$\Read\acq{#3}{#4}$}
    {$\Read\acq{#3}{#4}$}
}

\NewDocumentCommand{\mnodeRRrlx}{O{} O{} O{\loc} O{\intn}}{
  \mnode[#1]{#2}
    {$\Read\rlx{#3}{#4}$}
    {$\Read\rlx{#3}{#4}$}
}

\NewDocumentCommand{\mnodeRRna}{O{} O{} O{\loc} O{\intn}}{
  \mnode[#1]{#2}
    {$\Read\na{#3}{#4}$}
    {$\Read\rlx{#3}{#4}$}
}

\NewDocumentCommand{\mnodeWDMBW}{O{} O{} O{\loc} O{\intn}}{
  \mnode[#1]{#2}{$\Write\sc{#3}{#4}$}{\tikz{
    \node(A)[internal-node]{\;\;$\DMB\sy$};
    \node(B)[internal-node,below of=A,yshift=4pt]{$\Write{\rel}{#3}{#4}$};
    \draw[thick,shorten >=2pt,shorten <=2pt] (A) -- (B.north);
  }}
}

\NewDocumentCommand{\mnodeWWrel}{O{} O{} O{\loc} O{\intn}}{
  \mnode[#1]{#2}
    {$\Write\rel{#3}{#4}$}
    {$\Write\rel{#3}{#4}$}
}

\NewDocumentCommand{\mnodeWWrlx}{O{} O{} O{\loc} O{\intn}}{
  \mnode[#1]{#2}
    {$\Write{\sqsupseteq\rlx}{#3}{#4}$}
    {$\Write\rlx{#3}{#4}$}
}

\NewDocumentCommand{\mnodeFDMBacq}{O{} O{}}{
  \mnode[#1]{#2}{$\Fence\acq$}{$\DMB\ld$}
}

\NewDocumentCommand{\mnodeFDMB}{O{} O{} O{\neq\acq}}{
  \mnode[#1]{#2}{$\Fence{#3}$}{$\DMB\sy$}
}

\NewDocumentCommand{\poolLkp}{O{=} O{\poolA} O{\inti}}{
  #2[#3].
  \poolFields[#1]
}

\NewDocumentCommand{\poolFields}{O{=}}{
  \def\tmpEq{#1\,}
  \left\{\begin{array}{@{}l@{\;}l@{}}
  \regstLkp
}
\NewDocumentCommand{\regstLkp}{t> s O{\phi} t<}{
  \IfBooleanTF{#2}{\regst&\tmpEq#3\IfBooleanTF{#4}{}{\arrcr}}{}
  \evcounterLkp
}
\NewDocumentCommand{\evcounterLkp}{t> s O{\intj} t<}{
  \IfBooleanTF{#2}{\evcounter&\tmpEq#3\IfBooleanTF{#4}{}{\arrcr}}{}
  \nextcmdLkp
}
\NewDocumentCommand{\nextcmdLkp}{t> s O{\cmd} t<}{
  \IfBooleanTF{#2}{\nextcmd&\tmpEq#3\IfBooleanTF{#4}{}{\arrcr}}{}
  \end{array}\right.
}

\NewDocumentCommand{\graphUpt}{O{\graphA}}{
  \def\tmpGraph{#1}
  \tmpGraph.\left\{\begin{array}{@{}l@{\;}l@{}}
  \eventLkp
}
\NewDocumentCommand{\eventLkp}{s m t<}{
  \IfBooleanTF{#1}{
    \EventE&:=\,\tmpGraph.\EventE\uplus#2\IfBooleanTF{#3}{}\arrcr
  }{}
  \labLkp
}
\NewDocumentCommand{\labLkp}{s m t<}{
  \IfBooleanTF{#1}{\lab&:=\,\tmpGraph.\lab#2\IfBooleanTF{#3}{}{\arrcr}}{}
  \rmwLkp
}
\NewDocumentCommand{\rmwLkp}{s m t<}{
  \IfBooleanTF{#1}
    {\rmw&:=\,\tmpGraph.\rmw\cup#2\IfBooleanTF{#3}{}{\arrcr}}
    {}
  \end{array}\right.
}

\newcommand{\outcomeC}[1]{\color{chameleon3}{\,\textit{//}\;#1}} %
\newcommand{\outcomeI}[1]{\color{red2}{\,\textit{//}\;\underline{#1}}} %

\newcommand{\inarr}[1]{\begin{array}{@{}l@{}}#1\end{array}}
\newcommand{\inarrII}[2]{
  \begin{array}{@{}l@{~~}||@{~~}l@{}}
    \inarr{#1}&\inarr{#2}
  \end{array}
}
\newcommand{\inarrIII}[3]{
  \begin{array}{@{}l@{~~}||@{~~}l@{~~}||@{~~}l@{}}
    \inarr{#1}&\inarr{#2}&\inarr{#3}
  \end{array}
}
\usepackage{pifont}

\newcommand{\cmark}{\faStar\xspace}
\newcommand{\cmarkcell}{\cellcolor{black!04}\cmark}
\newcommand{\xmark}{\faStarO\xspace}
\newcommand{\xmarkcell}{\cellcolor{black!04}\xmark}
\newcommand{\hmark}{\faStarHalfO\xspace}
\newcommand{\hmarkcell}{\cellcolor{black!04}\hmark}

\let\oldparagraph\paragraph
\renewcommand{\paragraph}[1]{\oldparagraph{\textnormal{\textbf{#1}}}}

\Crefname{figure}{Figure}{Figures}
\crefname{figure}{Figure}{Figures}

\crefname{section}{\S{}}{sections}
\crefformat{section}{#2\S{}#1#3}
\Crefformat{section}{Section #2#1#3}

\crefname{definition}{Definition}{Definitions}

\crefname{theorem}{Theorem}{Theorems}

\crefname{appendix}{\S{}}{appendices}
\crefformat{appendix}{#2\S{}#1#3}
\Crefformat{appendix}{Appendix #2#1#3}

\newcommand{\transf}[2]{#1 \rightsquigarrow #2}

%% file: abstract.tex
\begin{abstract}
Programs written in C/C\texttt{++} often include \textit{inline assembly}: a snippet of architecture-specific assembly code used to access low-level functionalities that are impossible or expensive to simulate in the source language. Although inline assembly is widely used, its semantics has not yet been formally studied.

In this paper, we overcome this deficiency by investigating the effect of inline assembly on the \textit{consistency} semantics of~C/C\texttt{++} programs. We propose the first memory model of the~C\texttt{++} Programming Language with support for inline assembly for Intel's \textsf{x86} including \textit{non-temporal stores} and \textit{store fences}. We argue that previous provably correct compiler optimizations and correct compiler mappings should remain correct under such an extended model and we prove that this requirement is met by our proposed model.
\end{abstract}

%% file: intro.tex
\section{Introduction}
\label{section:intro}

Large software applications are rarely written in only one language.
While the bulk of an application is typically written
in a general-purpose programming language, such as \cpp,
some parts are invariably written in higher-level domain-specific languages
(for example, lexers and parsers, which generate C++ code)
and others directly in assembly code of the underlying architecture(s).

The latter kind is directly supported by mainstream \ccpp compilers
through \emph{inline-assembly} blocks,
which can be used
(1) to expose some hardware instructions that
are inaccessible or difficult to simulate in the source language,
(2) to write prologue and epilogue code of \emph{naked} functions~%
\cite{cpp-assembler-msvc},
and
(3) to keep the ordering of instructions at compile time~%
\cite{preshing-memory-ordering-at-compile-time}.
As such, inline assembly constitutes an important tool of~\ccpp,
whose significance is further attested by major projects, such as
the Linux kernel-based virtual machine (KVM)~%
\cite{linux-kvm}
and the GNU Compiler Collection (GCC)~%
\cite{gcc},
each counting with thousands of occurrences of
inline assembly.

Unlike some of the key features of \ccpp,
such as synchronization primitives,
which have been the subject of many research papers~\cite{c11mm,rc11},
and despite the extensive use of inline assembly,
inline assembly lacks a \emph{formal semantics}:
a precise unambiguous specification.

In this paper, we overcome this deficiency and
propose the first formal account of inline assembly.
We distinguish three classes of inline-assembly instructions:
\begin{enumerate}
\item Instructions,
such as complex arithmetic and bit-manipulating operations
and \emph{single instruction/multiple data}~\cite{flynn-72}
(SIMD) instructions,
whose effect can be expressed in the source language
(typically, as a sequence of arithmetic operations).
\item Instructions accessing memory and/or enforcing ordering
between instructions (such as \emph{store fences}~%
\cite[Vol. 2B, \S4]{intel-manual}),
whose effect cannot be expressed in the source language.
Such instructions are commonly used in libraries for parallel and persistent
programming, efficient moving of data, and communicating with external
devices.
\item Instructions that have a global effect and
may completely change the semantics of the subsequent program,
such as
raising an interrupt,
writing to the stack pointer register or to the page table entries~%
\cite{simner-al-22,alglave-24},
and flushing the \emph{translation lookaside buffer}~%
\cite[Vol. 2A, \S3]{intel-manual}.
\end{enumerate}

We narrow our scope to the second class of instructions for the Intel's \intel
architecture, whose consistency and persistency semantics have been
formalized by \citet{raad-22} in a model known as~\intelext.
We argue that supporting the first class of instructions is straightforward,
raising no challenges beyond that of providing accurate semantics for the
individual hardware instructions.
In contrast, the second class affects the
\textit{memory consistency model} of the programming language,
governing how concurrent programs are allowed to interact through
shared memory.
As we shall see,
the effect of this class of instructions on the language's model
leads to interesting semantic challenges.
As for the third class of instructions,
we declare them to be beyond the scope of this paper.

A particularly interesting use case of inline assembly are \intel
\emph{non-temporal stores}~\cite[Vol. 1, \S10.4.6.2]{intel-manual},
an \intel-specific feature that 
allows writing to memory while bypassing the cache.
Non-temporal stores are used in cases of bulk memory writes~\cite{raad-22},
whose relative order is immaterial, such as initializing a memory page with zeros.

Unlike regular \intel stores, non-temporal stores can be reordered with other
stores, and so the following \ccpp program with inline assembly, when compiled with
\texttt{gcc}~\cite{gcc} or
\texttt{clang}~\cite{clang},
can exhibit the following quite surprising
outcome (here and henceforth, we use pseudocode syntax with $x,y,\ldots$ being shared locations
and $a,b,\ldots$ being thread-local registers; we assume that all locations are initialized to $0$):
\begin{equation*}
\tag{MP-NT}
\label{prog:mp-nt}
\inarrII{
  \AsmNTWritePL\locx{1} \\
  \WritePL\rel\locy{1}
}{
  \ReadPL\acq\regA\locy\outcomeC{1} \\
  \ReadPL\rlx\regB\locx\outcomeC{0}
}
\quad\xrightarrow{\textit{compile}}\quad
\inarrII{
  \texttt{movnt}\ [\locx], 1 \\
  \texttt{mov}\ [\locy], 1
}{
  \texttt{mov}\ \regA, [\locy]\outcomeC{1} \\
  \texttt{mov}\ \regB, [\locx]\outcomeC{0}
}
\end{equation*}

Normally, \ccpp release-acquire accesses induce synchronization and thus
anything executed before a release write is deemed to have happened before
everything after an acquire read fulfilled by this write.
Yet, this is no longer the case with inline assembly.
Applying the standard
compilation scheme of mapping \ccpp release/acquire/relaxed accesses to regular
\intel accesses
results in a \intel program that can read $a=1 \land b=0$.
The only way to prevent the weak outcome is to add
an appropriate instruction working as a fence between
the two store instructions:
a store fence (\SFencePL) suffices,
but one may also use a \emph{memory fence} (\texttt{mfence}), 
a read-modify-write operation,
or a plain \intel store to~$\locx$.
However, without a formal specification,
such observations are unclear to developers,
who naturally expect release/acquire synchronization
to apply to all kinds of accesses.
\footnote{%
Indeed, Program~\ref{prog:mp-nt} illustrates
one of the concerns in a recent
Rust bug report: \url{https://github.com/rust-lang/rust/issues/114582}.
}
 
The question is how to provide an appropriate semantics for \ccpp programs with
inline assembly, such as the previous example of~\ref{prog:mp-nt}.
In \cref{section:overview}, we show that devising an appropriate semantics is by no
means trivial.
At the very least, one would require a solution that is:
\begin{itemize}
\item \emph{flexible},
that is, allowing arbitrary mixing of \ccpp and
inline-assembly accesses
with no partition on threads or memory locations
that can or cannot use \intel instructions,
since such restriction is not respected
by most use cases of inline assembly;
\item \emph{supporting a representative set of \intel and \ccpp features}
that have to do with accessing memory in a possibly concurrent setting;
\item \emph{preserving} the correctness of the existing
\emph{\ccpp compilation schemes}
to \intel and of \emph{local source-to-source code transformations}, 
since these are readily performed by \ccpp compilers;
\item \emph{precisely matching the \intel (resp.\ \ccpp) model} for programs
consisting purely of \intel (resp.\ \ccpp) constructs.
This last criterion acts as
a sanity check ensuring that the semantics of existing \ccpp programs
(without inline assembly)
will not be affected by our proposed extension of the \ccpp
concurrency model.
\end{itemize}

In addition, we would like our semantics to provide useful guarantees
for common correct uses of inline assembly,
such as the following variant of \ref{prog:mp-nt},
which rules out the weak outcome by
inserting a store fence
between the non-temporal store to~$\locx$ and the release write to~$\locy$:
\begin{equation*}
\tag{MP-NT-SF}
\label{prog:mp-nt-sf}
\inarrII{
  \AsmNTWritePL\locx{1} \\
  \AsmSFencePL \\
  \WritePL\rel\locy{1}
}{
  \ReadPL\acq\regA\locy\outcomeI{1} \\
  \ReadPL\rlx\regB\locx\outcomeI{0}
}
\quad\xrightarrow{\textit{compile}}\quad
\inarrII{
  \texttt{movnt}\ [\locx], 1 \\
  \texttt{sfence} \\
  \texttt{mov}\ [\locy], 1
}{
  \texttt{mov}\ \regA, [\locy]\outcomeI{1} \\
  \texttt{mov}\ \regB, [\locx]\outcomeI{0}
}
\end{equation*}
(In our examples,
certain read instructions are followed by comments.
When every comment is displayed in green, as~$\outcomeC{v}$,
the annotated outcome
can be observed on some architecture and should
therefore be allowed by the model.
When every comment is underlined and displayed in red, as~$\outcomeI{v}$,
the annotated outcome
cannot be observed and should
therefore be forbidden.)

As we explain in \cref{section:overview}, many direct approaches
to the problem of defining an appropriate semantics for
\ccpp with inline \intelext assembly
fail one or more of the stated requirements.

In response,
in \cref{section:model},
we develop a carefully designed extension of
the \ccpp consistency model with support for the user-mode
\intelext inline-assembly instructions that access memory:
namely,
plain loads and stores,
non-temporal stores,
read-modify-write operations, and
fences.
We prove that our model is an extension of the \intelext and \ccpp models,
in the sense that plain \intel and plain \ccpp programs have
unchanged semantics.

In \cref{section:metatheory}, we prove that the established sound compilation
schemes from \cpp to \intelext remain sound in spite of the presence of
inline-assembly blocks, and that, similarly,
so do the sound local source-to-source code transformations,
such as reordering of independent memory loads.
In addition,
we introduce a new, provably sound, compilation scheme to \intelext,
which compiles relaxed writes to non-temporal stores
for the price of including some additional store fences~%
(\defref{def:alt:scheme:rcelevenext}).

%% file: overview.tex
\section{Overview}
\label{section:overview}

In this section, we provide a gentle introduction to \cref{section:model},
where we formalize our contributions.
To this end,
in~\cref{subsection:desiderata},
we establish a series of desired properties that a model
for \ccpp with inline assembly should enjoy.
Then,
in~\cref{subsection:evaluation},
we show why direct approaches for devising such a model do not work.
Finally,
in \cref{subsection:overview-towards} and
\cref{subsection:intuitive:account},
we present an intuitive overview of our proposed model,
showing how it satisfies the established desiderata.

\subsection{Desiderata for a Hybrid Consistency Model for \ccpp and \intel Assembly}
\label{subsection:desiderata}

We argue that
tentative ``hybrid models'' for \ccpp with support for
inline \intelext assembly
should enjoy the following properties:

\hypertarget{P0}{}
\paragraph{P0: Flexibility.}
As a first minimal requirement, we ask the hybrid model to support all the
features of the respective \ccpp and \intel models, and to allow free mixing
of the two.
That is, we want to be able to write programs where
threads can mix both \ccpp and inline-assembly instructions
and where
memory locations can be accessed using
both types of instructions,
as we have seen in the~\ref{prog:mp-nt} and~%
\ref{prog:mp-nt-sf} programs.

\hypertarget{P1}{}
\paragraph{P1: Correctness of compiler mappings.}
In the weak-memory literature,
a \textit{compiler mapping}, or a \textit{compilation scheme},
maps the memory operations of the source language to
sequences of instructions of the target language that implement the
corresponding high-level memory operation.
Two standard compilation schemes from \ccpp to \intel exist~\cite{c11mm,rc11}:
the \textit{fence-after-\sc-write} scheme,
which places memory fences after \sc writes; and
the \textit{fence-before-\sc-read} scheme,
which places memory fences before \sc reads.
Both schemes have been proven \textit{correct} with respect to \rceleven~\cite{rc11}:
the compilation of a \ccpp program~$\prog$ following one of these schemes
can only exhibit behaviors that are assigned to~$\prog$ by \rceleven.
These schemes can be easily extended with support for
inline \intelext assembly by simply mapping
an inline-assembly
instruction~$\Asm[\cmd]$ to~$\cmd$.
This mapping is in agreement with how current
\ccpp compilers handle such instructions~\cite[Chapter 6.6]{compcert}.
It is therefore desirable that
these schemes remain correct with respect to
a hybrid model for \ccpp with inline \intelext assembly.

\hypertarget{P2}{}
\paragraph{P2: Correctness of standard compiler optimizations.}
To improve program performance, \ccpp compilers perform a sequence of local
source-to-source transformations, whose correctness (in the absence of inline
assembly) has been established by prior work \cite{vafeiadis-et-al-15,rc11}.
\ccpp compilers readily perform these transformation even when the program
contains inline assembly.
It is thus important that these transformations remain correct in any \ccpp model
extended with inline assembly.

\hypertarget{P3}{}
\paragraph{P3: Extension of source.}
For programs that do not use inline assembly, we want our model
to coincide with the model of the source language.
Concretely, we consider \rceleven as the source model,
and we say that a model~$M$ is an
\textit{extension of~\rceleven}
if the semantics given by~$M$ to plain \ccpp programs
agrees with the semantics given by~\rceleven.
If this property did not hold of a candidate
hybrid model~$M$,
then
plain \ccpp and \ccpp with support for inline assembly
should be seen as different programming languages,
because programs could
have different semantics depending on whether~%
\rceleven or the hybrid model~$M$ is used.
We see this distinction as artificial
and compromising to the language.

\hypertarget{P4}{}
\paragraph{P4: Extension of target.}
Analogously, we argue that a candidate hybrid model~$M$
should be an
\textit{extension of~\intelext}:
the semantics given by~$M$
to a \ccpp program~$\prog$
written entirely using inline~\intelext assembly
should agree with the semantics given by~\intelext
(to the obvious~\intelext program corresponding to~$\prog$).
The model~$M$ cannot give a stronger
semantics to~$\prog$ than~$\intelext$
because the compilation scheme
of inline assembly is the straightforward identity map.
Therefore,
if there was a mismatch,
then the model~$M$ would be necessarily
assigning a more relaxed semantics to~$\prog$
than~$\intelext$.
This weakness in reasoning is undesirable.

\hypertarget{P5}{}
\paragraph{P5: Architecture-specific guarantees for mixed programs.}
The \rceleven model is sufficiently relaxed so as to
support efficient compilation to multiple hardware architectures.
This generality has the downside that \rceleven
may allow behaviors that cannot be observed by most implementations.
The following program, for example, depicts such a behavior
(known as \textit{independent reads from independent writes} - IRIW):
\begin{equation*}
\tag{IRIW}
\label{prog:iriw}
\inarrII{\inarrII{
  \WritePL\rel\locx{1}
}{
  \ReadPL\acq\regA\locx\outcomeC{1} \\
  \ReadPL\rlx\regB\locy\outcomeC{0}
}}{\inarrII{
  \ReadPL\acq\regC\locy\outcomeC{1} \\
  \ReadPL\rlx\regD\locx\outcomeC{0}
}{
  \WritePL\rel\locy{1}
}}
\end{equation*}

This behavior is allowed by \rceleven and observed when the program is
run on the \power~\cite{alglave-14} architecture.  It illustrates that the two independent
writes in the first and fourth threads can be observed in different orders
by the second and third threads,
even though the accesses in these two middle threads
have to be executed in order
(the \acq access mode prevents reordering with subsequent accesses).

When, however, the \ref{prog:iriw} program is compiled to \intel and to recent
versions of \arm**~\cite{pulte-al-17},
the annotated weak outcome cannot be
observed because these target architecture models
provide the \textit{multi-copy atomicity} guarantee,
which postulates that any two writes must be observed by all threads,
except the ones performing the two writes, in the \emph{same} order.
This multi-copy atomicity guarantee is a key property of the
\intel and \arm** architectures.
It can be exploited to simplify reasoning about the
correctness of a given program and, in some cases,
to write more efficient ones.

The problem is that the \rceleven model does not provide an efficient way
of enforcing multi-copy atomicity even when the target architecture provides
this guarantee.
\rceleven, in fact, provides only two ways to forbid the weak behavior of
\ref{prog:iriw},
both of which incur an non-negligible implementation cost on \intel.
One can either (1) strengthen all access modes to \sc, or
(2) insert an \sc fence between the two pairs of read operations.
In the context of \intel, both solutions are unsatisfactory,
as they involve additional unnecessary fences.
With the support for inline \intelext assembly,
one could imagine a third solution that consists
in strengthening the first read operation of each thread
as follows:
\begin{equation*}
\tag{IRIW-TSO}
\label{prog:iriw:asm}
\inarrII{\inarrII{
  \WritePL\rel\locx{1}
}{
  \AsmReadPL\regA\locx\\
  \ReadPL\rlx\regB\locy
}}{\inarrII{
  \AsmReadPL\regC\locy\\
  \ReadPL\rlx\regD\locx
}{
  \WritePL\rel\locy{1}
}}
\end{equation*}

One would expect this solution to work because (similar to \acq accesses)
\intelext disallows the reordering of a read operation with any other
subsequent operation.
This solution avoids the emission of fences and highlights the reliance on
an architecture-specific guarantee.

\subsection{Evaluation of Candidate Models}
\label{subsection:evaluation}

\input{figure-evaluation}

We now consider multiple tentative
hybrid models and evaluate them according to
our established criteria.
\Cref{fig:evaluation} contains a summary of our discussion.
The candidate models are organized by lines,
and the desired properties by columns.
A full star means that a model enjoys the corresponding property;
an empty star means that it does not;
a half star means that the property is partially met.

\hypertarget{A1}{}
\paragraph{Hardware approach.}
The hardware approach is perhaps the first and simplest
solution that comes to mind:
it consists of using the hardware model \intelext itself
as the hybrid model.
This seems like a plausible solution,
because a program that uses inline \intelext assembly
can only be executed on this specific architecture.
However,
one immediate deficiency of this approach is
that the \intelext model is not directly applicable to a \ccpp program;
one would first have to consider its compilation
to \intelext and only then apply the hardware model.
As a consequence,
one would have to commit to one of the compilation schemes to \intelext.
Therefore, under this approach,
the correctness of standard compilation mappings would not hold in general.
Another downside is that
this model is not an extension of \rceleven:
the semantics of a program under \intelext can clearly
disagree from that given by \rceleven.
Finally, this approach would not validate standard
compilation optimizations
as many of them, such as reordering of independent reads,
is unsound under \intelext.

\hypertarget{A2}{}
\paragraph{Branching approach.}
A slight refinement of the hardware approach
is to branch on whether the program uses inline assembly:
if it does, then the semantics is given by \intelext ;
otherwise, the semantics is given by \rceleven.
This approach improves on the previous one by
constituting an extension of \rceleven, however
most compiler optimizations would still
be unsound in programs with inline assembly.

\hypertarget{A3}{}
\paragraph{The TSO-as-RA approach.}
The next approach is to keep the \rceleven model,
and to simply map each inline assembly instruction to
an existing \ccpp construct with the same or slightly
weaker semantics.
In particular,
plain \intelext stores can be mapped to \rceleven\ \rel stores,
plain \intelext loads can be mapped to \rceleven\ \acq loads,
\intelext memory fences to \rceleven\ \sc fences, and
\intelext store fences to \rceleven\ \acqrel fences.

This approach has three major downsides.
First, it does not give any semantic benefit to using inline assembly
(\hyperlink{P5}{P5}).
Second,
it does not match the \intelext semantics for programs
consisting purely of inline assembly (\hyperlink{P3}{P3}).
For example, consider a version of~\ref{prog:iriw}
written entirely using inline assembly;
that is,
using inline-assembly reads and writes
instead of \cpp reads and writes.
According to the TSO-as-RA approach,
this inline-assembly version of~\ref{prog:iriw}
can exhibit the annotated behavior of~\ref{prog:iriw},
even though, in practice, it can never be observed.
Third, the TSO-as-RA approach cannot model all relevant \intelext features.
In particular, it cannot model \intelext non-temporal stores
because there is no corresponding \rceleven store construct
that permits the weak behavior of \ref{prog:mp-nt} from \cref{section:intro}.

\hypertarget{A4}{}
\paragraph{Projection approach.}
Given that neither \intelext nor \rceleven alone are appropriate
for ascribing semantics to \ccpp programs with inline assembly,
a natural choice is to use both models together.

At a very high level, the two models seem compatible:
they are defined in a \emph{declarative style}
as a set of constraints that
program executions should satisfy.
For instance, \rceleven states that
a read operation cannot \emph{happen before}
the write instruction from which it reads.
An instruction is said to happen before another one
(1) if it appears earlier in the same thread,
or
(2) if it appears before some release-acquire synchronization,
such as seen in the example of~\ref{prog:mp-nt}.
\intelext, on the other hand, imposes multi-copy atomicity:
the order in which independent writes are observed
is the same across all threads
(except the ones performing those writes
as they may observe their own writes early).

A natural definition for a combined model would be to take the conjunction
of the constraints of the two models, each applied only to the instructions of
the corresponding model.
In other words,
to apply the \intelext constraints to the inline-assembly
instructions and the \rceleven constraints to the \rceleven accesses.
Such a definition is clearly an
extension of \rceleven and \intelext.
Moreover, it supports the existing compilation schemes and
compiler optimizations.
It fails, however,
to provide useful semantics for programs with inline assembly:
for instance,
it does not rule out the weak behaviors of
the \ref{prog:mp-nt-sf} and~\ref{prog:iriw:asm} programs,
because it does not rule out cycles
with accesses from both models.

\hypertarget{A5}{}
\paragraph{Compound memory model approach.}

\citet{goens-al-23} propose another way of combining two memory models based on
operational semantics,
where each thread follows a single operational memory model.
Their approach is, however,
not applicable to the setting of inline assembly
because it is too inflexible:
it does not allow the use of both \intel and
\ccpp instructions in the same thread.

\hypertarget{A6}{}
\subsection{Towards a Good Hybrid Model}
\label{subsection:overview-towards}

From the approaches seen so far, only the \hyperlink{A4}{projection} approach
comes close to achieving our desiderata for a hybrid memory consistency model.
To arrive at a good hybrid model, we will therefore start with the projection
approach and refine it to strengthen the guarantees given to programs containing
both \ccpp accesses and inline \intel assembly.

\paragraph{Supporting correct message-passing patterns.}

The first necessary strengthening comes from carefully inspecting the
\ref{prog:mp-nt} and \ref{prog:mp-nt-sf} examples.
\rceleven forbids the weak behavior of the corresponding programs with only
\ccpp accesses with its \emph{coherence} condition, which says that the
\emph{extended coherence order} ($\eco$) cannot contradict the model's
\emph{happens-before} relation ($\hb$).

The extended coherence order~$\eco$, orders accesses at a given memory
location in the order they appear to have executed.
For instance,
it places all writes to the same location, say~$\locx$,
in a total order.
A read~$r$ to~$\locx$ is placed by~$\eco$
after the write~$w$ from which~$r$ reads
and before every other write that follows~$w$ according to~$\eco$ itself.
In the executions leading to the annotated outcomes of~\ref{prog:mp-nt}
and~\ref{prog:mp-nt-sf},~$\eco$ orders the write to~$\locx$
before the read to~$\locx$
(as the latter reads the initialization value,~$0$)
and orders the write to~$\locy$
before the read to~$\locy$
(as the latter reads from the former).

The happens-before relation~$\hb$,
defined as~$\transC{(\po\disj\sw)}$,
is given as the transitive closure of the union of two components:
program-order edges
($\po$, relating instructions of the same thread
in the order they appear in the program)
and synchronization edges ($\sw$) between threads,
when one thread reads from another in a synchronizing
fashion (for example, using~\rel/\acq accesses). 
In our example, the write to~$\locy$
synchronizes with the read to~$\locy$, and thus the previous write
to~$\locx$
happens before the read to~$\locx$ according to~\rceleven,
and so the read to~$\locx$
cannot read~$0$.

Clearly, to regain soundness in the model with inline assembly, we need to
adapt the definition of~$\hb$ to exclude program-order edges from
non-temporal stores to subsequent stores because these can be reordered
by~\intel. Blindly restricting the definition of~$\hb$
to relate only~\ccpp events
(as in the projection approach) is too weak because
the behavior of~\ref{prog:mp-nt-sf} would then be allowed.
A suitable definition is thus to remove
from~$\hb$ only
the~$\po$ edges between a non-temporal store and
any later instruction that is not a fence.
That is, we redefine~$\hb$
as~$\transC{(\porceleven\disj\sw)}$,
where
the relation~$\porceleven$
excludes such~$\po$ edges (see \cref{section:model}).

\paragraph{Supporting stronger architecture-specific behaviors.}

Next, we also need to strengthen the model to support the \ref{prog:iriw:asm}
example. If all accesses in the example were \intel accesses,
\intelext would forbid this outcome by its general acyclicity
condition which forbids cycles consisting of external~$\eco$ edges
(that is, ones between accesses from different threads) and its
\emph{preserved program order} ($\ppo$), which includes the
program-order edges between instructions whose ordering is
guaranteed on \intel
(for example, from \intel reads to all subsequent memory instructions).

A minimal way to extend the applicability of this condition
would be to require the cycle to contain at least one
inline-\intel-assembly instruction.
Requiring at least one assembly instruction in the cycle
prevents this new condition from breaking
\hyperlink{P3}{Property P3}:
the additional condition simply does not apply to
programs without inline assembly.
Moreover, it ascribes the intended semantics to the \ref{prog:iriw:asm}
program, forbidding its annotated weak outcome.

Sadly, however, this minimal way of adapting the \intelext
model is flawed as it does not validate compiler optimizations.
To see this, 
consider the following variant of~\ref{prog:iriw:asm}:
\begin{equation*}
\tag{IRIW-TSO-2}
\label{counter:ex:seq}
\inarrII{\inarrII{
  \WritePL\rlx\locx{1}
}{
  \AsmReadPL\regA\locx\outcomeI{1}\\
  \ReadPL\rlx\regB\locy\outcomeI{0}
}}{\inarrII{
  \ReadPL\rlx\regC\locy\outcomeI{1}\\
  \ReadPL\rlx\regD\locx\outcomeI{0}
}{
  \WritePL\rlx\locy{1}
}}
\end{equation*}

The annotated behavior is disallowed under this model
because the cycle contains one inline-assembly instruction.
However, a \ccpp compiler can reorder the accesses of the third thread and
arrive at the following program:
\begin{equation*}
\inarrII{\inarrII{
  \WritePL\rlx\locx{1}
}{
  \AsmReadPL\regA\locx\outcomeC{1}\\
  \ReadPL\rlx\regB\locy\outcomeC{0}
}}{\inarrII{
  \ReadPL\rlx\regD\locx\outcomeC{0}\\
  \ReadPL\rlx\regC\locy\outcomeC{1}
}{
  \WritePL\rlx\locy{1}
}}
\end{equation*}

The depicted outcome is now allowed: 
first~$\ReadPL\rlx\regD\locx$ reads~$0$,
then the first and second threads execute,
then the fourth thread writes~$1$ to~$\locy$,
which is finally read by the third thread.

\hypertarget{A7}{}
\subsection{Our Approach}
\label{subsection:intuitive:account}

Counterexample \ref{counter:ex:seq} shows that it is too strong to stipulate
the absence of
\intelext-consistency-violating
cycles that contain at least one \intelext event.
The weak behavior of~\ref{counter:ex:seq}
should be allowed by our model so as to validate
the reordering of \rceleven relaxed accesses
on the third thread of the program.

In order to allow the annotated behavior of \ref{counter:ex:seq}, our idea 
is to insist that \emph{all~$\ppo$ edges} in a~%
$(\ppo\disj\eco)$-cycle
(that is, in a \intelext-consistency-violating cycle)
contain at least one
\intel instruction or a \sc fence.
This is because neither \intel
instructions nor \sc fences can be optimized by the
compiler in a thread-local fashion.
Therefore,
the third thread of \ref{counter:ex:seq}
cannot contribute to the cycle that violates
\intelext-consistency,
because it contains only plain \ccpp instructions.

Extending \rceleven with this refined condition leads to a hybrid model that
enjoys all our established desiderata:
(1) it supports the established compilation schemes to \intelext;
(2) it supports all existing local compiler optimizations,
because these only affect \ccpp operations,
and thus do not affect our model's preserved program
order relation, which must include an assembly instruction or a \sc fence;
(3) it extends both \rceleven and \intelext; and
(4) it provides the intended semantics to
Program~\ref{prog:mp-nt} and
to all variants of
Program~\ref{prog:iriw} that we have encountered.

%% file: figure-evaluation.tex
\begin{figure}[t]
\begin{minipage}{.54\textwidth}
\renewcommand{\arraystretch}{1.3}\centering\small
\begin{tabular}{|c||c|c|c|c|c|c|}
\cline{2-7}
  \multicolumn{1}{c|}{}
  &{P0}&{P1}&{P2}&{P3}&{P4}&{P5}
\\
\hhline{-|=|=|=|=|=|=|}
  \hyperlink{A1}{Hardware}
    &\cmarkcell &\hmarkcell &\xmarkcell &\xmarkcell &\cmarkcell &\cmarkcell
\\\hline\hyperlink{A2}{Branching}
    &\cmarkcell &\hmarkcell &\xmarkcell &\cmarkcell &\cmarkcell &\cmarkcell
\\\hline\hyperlink{A3}{TSO-as-RA}
    &\xmarkcell &\cmarkcell &\cmarkcell &\cmarkcell &\xmarkcell &\xmarkcell
\\\hline\hyperlink{A4}{Projection}
    &\cmarkcell &\cmarkcell &\cmarkcell &\cmarkcell &\cmarkcell &\xmarkcell
\\\hline\hyperlink{A5}{{\protect\NoHyper\citet{goens-al-23}\protect\endNoHyper}}
    &\xmarkcell &\cmarkcell &\cmarkcell &\cmarkcell &\cmarkcell &\xmarkcell
\\\hline\hyperlink{A6}{Approach of \cref{subsection:overview-towards}}
    &\cmarkcell &\cmarkcell &\xmarkcell &\cmarkcell &\cmarkcell &\cmarkcell
\\\hline\hyperlink{A7}{Our approach}
    &\cmarkcell &\cmarkcell &\cmarkcell &\cmarkcell &\cmarkcell &\cmarkcell
\\\hline
\end{tabular}
\end{minipage}
\begin{minipage}{.45\textwidth}
\hyperlink{P0}{P0 - Flexibility}\\
\hyperlink{P1}{P1 - Correctness of compiler mappings}\\
\hyperlink{P2}{P2 - Correctness of compiler optimizations}\\
\hyperlink{P3}{P3 - Extension of \rcelevenname}\\
\hyperlink{P4}{P4 - Extension of \intelext}\\
\hyperlink{P5}{P5 - Strong guarantees for mixed programs}\\
\end{minipage}

\caption{%
Comparison of approaches according to several desired properties.%
}
\Description{}
\label{fig:evaluation}
\end{figure}

%% file: model.tex
\section{The Extended Model}
\label{section:model}

In this section, we present
our extension of \cpp's memory model with support for
inline \intelext assembly.
We use \rceleven~\cite{rc11} as the memory model for \cpp.
With the interest of recalling the basic notions
of~\rceleven
and setting up notation and useful definitions for the
next subsections,
we start with a brief presentation of \rceleven.
We mainly follow the original presentation by~\citet{rc11}.
We also rely on~\citet{imm-19}
for the precise construction of
\textit{execution graphs}.

\subsection{The \rceleven Memory Model}
\label{subsection:rceleven}

\rceleven defines the semantics of
multithreaded \ccpp programs.
More specifically,
\rceleven formalizes how the memory, which initially
maps every location to a default value
(usually the integer~$0$),
is updated after the execution of a program.
To account for non-determinism
(for example, due to the concurrent execution of threads),
the model associates a program~$\prog$ not with a single
final memory, but with the set of states
in which the memory can be found after the execution of~$\prog$.

The \rceleven model follows the \textit{declarative approach}.
In the declarative approach,
the set of final memory states associated with a program~$\prog$
is defined in three steps.
The first step consists in an operational semantics;
that is, a formalization of program execution.
However, this formalization does not strive to capture exactly
how the program~$\prog$ runs.
Instead, it follows a simple \textit{thread-interleaving} semantics
where threads non-deterministically take turns and contribute
to the construction of an abstract structure called
an execution graph.
An execution graph stores,
in the form of nodes,
the memory operations
(such as writes, reads, and synchronization barriers)
issued by threads.
These nodes are also called \textit{events}.
The result of the first step is thus the construction
of a set of execution graphs associated with~$\prog$.
The second step is the selection,
among this resulting set of execution graphs,
of the \textit{consistent} execution graphs.
A consistent execution graph
is one whose nodes can be connected
by extra relations
in a way that satisfies conditions
postulated by the model in question.
These conditions capture how the
model deviates from one
that would tolerate only sequentially consistent behaviors.
The third and final step amounts to mapping
every consistent execution graph to the memory
state it represents.

To illustrate the \rceleven model, we introduce \rcelevenlang,
a simple concurrent imperative programming language
with support for \cpp's memory-access modes.
Opting for a simple set of programming constructs
allows us to concentrate on the key aspect of the memory model:
the definition of the semantics of memory
operations such as read, writes, and synchronization barriers.

\input{figure-cpp-syntax}

\fref{fig:cpp:syntax} shows the syntax of \rcelevenlang.
The language is parametric on a set of registers, $\typeReg$,
and introduces a set of (preallocated)
memory locations, $\typeLoc$,
defined as the set of natural numbers.
Expressions~$\expr$ are used
to compute numbers~$\intn$ or locations~$\loc$
by reading numbers stored in registers~$\reg$
and performing arithmetic operations.
The syntactic category of commands,~$\typeCmd$,
includes \kw{if} branching, \kw{while} loops,
sequential composition, a \kw{skip} instruction,
and memory operations,
such as
reads, writes, read-modify-writes (RMWs), and fences.
The notation~$\kw{[}\expr\kw{]}$
is used to indicate that~$\expr$ denotes
a memory location rather than a number.
Every memory operation carries an access mode~$\md$.
Access modes are ordered according to the
diagram depicted in \fref{fig:cpp:syntax}.
To give an (over-simplistic) intuitive explanation of access modes,
we can say that~$\sc$ operations follow
a sequentially consistent semantics,
and operations with a weaker access mode~$\md$
follow a semantics that deviates from sequential consistency
to a degree that is proportional to how distant~%
$\md$ is from~$\sc$.
Only certain access modes are permitted per operation:
\begin{itemize}
\item Modes~$\na$, $\rlx$, $\rel$, and~$\sc$ apply to writes.
\item Modes~$\na$, $\rlx$, $\acq$, and~$\sc$ apply to reads.
\item Modes~$\acq$, $\rel$, $\acqrel$, and~$\sc$ apply to fences.
\item Modes~$\rlx$, $\acq$, $\rel$, $\acqrel$, and~$\sc$ apply to read-modify-writes.
\end{itemize}

Finally, a program~$\prog\in\typeProg$
is defined as a collection of commands, %
represented as a finite map from numbers
(or \textit{thread identifiers}) to commands:~%
$\typeProg\eqdef\typeFMap\NN\typeCmd$.

We formalize an \textit{event} either as
an \textit{initialization event}~$\initEv\loc$,
representing the initialization of $\loc$
with the default value~$0$,
or as a pair of natural numbers~$\pair\inti\intj$,
where~$\inti$ is a thread identifier
and~$\intj$ is the order of this event
with respect to the
events emitted by thread~$\inti$.
(These numbers are used, for example,
in the definition of the~%
\hyperlink{PO}{\emph{program-order}}
relation.)
An execution graph is represented as a pair
of a set of events~$\EventE$
and a map~$\lab$ from events to \textit{labels}.
A label specifies both the type of a memory event
(whether it is a read, a write, a read-modify-write, or a fence)
and its arguments.
A read label is represented as~$\Read\md\loc\intn$;
a write label is represented as~$\Write\md\loc\intn$;
a fence is represented as $\Fence\md$;
and a read-modify-write label is represented as~%
$\RMW[\md]\loc\intn{\intm^{?}}$,
where~$\intm^{?}$ denotes either a number or
the marker~$\bot$ representing the case of a failed
read-modify-write operation.
We are often lax about the distinction between events and labels;
we use them interchangeably.
Moreover, we write~$\ReadE$,~$\WriteE$,~$\FenceE$, and~$\RMWE$
to denote respectively the sets of events whose label is a read,
a write, a fence, and a read-modify-write.
We further partition~$\RMWE$
into its subset of successful read-modify-writes~$\RMWSE$
and its subset of failed read-modify-writes~$\RMWFE$.

The construction of the set of execution graphs associated
with a program relies on the notions of \textit{threads}
and \textit{thread pools}.
A thread pool is modeled as a finite map from
thread identifiers to threads.
A thread, in its turn, is modeled as
a tuple containing the following fields:
\regst, which maps a register to the number it stores;
\evcounter, which stores the number of events
issued by the thread;
and
\nextcmd, which stores the next command to be
executed by the thread.
In sum, here is the definition of the
set of threads,~$\typeThread$,
and of the set of thread pools,~$\typePool$:
\[\begin{array}{l@{\hspace{8mm}}r}
\pool\in\typePool\eqdef\typePoolDef&
\thread\in\typeThread\eqdef\typeThreadDef
\end{array}\]

On top of these definitions,
the set of candidate execution graphs
associated with a program is captured by
the \textit{pool reduction} relation,
a relation between pairs of pools and execution graphs.
It is noted~$\poolstep\poolA\graphA\poolB\graphB$.
Intuitively, the statement~%
$\poolstep*{\toPool\prog}\initGraph\emptyset\graphA$
expresses that~$\graphA$ is an execution graph
associated with~$\prog$.
The graph~$\initGraph$ in this statement
denotes the \textit{initial execution graph},
a graph
where~$\initGraph.\EventE$ is the set of initialization events~$\initEv\loc$
for every location~$\loc$,
and
where~$\initGraph.\lab$ maps~$\initEv\loc$ to~$\Write\na\loc{0}$.
The pool~$\emptyset$ denotes a thread pool whose domain is empty.
The pool~$\toPool\prog$ denotes a thread pool
in its initial state:
\[
\toPool\prog
  \eqdef
    \lambda\inti\in\dom\prog.\;
      \left\{\begin{array}{@{\,}l@{\,}}
        \regst=\lambda\_.0;\,
        \evcounter=0;\,
        \nextcmd=\Seq{prog(\inti)}{\Skip}
      \end{array}\right\}
\]

\input{figure-cpp-op-semantics}

\fref{fig:pool:reduction} includes some illustrative cases
of the pool reduction relation.
The complete definition can be found in the
Appendix~\appendixref{section:execution:graphs:app}.
Some cases rely on the interpretation of an expression~$\expr$
under a map~$\phi$ from registers to numbers.
This interpretation, noted $\interp\expr_\phi$,
is simply defined as the interpretation of the syntactic
arithmetic operators as their mathematical counterpart.
Rule~\ruleReadStep shows how a new read event~$\evA$
is added to the execution graph when a read operation
is executed.
There is no restriction to the
value~$\intn$ returned by the read operation.
It is only at the level of execution graphs that
consistency conditions are imposed
and certain values are ruled out.
Rule~\ruleTerminateStep shows
how completed threads are removed from the pool.
Eventually, all threads complete their execution
and the pool degenerates to~$\emptyset$.

To define \rceleven's notion of a consistent execution graph,
we need to introduce the program-order relation~$\po$
and
we need
to consider the extension of an execution graph with a
\textit{reads-from} relation~$\rf$
and a 
\textit{modification-order} relation~$\mo$.
We are often lax about the distinction between
an execution graph~$\graphA$ and its
extension~$\tuple\graphA{\rf,\mo}$.

\paragraph{Notation.}
The metavariables~$\evA$,~$\evB$,~$\evC$,~$\evD$, and~$\evE$
range over events.
An event, as we recall,
is formalized as either
an initialization event,~$\initEv\loc$,
or as a pair of natural numbers,~$\pair\inti\intj$,
where~$\inti$ is a thread identifier
and~$\intj$ is the order of the event.
The terms~$\fst\evA$ and~$\snd\evA$
denote the first and the second projections
of~$\evA$ in the case where~$\evA$ is a pair.
The relation~$\inv{R}$ is the \textit{inverse relation} of~$R$:
$\pair\evB\evA\in\inv{R} \iff \pair\evA\evB\in R$.
The relation~$R_1\seq R_2$ is the \textit{sequential composition}
of~$R_1$ and~$R_2$:
$\pair\evA\evC\in R_1\seq R_2 \iff
\exists\evB.\;
\pair\evA\evB\in R_1 \land
\pair\evB\evC\in R_2 $.
The relation~$\settorel{S}$
is the smallest reflexive relation on a set~$S$;
it is defined as~$\{\pair{s}{s} \mid s \in S\}$.
The relations~$\refC{R}$,~$\transC{R}$, and~$\reftransC{R}$
respectively denote
the reflexive closure,
the transitive closure,
and the reflexive-and-transitive closure
of~$R$.
The relations~$\internal{R}$ and~$\external{R}$
are the \textit{internal} and \textit{external}
components of~$R$:
$\pair\evA\evB\in\internal{R}
\iff
\pair\evA\evB\in{R}\land\fst\evA=\fst\evB$,
and,
$\external{R} = R \setminus \internal{R}$.
Given a graph~$\graphA$,
the relation~$\atloc{R}\loc$
is the \textit{at-$\loc$} restriction of~$R$:
it restricts~$R$ to events~$\evA$
such that~$\graphA.\lab(\evA)$ accesses~$\loc$.
The term~$\dotLoc\evA$ denotes the location
accessed by~$\evA$.
The relation~$\perloc{R}$ is the
\textit{per-location} restriction of~$R$:
$\pair\evA\evB\in\perloc{R}
\iff
\pair\evA\evB\in{R}\land\dotLoc\evA=\dotLoc\evB$.
The relation~$\neqloc{R}$ is the
\textit{distinct-locations} restriction of~$R$:
$\neqloc{R} = R \setminus \perloc{R}$.%
All these restrictions can be similarly applied to sets of events.
The graph~$\graphA$
is usually clear from the context
and left implicit.

\hypertarget{PO}{}
\paragraph{Program order.}
The program order reflects the order
in which events were emitted by a given thread:~%
$
\pair\evA\evB\in\po
  \iff
(\evA = \initEv\_
  \,\land\,
\evB \neq \initEv\_)
\,\lor\,
(\fst\evA = \fst\evB
  \,\land\,
\snd\evA < \snd\evB)
$.

\paragraph{Reads-from.}
The reads-from relation relates write events to read events,
$\rf\subseteq \typeProd{(\WriteE\cup\RMWSE)}{(\ReadE\cup\RMWE)}$.
It captures how information flows
from a write to a read on the same location.
There are two conditions.
First,
for every read~$\evB$, there must be a unique write~$\evA$
such that~$\pair\evA\evB\in\rf$.
Second, for every pair~$(\evA,\evB)\in\rf$,
the events~$\evA$ and~$\evB$ must act on the same location
and the value read by~$\evB$ must be
equal to the value written by~$\evA$.

\paragraph{Modification order.}
The modification order is a relation on write
and successful read-modify-write events,
$\mo\subseteq\typeProd{(\WriteE\cup\RMWSE)}{(\WriteE\cup\RMWSE)}$.
Intuitively, it describes how single memory cells
have been observed to evolve during program execution.
The~$\mo$ relation is equal to the disjoint union
of the relations~$\atloc\mo\loc$,
defined as the restriction of~$\mo$ to events in~$\loc$:
$\mo = \bigcupplus_{\loc\in\typeLoc} \atloc\mo\loc$.
Moreover, for every~$\loc$, the relation~$\atloc\mo\loc$
is a \textit{strict total order} (transitive, irreflexive, and total).

We are finally in position to introduce the
\rceleven-consistency conditions:
\begin{definition}[\rceleven-Consistency]
\label{def:rceleven:consistent}
\input{rc11-consistency}
\end{definition}

These consistency conditions are equivalent to
the ones formulated by \citet{margalit-lahav-21},
who diverge from~\citet{rc11} only in a minor way:
the synchronizes-with relation relies on a
simplified notion of \emph{release sequences},
defined as the reflexive-and-transitive closure of~$\rf$.
This simplification is in agreement with
the current documentation of the
\cpp programming language~\cite{cppref-memorder}.
We further adapt the statement of \axiomAtomicity according
to our design choice of modeling RMWs as single events
rather than as pairs of reads and writes related by an extra
relation~$\rmw$.

To complete the description of~\rceleven,
showing how it defines the semantics of a program,
we need to introduce the notions of \textit{data race}
and of \textit{undefined behavior} ($\UB$):
\begin{definition}[Data Race]
\label{def:data:race}
A pair of events~$\pair\evA\evB$ forms
a \emph{data race}
if the following conditions hold:
(1)~$\evA \neq \evB$,
(2)~$\dotLoc\evA = \dotLoc\evB$,
(3)~$\{\evA, \evB\}\cap(\WriteE\cup\RMWSE) \neq \emptyset$,
and
(4)~$\pair\evA\evB\notin\hb\disj\inv\hb$.
\end{definition}

\begin{definition}[\rceleven-Behaviors]
\label{def:rceleven:behaviors}
\[\begin{array}{rcl}
\left(\begin{array}{@{\,}c@{\,}}
\poolstep*{\toPool\prog}\initGraph\emptyset\graphA
  \;\land\;
\tuple\graphA{\rf,\mo}\;\textit{is}\;
    \textit{\rceleven-consistent}
\end{array}\right)
  &\vdash&
    \progbeh\prog{\tuple\graphA{\rf,\mo}}
\\[1mm]
\left(\begin{array}{@{\,}c@{\,}}
\poolstep*{\toPool\prog}\initGraph\_\graphA
  \;\land\;
\tuple\graphA{\rf,\mo}\;\textit{is}\;
    \textit{\rceleven-consistent}
\\
\land\;
\pair\evA\evB\;\textit{forms a data race}
 \;\land\;
\na\in\{\dotMode\evA,\dotMode\evB\}
\end{array}\right)
  &\vdash&
    \progbeh\prog\UB
\end{array}\]
\end{definition}

Each consistent execution graph~$\tuple\graphA{\rf,\mo}$
represents one of the possible final-memory states of a program.
We use the function~$\finalSt$ to extract this memory state:~%
$\finalStTypical$ denotes the memory
where a location~$\loc$ stores the value~$\intn$
of the last write event~$\Write{}\loc\intn$
in~$\graphA$ with respect to~$\mo$.
The memory~$\finalStTypical$ is represented as a partial map
where a location~$\loc$ belongs to~$\dom{\finalStTypical}$
iff there exists~$\evA\neq\initEv\_$ such that~%
$\graphA.\lab(\evA)\in\atloc\WriteE\loc\cup\atloc\RMWSE\loc$.
\begin{definition}[\rcelevenlang Semantics]
\label{def:rcelevenlang:semantics}
The semantics of a \rcelevenlang program~$\prog$ is defined as
its set of final states:
\[
\store\in\interp[\rceleven]\prog
  \iff
\progbeh\prog\UB
\;\lor\;
\exists\,\graphA,\,\rf,\,\mo.\;
\progbeh\prog{\tuple\graphA{\rf,\mo}}
  \,\land\,
\store = \finalStTypical
\]
\end{definition}

\subsection{The \rcelevenext Memory Model - An Extension of \rceleven with Inline \intelext Assembly}
\label{subsection:rcelevenext}

We now introduce \rcelevenext,
an extension of \rceleven with inline \intelext assembly.
We illustrate the model in
an extension of \rcelevenlang with inline assembly,
called \rcelevenextlang.

\begin{figure}[t]
\[\begin{array}{@{}r@{\;}r@{\;}l@{}}
  \typeCmd\ni\cmd ::=
             \ldots
    & \mid & \AsmReadPL\reg\expr
      \mid   \AsmWritePL\expr\expr
      \mid   \AsmRMWPL\reg\expr\expr\expr
      \mid   \AsmMFencePL
  \\
    & \mid & \AsmNTWritePL\expr\expr
      \mid   \AsmSFencePL
\end{array}\]
\caption{Syntax of \rcelevenextlang.}
\Description{}
\label{fig:cpp:x86:syntax}
\end{figure}

\fref{fig:cpp:x86:syntax}
shows the syntactical increments of \rcelevenextlang
over \rcelevenlang.
The main difference with respect to \fref{fig:cpp:syntax}
is the addition of inline-assembly commands,
distinguished by the prefix~\kw{\textbf{asm}}.
They allow one to access the following
\intelext-specific instructions:
\textit{plain \intelext reads, writes, and read-modify-writes};
\textit{non-temporal stores};
\textit{store fences}; and
\textit{memory fences}.

To give an intuitive operational account of
these instructions,
we can rely on the formal operational model of
\intelext~\cite{raad-22}.
In this operational model,
every thread contains a local buffer where write
instructions first take effect before reaching
the global main memory,
which is shared among all threads.
A non-temporal store~$\NTWritePL\exprA\exprB$
bypasses the local buffer,
if the buffer contains no writes to the same location.
Therefore,
a non-temporal store can be reordered
with respect to writes or non-temporal stores
to different locations.
A store fence~$\SFencePL$ can be used
to avoid the reordering of non-temporal stores.
A memory fence~$\MFencePL$ can be used for the same
purpose.
Additionally,
it can be used to stop the reordering
of a write followed by a read.

\input{figure-cpp-intelext-diagram}

To distinguish events emitted by inline-assembly commands
from events emitted by pure \rcelevenlang commands,
we introduce three new access modes:
\[\typeMode\ni\md ::= \ldots\mid\nt\mid\stf\mid\tso\]

Events emitted by plain \intelext reads, writes,
and read-modify-writes carry the mode~\tso:~%
$\WriteMd\tso$,~$\ReadMd\tso$, and~$\RMWMd\tso$.
Events emitted by non-temporal stores
carry the mode~\nt:~$\WriteMd\nt$.
Events emitted by store fences carry
the mode~\stf:~$\FenceMd\stf$.
Events emitted by \intelext memory fences are
indistinguishable from
those emitted by \sc fences,
they all carry the mode~\sc.
Of course, it would be possible to distinguish
events emitted by memory fences by using
an extra mode, say~$\mf$.
However, our model assigns the same strength
to \sc fences and to memory fences,
so we prefer to simply use the mode~\sc.
(In other words, in our proposed model, programmers have no good reason to use~$\AsmMFencePL$,
as they can equivalently use~$\FencePL\sc$;
we include~$\AsmMFencePL$
only for comprehensiveness.)

The following definition introduces \rcelevenext-consistency.
Many of the conditions are identical to those from \rceleven
(\defref{def:rceleven:consistent}).
Therefore, to avoid repetition,
we include only the differences with respect to \rceleven.
For clarity,
we highlight these differences
using a \greenhl{colored background}.
Finally, we observe that
(in both the new definitions and in those inherited from \rceleven)
the ranges of access modes should be interpreted using
the graph from \fref{fig:cpp:x86:diagram};
that is,
using the order induced by
the reflexive-and-transitive closure of the
directed-edge relation from \fref{fig:cpp:x86:diagram}.

\begin{definition}[\rcelevenext-Consistency]
\label{def:rcelevenext:consistent}
\input{rc11-ex86-consistent}
\end{definition}

This definition diverges from \rceleven in multiple ways:
\begin{description}
\item[Diagram of access modes.]
The diagram of access modes unites \rceleven modes
and \intelext-inspired modes
into the same picture.
It is intriguing because
it misses some
orderings that one would naturally expect,
such as~$\tso\sqsubset\sc$ or perhaps
even~~$\na\sqsubset\nt$.
Given that non-temporal stores break
release-acquire synchronization,
as we shall explain,
it is not
difficult to understand the absence of
the ordering~$\na\sqsubset\nt$.
Perhaps more striking
is the absence of the ordering~$\tso\sqsubset\sc$.
We explain in~\cref{subsubsection:diagram}
that adding such an ordering
violates (at least) one of our desiderata.

\item[Definition of $\hb$.]
Instead of the full~$\po$ relation,
now the definition of~$\hb$ uses
a restricted version of~$\po$ that excludes
edges starting in non-temporal stores,
unless they reach a~\sc fence,
a~\stf fence,
a~\tso read-modify-write,
or a write to the same location.
In~\cref{subsubsection:definition:hb},
we explain in detail the motivation for this change,
but, for now, let us simply say that 
this relaxation of~$\hb$ is necessary, for example,
to allow the weak behavior of Program~\ref{prog:mp-nt}.

\item[Definition of $\eco$.]
In~\rceleven, the relation~$\eco$
can be defined using either the full~$\rf$ relation
or the external restriction~$\rfe$.
The two formulations of~\rceleven are equivalent.
In the presence of inline assembly,
especially of non-temporal stores,
however,
the definition of~$\eco$ must use~$\rfe$:
a formulation of~\rcelevenext where~$\eco$
is defined using~$\rf$ is unsound.
In~\cref{subsubsection:definition:eco},
we explain in detail why this is the case.

\item[Consistency Condition] - \axiomCoherenceII.
The consistency conditions now postulate the
absence of cycles in~$\ppoasm\disj\eco$.
This condition is the key principle that
allows one to reason about inline assembly
using our model.
In~\cref{subsubsection:condition:coherence:II},
we shall see that this condition is an 
adaptation of one of \intelext-consistency
conditions.
We believe that extensions
of \rceleven with support for inline assembly
for other architectures could be obtained by
redefining $\ppoasm$.

\item[Consistency Condition] - \axiomCoherenceIII.
The addition of this condition is a technicality.
In \rceleven,
Condition~\axiomCoherenceI ensures that~$\moi$
and~$\rbi$ are included in~$\po$.
In \rcelevenext,
however,
Condition~\axiomCoherenceI is insufficient to rule out
cases that violate these properties,
because a~$\po$ edge that starts with a non-temporal store
is not necessarily included in~$\hb$.
As a consequence,
the existence of an event~$\evA$
such that~%
$\pair\evA\evA\in
\settorel{\WriteMd\nt}\seq\po\seq(\rb\disj\mo)$
is not a contradiction to~$\irr{\hb\seq\eco}$.
This new condition must therefore be included.
\end{description}

\subsubsection{Diagram of Access Modes.}
\label{subsubsection:diagram}
Let us start by explaining how the mode~\stf
fits in \fref{fig:cpp:x86:diagram}.
It naturally sits between the two strongest
modes allowed in a fence:~$\acqrel$ and~$\sc$.
This positioning is natural because
an~$\acqrel$ fence is erased by the
standard compilation schemes to \intel,
so they cannot be used to
stop the reordering of non-temporal stores.
Moreover, a \sc fence can be used to
stop the reordering of a write and a read,
for which a store fence is insufficient.
This explains the ordering~$\stf\sqsubset\sc$.

An interesting implication of the (derived) ordering~%
$\rel\sqsubset\stf$ is that the model allows store fences
to establish release-acquire synchronization.
In other words,
a store fence is allowed in the beginning of a~$\sw$ edge.
It can thus be used to rule out behaviors that contradict
the irreflexivity of~$\hb\seq\refC\eco$ (\axiomCoherenceI).
This is exhibited by the following pair of programs:

\begin{equation*}
\begin{array}{@{}l@{\hspace{2cm}}r@{}}
\inarrII{
  \AsmNTWritePL\locx{1}\\
  \mathhl{\FencePL\rel}\\
  \WritePL\rlx\locy{1}
}{
  \ReadPL\acq\regA\locy\outcomeC{1}\\
  \ReadPL\rel\regB\locx\outcomeC{0}
}
&
\inarrII{
  \AsmNTWritePL\locx{1}\\
  \mathhl{\AsmSFencePL}\\
  \WritePL\rlx\locy{1}
}{
  \ReadPL\acq\regA\locy\outcomeI{1}\\
  \ReadPL\rel\regB\locx\outcomeI{0}
}
\end{array}
\end{equation*}

The behavior depicted is allowed by our model in the program on the left,
but forbidden in the program on the right.
This is in agreement with the behavior exhibited by these
programs in \intelext after compilation,
because the \rel fence would then be erased.

Let us now explain the positioning of~\nt
in the diagram.
That non-temporal stores are deemed weaker than
relaxed writes is easy to understand
when we take Program~\ref{prog:mp-nt} into account.
Indeed,
the weak behavior of~\ref{prog:mp-nt}
is disallowed when a \rlx write is used
instead of a non-temporal store:
\begin{equation*}
\inarrII{
  \mathhl{\WritePL\rlx\locx{1}}\\
  \WritePL\rel\locy{1}
}{
  \ReadPL\acq\regA\locy\outcomeI{1}\\
  \ReadPL\rlx\regB\locx\outcomeI{0}
}
\end{equation*}

This explains the ordering~$\nt\sqsubset\rlx$.

The lack of the ordering~$\na\sqsubset\nt$
can be similarly explained:
\begin{equation*}
\begin{array}{@{}l@{\hspace{2cm}}r@{}}
\inarrII{
  \mathhl{\WritePL\na\locx{1}}\\
  \WritePL\rel\locy{1}
}{
  \ReadPL\acq\regA\locy\outcomeI{1}\\
  \IfThenVertical{\regA\;\kw{==}\;1}{
    \ReadPL\rlx\regB\locx\outcomeI{0}
  }
}
  &
\inarrII{
  \mathhl{\AsmNTWritePL\locx{1}}\\
  \WritePL\rel\locy{1}
}{
  \ReadPL\acq\regA\locy\outcomeC{1}\\
  \IfThenVertical{\regA\;\kw{==}\;1}{
    \ReadPL\rlx\regB\locx\outcomeC{0}
  }
}
\end{array}
\end{equation*}

The \kw{if}-branching is just to prevent a data race
between the~\na write and the~\rlx read to~$\locx$:
it makes sure that, when the read is issued,
it is preceded by the write
with respect to~$\hb$.
The program on the left cannot exhibit
the depicted behavior because of a cycle
in~$\hb\seq\rb$,
forbidden in both \rceleven
and~\rcelevenext
(since it is an \hyperlink{P4}{extension of \rceleven}).
The program on the right can exhibit
the annotated behavior because of the
reordering of non-temporal stores with writes
to distinct locations.

The lack of the ordering $\nt\sqsubset\na$
is justified by the \textit{catch-fire}
semantics of $\na$.
A data race makes every behavior
allowed by the model:
\begin{equation*}
\begin{array}{@{}l@{\hspace{2cm}}r@{}}
\inarrII{
  \mathhl{\WritePL\na\locx{1}}
}{
  \ReadPL\rlx\regA\locx\\
  \ReadPL\rlx\regB\locy
  \outcomeC{42}
}
&
\inarrII{
  \mathhl{\AsmNTWritePL\locx{1}}
}{
  \ReadPL\rlx\regA\locx\\
  \ReadPL\rlx\regB\locy
  \outcomeI{b \neq 0}
}
\end{array}
\end{equation*}

This example might instigate the reader to ask the question:
why do non-temporal stores,
or, more generally, inline-assembly accesses,
not follow a catch-fire semantics?
There are multiple reasons to avoid this approach.
First,
assigning catch-fire semantics to racy inline-assembly accesses
compromises \hyperlink{P5}{Property P5}
(because it allows the behavior of~\ref{prog:iriw})
and \hyperlink{P4}{Property P4}
(because the semantics of a
racy program written entirely using inline
\intelext assembly would diverge from the semantics
given by \intelext).
Additionally,
the reasons that justify the catch-fire semantics of \na accesses
do not apply to inline-assembly accesses.
Indeed, there are roughly two reasons why
the catch-fire semantics of \na accesses
is necessary:
(1) to validate compiler optimizations
(for example,
the reordering of \na accesses to different locations),
and
(2) to support the mapping of \na accesses to plain accesses in
architectures that do not enforce the acyclicity of~$\po\disj\rf$.
In our setting, the compiler is not expected to reorder
inline assembly, and
our compilation schemes are only to \intelext,
which enforces this acyclicity condition.

Finally, let us explain
how \tso is placed in the diagram.
Because the \hyperlink{P3}{strengthening to \tso accesses}
is one of our desired properties,
\tso is placed above every non-\sc access.
The lack of the ordering~$\tso\sqsubset\sc$
however
is intriguing,
because sequential consistency is stronger than
\textit{total store order}~\cite{sindhu-92}.
The problem is that,
in general,
\rceleven does not enforce
SC semantics to programs that mix \sc and non-\sc accesses
to the same location.
The following pair of examples
(inspired by the Z6.U example from \cite{rc11})
shows that
the semantics assigned to \sc accesses by \rceleven
can be weaker than
the semantics assigned to \tso accesses by our model:
\begin{equation*}
\begin{array}{@{}lr@{\hspace{4cm}}}
\inarrIII{
  \mathhl{\AsmWritePL\locx{1}}\\
  \WritePL\rel\locy{1}
}{
  \mathhl{\AsmReadPL\regA\locy}\outcomeI{1}\\
  \ReadPL\rlx\regB\locz\outcomeI{0}
}{
  \WritePL\sc\locz{1}\\
  \FencePL\sc\\
  \ReadPL\rlx\regC\locx\outcomeI{0}
}
    &
\end{array}
\end{equation*}

\begin{equation*}
\begin{array}{@{\hspace{4cm}}r@{}}
\inarrIII{
  \mathhl{\WritePL\sc\locx{1}}\\
  \WritePL\rel\locy{1}
}{
  \mathhl{\ReadPL\sc\regA\locy}\outcomeC{1}\\
  \ReadPL\rlx\regB\locz\outcomeC{0}
}{
  \WritePL\sc\locz{1}\\
  \FencePL\sc\\
  \ReadPL\rlx\regC\locx\outcomeC{0}
}
\end{array}
\end{equation*}

\subsubsection{Consistency Condition} - \axiomCoherenceII.
\label{subsubsection:condition:coherence:II}
Condition \axiomCoherenceII is the key principle
that allows one to reason about programs with inline assembly.
Ideally, one would like to reason about such instructions
using the hardware model, \intelext,
by relying on the guarantee that every cycle containing
at least one inline-assembly instruction should comply
to \intelext-consistency.
However,
as explained in \cref{subsection:intuitive:account},
such an approach would be too strong,
ruling out behaviors that could be introduced
by standard compiler optimizations.
We thus argued that a possible solution would be
to enforce the guarantee that every cycle in which
every pair of $\po$-separated events contains
at least one inline assembly instruction should comply
to \intelext-consistency.
This is the approach implemented by \axiomCoherenceII,
with some small caveats.

The formulation of \intelext-consistency, as introduced by~\citet{raad-22},
includes two consistency conditions:
an \emph{internal} condition,
which applies to cycles confined within single threads;
and an \textit{external} condition,
which posits the absence of certain cycles
spanning over multiple threads.

The internal condition in \intelext
posits the irreflexivity of~$\po\seq(\rfi\disj\moi\disj\rbi)$.
This condition
is equivalent to the irreflexivity of both~$\po\seq\rfi$
and~$\po\seq(\moi\disj\rbi)$.
\footnote{%
The condition
$\irr{{A}\seq({B}\disj{C})}$
is equivalent to
$\irr{{A}\seq{B}}\land\irr{{A}\seq{C}}$,
for any relations~$A$,~$B$, and~$C$.
}
Condition \axiomNTA is stronger than the irreflexivity of~$\po\seq\rfi$,
and
Conditions~\axiomCoherenceI and~\axiomCoherenceIII
together rule out reflexive edges in~$\po\seq(\moi\disj\rbi)$.

Therefore,
Condition \axiomCoherenceII focus
on integrating the external condition to the model.
To recall the definition of \intelext's external condition,
and to make its comparison with \axiomCoherenceII clear,
we include this definition here,
putting it side-by-side with \axiomCoherenceII:
\hypertarget{\nameExternal}{}
\[\begin{array}{@{}l@{\quad\quad}|@{\quad\quad}r@{}}
\includegraphics{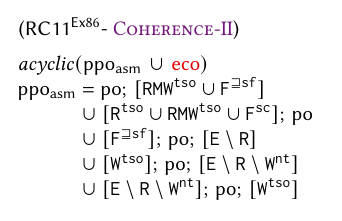}
&
\includegraphics{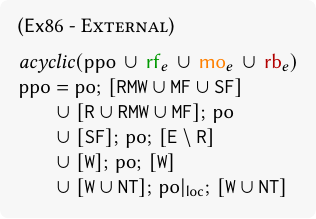}
\end{array}\]

We keep the notation used by~\citet{raad-22} in the statement
of \axiomExternal,
which diverges from ours in two minor ways:
(1) instead of a single set of fences,
\intelext introduces one set exclusively for store fences (\SFenceE)
and one set exclusively for memory fences (\MFenceE);
(2) analogously,
instead of a single set of write events,
there is one exclusive set for non-temporal stores (\NTWriteE)
and one for regular writes (\WriteE).

The side-by-side comparison reinforces the claim
that \axiomCoherenceII integrates \intelext-consistency
into \rcelevenext under the condition that
pairs of~$\po$-separated events in a violating cycle
include at least one inline-assembly event.
Indeed, most cases of~$\ppoasm$ edges
either start or end in a event with mode~$\tso$,~$\nt$, or~$\stf$.
There is one exception to this case:
edges that either either start or end in a \sc fence.
This is explained by how we model memory fences.
The condition therefore rules out certain kinds of cycles
with no inline-assembly instructions,
provided that the $\po$-separated events include a \sc fence.
Such cycles however are already ruled out by Condition~\axiomSC.

To conclude, let us comment on the differences
between the statements of the acyclicity conditions:
\axiomCoherenceII uses~$\eco$,
which includes the internal edges~$\moi$ and~$\rbi$;
whereas
\axiomExternal uses~$\rfe\disj\moe\disj\rbe$,
thereby including only external edges.
The inclusion of~%
$\settorel{\WriteE\cup\NTWriteE}\seq
\perloc\po\seq
\settorel{\WriteE\cup\NTWriteE}$
edges in the definition of~$\ppo$
compensates for the absence of~$\moi$,
whereas
the inclusion of~$\settorel{\ReadE}\seq\po$ edges
compensates for the absence of~$\rbi$.
This explanation also justifies why,
in the statement of \axiomCoherenceII,
we can omit the ``per-location'' case in the definition of~$\ppoasm$,
and reuse~$\eco$.
The attentive reader might notice that the internal edges
in~$\eco$ evade the constraint of one inline-assembly event
per pair of $\po$-separated events.
They however pose no risk to the soundness
of compiler optimizations,
because
(1) no optimization applies to pairs of
a read and a write to the same location,
so~$\rbi$ edges cannot be undone;
and (2)~$\moi$ edges between
plain \rceleven accesses
in a~$\ppoasm\disj\eco$ cycle can always be merged
into an edge of type~$\moe$, $\rbe$, or~$\ppoasm$.

\subsubsection{Definition of~$\hb$.}
\label{subsubsection:definition:hb}
To see why \hb is defined using~$\porceleven$ instead of~$\po$,
let us consider Program~\ref{prog:mp-nt}.
As we shall see,
whether the final state~$\store$ that maps both~$\locx$
and~$\locy$ to~$1$ is allowed
(that is, whether~$\store\in\interp{\textnormal{\ref{prog:mp-nt}}}$)
depends on the definition of~$\hb$.

In our model,
the final state~$\store$ is allowed,
thanks to the use of $\porceleven$ in the definition of~$\hb$.
If, however,~$\hb$ was defined as in \rceleven,
that is, $\hb_\rceleven = \transC{(\po\disj\sw)}$,
then the state~$\store$ would be disallowed.
This is of course problematic,
because
the behavior is allowed by the \intelext-compiled
version of this program.

In~\cref{section:intro},
we informally justified why this behavior
is allowed in \intelext after compilation
in terms of possible reorderings.
Having introduced the key consistency condition
of \intelext (Condition \axiomExternal),
we can now formally justify why this is the case.
We take this opportunity to illustrate our idea
of \textit{mixed execution graphs},
a reasoning tool we introduce to conduct
proofs of compilation correctness.
It allows us to represent graphs from
both source and compiled programs simultaneously:
\begin{equation*}
\includegraphics{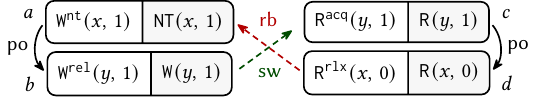}
\end{equation*}

Nodes in this graph carry
pairs of a \rcelevenext event, issued by the source program,
and a \intelext event, issued by the compiled program.
Using this structure,
we are able to make several observations:
\begin{enumerate}
\item The behavior is allowed by \intelext after compilation,
because~$\pair\evA\evB\notin\ppo$,
therefore
the cycle~$\tuple\evA{\evB,\evC,\evD}$ does
not violate~\axiomExternal.

\item The behavior is allowed by~\rcelevenext.
Two conditions could potentially be violated by the
cycle~$\tuple\evA{\evB,\evC,\evD}$:
\axiomCoherenceI and \axiomCoherenceII.
The cycle does not violate \axiomCoherenceI,
because~$\pair\evA\evB\notin\porceleven$.
The cycle does not violate \axiomCoherenceII,
because~$\pair\evA\evB\notin\ppoasm$.

\item The behavior breaks the irreflexivity of~$\hb_\rceleven\seq\eco$,
because~$\pair\evA\evD\in\hb_\rceleven$
and~$\pair\evD\evA\in\rb\subseteq\eco$.
Therefore, a naive extension of~\rceleven that
keeps~$\hb_\rceleven$ would be unsound.
\end{enumerate}

\subsubsection{Definition of~$\eco$.}
\label{subsubsection:definition:eco}
To see why~$\rfe$ is used in~$\eco$,
let us consider the following example:
\begin{equation*}
\begin{array}{@{}c@{}}
  \inarrII{
    \AsmNTWritePL\locx{1}\\
    \ReadPL\rlx\regA\locx\outcomeC{1}\\
    \WritePL\rel\locy{1}
  }{
    \ReadPL\acq\regB\locy\outcomeC{1}\\
    \ReadPL\rlx\regC\locx\outcomeC{0}
  }
\end{array}
\end{equation*}

This program is a slight variation of \ref{prog:mp-nt},
where we add a read instruction between the
non-temporal store and the write to~$\locy$.
Again, we wish to study whether the annotated behavior
is allowed by \intelext after compilation.
If that is the case,
then the behavior must be allowed by our model.
As we shall see,
the behavior is indeed exhibited by the
compiled program and our model
correctly allows it,
thanks to the exclusion of~$\rfi$ edges from~$\eco$.
The following mixed execution graph
helps to sustain these claims:
\begin{center}
\includegraphics{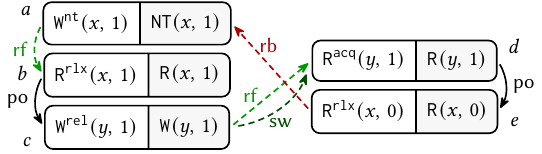}
\end{center}

This is the only execution graph that corresponds to
the annotated behavior,
because these~$\rf$ edges are the only ones that
comply with the results of the read operations.
Here is the summary of the conclusions we can
draw by studying this graph:
\begin{enumerate}
\item The behavior is allowed by \intelext after compilation,
because the graph is \intelext-consistent.
Indeed, both the edges~$\pair\evA\evB$
and~$\pair\evA\evC$ do not belong to~$\transC\ppo$,
therefore~$\tuple\evA{\evC,\evD,\evE}$ does
not violate~\axiomExternal.

\item The behavior is allowed by~\rcelevenext.
Two conditions could potentially be violated by the
cycle~$\tuple\evA{\evC,\evD,\evE}$:
\axiomCoherenceI and \axiomCoherenceII.
The cycle does not violate \axiomCoherenceII,
because~$\pair\evA\evC\notin\transC\ppoasm$.
The cycle does not violate \axiomCoherenceI,
because~$\pair\evE\evA$ is the longest~$\eco$ edge
starting from~$\evE$,
and because~$\pair\evA\evB\notin\porceleven$,
so extending the~$\hb$ edge~$\pair\evB\evE$ with~%
$\eco$ does not close the cycle.

\item The behavior breaks the irreflexivity of~$\hb\seq\eco_\rceleven$,
even when the \rcelevenext definition of~$\hb$ is used.
Indeed,
both the edges~$\pair\evB\evC$ and~$\pair\evD\evE$
belong to~$\porceleven$,
and~$\pair\evA\evB\in\rfi\subseteq\eco_\rceleven$,
so~$\pair\evB\evB$ forms a reflexive edge in~%
$\hb\seq\eco_\rceleven$.
Therefore, a naive extension of~\rceleven that
keeps~$\eco_\rceleven$ would be unsound.
\end{enumerate}

%% file: figure-cpp-syntax.tex
\begin{figure}[t]

\raggedright

\begin{minipage}{.75\textwidth}
\judgmentpar{Syntax of expressions, commands, and access modes}
\[\begin{array}{@{}r@{\;}r@{\;}l@{}}
  \typeExpr\ni\expr & ::=
               & \intn\;(\in\NN)
          \mid   \reg\;(\in\typeReg)
          \mid   \loc\;(\in\typeLoc\triangleq\NN)
\\
         &\mid&  \Plus\expr\expr
          \mid   \Sub\expr\expr
          \mid   \Times\expr\expr

\\[.5mm]

  \typeCmd\ni\cmd & ::=
               & \ReadPL\md\reg\expr
         \mid    \WritePL\md\expr\expr
         \mid    \RMWPL\md\reg\expr\expr\expr
\\
        &\mid&   \FencePL\md
         \mid    \IfThen\expr\cmd
         \mid    \While\expr\cmd
         \mid    \Seq\cmd\cmd
         \mid    \Skip

\\[.5mm]

  \typeMode\ni\md & ::=
            & \na
         \mid \rlx
         \mid \rel
         \mid \acq
         \mid \acqrel
         \mid \sc
\end{array}\]
\end{minipage}
\begin{minipage}{.15\textwidth}
\centering
\hspace{-1.1cm}
\includegraphics{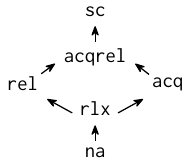}
\end{minipage}
\vspace{-10pt}
\caption{Syntax of \rcelevenlang.}
\vspace{-10pt}
\Description{}
\label{fig:cpp:syntax}
\end{figure}

%% file: figure-cpp-op-semantics.tex
\begin{figure}[t]

\raggedright

\judgment[Pool reduction]{\poolstep\pool\graph\pool\graph}

\begin{mathpar}
    \inferrule[\nameReadStep]{
      \poolLkp>*>*>*[\Seq{\ReadPL\md\reg\expr}\cmd]<
        \and
      \poolB =
        \poolA\left[\begin{array}{@{}l@{}}
          \inti :=
          \poolLkp[:=]
            >*[\phi[\reg:=\intn]]     %
            >*[\intj+1]               %
            >*<                       %
        \end{array}\right]
      \\\\
      {\begin{minipage}{.3\textwidth}\begin{mathpar}
           \loc=\interp\expr_\phi\and\evA=\pair\inti\intj
       \end{mathpar}\end{minipage}}
        \and
        \graphB =
          \graphUpt*{\{\evA\}}                   %
                   *{[\evA:=\Read\md\loc\intn]}< %
                   {}                            %
    }{
      \poolstep\poolA\graphA\poolB\graphB
    }
\\
    \inferrule[\nameTerminateStep]{
      \poolA[\inti].\nextcmd = \Skip
        \and
      \poolB = 
        \lambda\intj\in\dom\poolA\setminus\{\inti\}.\;\poolA[\intj]
    }{
      \poolstep\poolA\graphA\poolB\graphA
    }
\end{mathpar}

\caption{Definition of pool reduction.}
\Description{}
\label{fig:pool:reduction}
\end{figure}

%% file: rc11-consistency.tex
An execution graph~$\tuple\graphA{\rf,\mo}$
is \emph{\rceleven-consistent} if the conditions
\begin{itemize}
\item \axiomCoherenceDef\hfill\textnormal{(\RefTirName\nameCoherence)}
\item \axiomSCDef\hfill\textnormal{(\RefTirName\nameSC)}
\item \axiomAtomicityDef\hfill\textnormal{(\RefTirName\nameAtomicity)}
\item \axiomNTADef\hfill\textnormal{(\RefTirName\nameNTA)}
\end{itemize}
hold, where the relations
\emph{happens-before}~($\hb$),
\emph{synchronizes-with}~($\sw$),
\emph{extended coherence order}~($\eco$),
\emph{reads-before}~($\rb$),
\emph{partial-SC}~($\psc$), and
\emph{SC-before}~($\scb$)
are defined as follows:

\[\begin{array}{@{}l@{\hspace{1cm}}r@{}}
\begin{array}{@{}r@{\;}c@{\;\;}l@{}}
  \rb &\eqdef&
    (\inv\rf\seq\mo)\setminus\settorel\EventE
\\[.4mm]
  \hb &\eqdef&
    \transC{(\union\po\sw)}
\\[.4mm]
  \eco &\eqdef&
    \transC{(\union\rf{\union\mo\rb})}
\end{array}
&
\begin{array}{@{}l@{}}
  \sw\eqdef
    \left\{\begin{array}{@{}l@{}}
    \settorel{\EventAtLeast\rel}\seq
     \refC{(
       \settorel\FenceE\seq
       \po
     )}\seq
    \settorel{\WriteAtLeast\rlx}\seq\\
    \quad\transC\rf\seq\\
    \settorel{\ReadAtLeast\rlx}\seq
    \refC{(\po\seq{\settorel\FenceE})}\seq
    \settorel{\EventAtLeast\acq}
    \end{array}\right.
\end{array}
\end{array}\]

\begin{minipage}{.2\textwidth}
\[\begin{array}{r@{\;}c@{\;\;}l}
  \psc &\eqdef& \union\pscBase\pscFence
\end{array}\]
\end{minipage}
\begin{minipage}{.8\textwidth}
\[\begin{array}{r@{\;}c@{\;\;}l}
  \scb &\eqdef&
    \union\po{
    \union{{\neqloc\po}\seq{\hb\seq{\neqloc\po}}}{
    \union{\perloc\hb}{
    \union\mo\rb
    }}}
\\[.4mm]
  \pscBase &\eqdef&
    {(
      \union{\settorel{\EventMd\sc}}
            {{\settorel{\FenceMd\sc}}\seq{\refC\hb}}
    )}\seq{
      \scb\seq{(
        \union{\settorel{\EventMd\sc}}
              {{\refC\hb}\seq{\settorel{\FenceMd\sc}}}
      )}
    }
\\[.4mm]
  \pscFence &\eqdef&
    {\settorel{\FenceMd\sc}}\seq{
      {
        (\union\hb{{\hb\seq\eco}\seq\hb})
      }\seq{
        \settorel{\FenceMd\sc}
      }}
\end{array}\]
\end{minipage}

%% file: figure-cpp-intelext-diagram.tex
\begin{figure}[t]
\small
\[\begin{array}{c}
\includegraphics{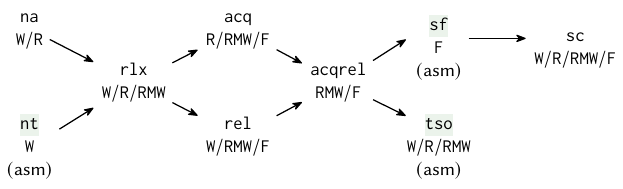}
\end{array}\]
\caption{Diagram of access modes of \rcelevenext.}
\Description{}
\label{fig:cpp:x86:diagram}
\end{figure}

%% file: rc11-ex86-consistent.tex
An execution graph $\tuple\graphA{\rf,\mo}$ is
\emph{\rcelevenext-consistent}
if, in addition to the conditions from~\defref{def:rceleven:consistent}
(where~\axiomCoherence is renamed to~\RefTirName\nameCoherenceI),
the conditions
\begin{itemize}
\item \greenhl{\axiomCoherenceIIDef}\hfill
      \greenhl{\textnormal{(\RefTirName\nameCoherenceII)}}
\item \greenhl{\axiomCoherenceIIIDef}\hfill
      \greenhl{\textnormal{(\RefTirName\nameCoherenceIII)}}
\end{itemize}
hold, where the relations
$\hb$,
$\eco$,
$\porceleven$,
and~$\ppoasm$ are defined as follows:

\begin{mathpar}
  \hb \triangleq
    \transC{(\mathhl\porceleven\disj\sw)}
\and
  \eco \triangleq
    \transC{(
      \mathhl\rfe\disj\mo\disj\rb
    )}
\\
\begin{array}{r@{}l}
  \mathhl{\porceleven}
         &\mathhl{\;\triangleq
    \settorel{\EventE\setminus\WriteMd\nt}\seq
    \po}
\\
         &\mathhl{\;\cup\;
    \po\seq
    \settorel{\RMWMd\tso\cup\FenceAtLeast\stf}}
\\
         &\mathhl{\;\cup\;
    \perloc\po\seq\settorel\WriteE}
\end{array}
\and
\begin{array}{r@{}l}
  \mathhl{\ppoasm}
         &\mathhl{\;\triangleq
    \po\seq\settorel{\RMWMd\tso\cup\FenceAtLeast\stf}}
\\
         &\mathhl{\;\cup\;
    \settorel{\ReadMd\tso\cup\RMWMd\tso\cup\FenceMd\sc}\seq\po}
\\

         &\mathhl{\;\cup\;
    \settorel{\FenceAtLeast\stf}\seq\po\seq\settorel{\EventE\setminus\ReadE}}
\\
         &\mathhl{\;\cup\;
    \settorel{\WriteMd\tso}\seq\po\seq\settorel{\EventE\setminus\ReadE\setminus\WriteMd\nt}}
\\
         &\mathhl{\;\cup\;
    \settorel{\EventE\setminus\ReadE\setminus\WriteMd\nt}\seq\po\seq
    \settorel{\WriteMd\tso}}
\end{array}
\end{mathpar}

%% file: metatheory.tex
\section{Metatheory}
\label{section:metatheory}

In this section,
we study properties of~\rcelevenext.
In particular,
we study the correctness of compilation,
the correctness of compiler optimizations,
and the \emph{data-race-freedom} property:
the property that,
if a program~$\prog$ has races only on \sc accesses,
then~$\prog$ can exhibit only sequentially consistent behaviors.
Data-race freedom is one of the main design goals of~\rceleven,
so it is important to show that~\rcelevenext preserves this property.

The discussion is organized as follows.
In~\cref{subsection:schemes},
we define two compilation schemes to \intelext.
In~\cref{subsection:mixed:graphs},
we introduce the notion of mixed execution graphs,
a key concept in %
our proofs of compilation correctness,
whose sketch we present in~\cref{subsection:proof:sketch}.
In~\cref{subsection:optimizations},
we discuss our results of compiler-optimization correctness.
Finally,
in~\cref{subsection:drf},
we present the formal statement of data-race freedom.
The property that~\rcelevenext is
an extension of \rceleven and \intelext
is in the
Appendix~\appendixref{thm:extension:I:rcelevenext:app,thm:extension:II:rcelevenext:app}.

\subsection{Compilation Schemes -- Definition and Correctness}
\label{subsection:schemes}

Following the traditional approach in the weak-memory literature,
we formalize the notion of compilation as a
\textit{compilation scheme}.
Roughly speaking, a compilation scheme is a
program transformation that modifies only memory instructions:
the main structure of the program,
including control flow and the distribution of threads,
is kept,
whereas memory instructions from the source language are
mapped to zero, one, or multiple instructions from
the target language.
Therefore, this approach allows us to concentrate on
how the transition from the model of the source language
to the model of the target language affects
the way in which the program interacts with memory.
Intuitively, the compilation scheme is correct
if the execution of the transformed program can
update memory only to a subset of the final states
reachable from the execution of the source program.

\begin{definition}[Compilation Scheme from \rcelevenextlang to \intelextlang]
\label{def:scheme:rcelevenext}
\[\begin{array}{@{}c@{~~~~}c@{}}
\begin{array}{@{}r@{}c@{\;}l@{}}
\compile{\WritePL\sc\exprA\exprB}
  &\eqdef&
    \Seq{\WritePL{}\exprA\exprB}{\MFencePL}
\\
\compile{\WritePL{\neq\sc}\exprA\exprB}
  &\eqdef&
    \WritePL{}\exprA\exprB
\\
\compile{\ReadPL\md\reg\expr}
  &\eqdef&
    \ReadPL{}\reg\expr
\\
\compile{\Seq\cmdA\cmdB}
  &\eqdef&
    \Seq{\compile\cmdA}{\compile\cmdB}
\\
\compile{\Skip}
  &\eqdef&
    \Skip
\\
\compile{\Asm[\cmd]}
  &\eqdef&
    \cmd
\end{array}
&
\begin{array}{@{}r@{}c@{\;}l@{}}
\compile{\FencePL\sc}
  &\eqdef&
    \MFencePL
\\
\compile{\FencePL{\neq\sc}}
  &\eqdef&
    \Skip
\\
\compile{\RMWPL\md\reg\exprOne\exprTwo\exprThree}
  &\eqdef&
    \RMWPL{}\reg\exprOne\exprTwo\exprThree
\\
\compile{\While\expr\cmd}
  &\eqdef&
    \While\expr{\compile\cmd}
\\
\compile{\IfThen\expr\cmd}
  &\eqdef&
    \IfThen\expr{\compile\cmd}
\\
&&
\end{array}
\end{array}\]
\end{definition}

\begin{definition}[Alternative Compilation Scheme]
\label{def:alt:scheme:rcelevenext}
Same as Def.~\ref{def:scheme:rcelevenext} except for the following cases:
\[\begin{array}{cc}
\begin{array}{r@{}c@{\;}l}
\compilealt{\WritePL{\sc}\exprA\exprB}
  &\eqdef&
    \Seq\SFencePL{
    \Seq{\WritePL{}\exprA\exprB}
        \MFencePL
    }
\\
\compilealt{\WritePL{\rel}\exprA\exprB}
  &\eqdef&
    \Seq\SFencePL{\WritePL{}\exprA\exprB}
\end{array}
&
\begin{array}{r@{}c@{\;}l}
\compilealt{\WritePL\rlx\exprA\exprB}
  &\eqdef&
    \NTWritePL\exprA\exprB
\\
\compilealt{\FencePL{\rel,\,\acqrel}}
  &\eqdef&
    \SFencePL
\end{array}
\end{array}\]
\end{definition}

\defref{def:scheme:rcelevenext}
follows largely the scheme from~\citet{rc11}.
Perhaps more striking is
\defref{def:alt:scheme:rcelevenext},
which provides an alternative scheme for \intelext,
where relaxed writes can be compiled to non-temporal stores.
The price to pay is the addition of store fences
to the compilation of \rel/\sc writes
and \rel/\acqrel fences.
The idea is to ensure that
every $\sw$ edge starts with a store fence.
In this way,
non-temporal stores,
even when emitted from the compilation of \rlx writes,
cannot invalidate release-acquire synchronization.

In a similar way to how we constructed the function~$\interp[\rceleven]\_$,
which defines the semantics of \rceleven programs,
and to how we implicitly constructed~$\interp[\rcelevenext]\_$,
we can introduce the function~$\interp[\intelext]\_$
defining the semantics of \intelextlang programs.
The definition is in the
Appendix~\appendixref{def:intelextlang:semantics:app}.
The statement of compilation correctness
is then straightforward:
\begin{theorem}
\label{thm:correctness:rcelevenext}
[Correctness of Definitions \ref{def:scheme:rcelevenext} and \ref{def:alt:scheme:rcelevenext}]
For every~\rcelevenextlang program~$\prog$,
the set of final states of~$\compile\prog$
defined by \intelext
is included in the set of final states of $\prog$
defined by \rcelevenext:
\[
\forall\prog.\;
\interp[\intelext]{\compile\prog}
\subseteq
\interp[\rcelevenext]{\prog}
\]
\end{theorem}

\subsection{Mixed Execution Graphs}
\label{subsection:mixed:graphs}

Our proofs of compilation correctness
rely on the novel notion of \textit{mixed execution graphs},
a type of execution graph whose nodes
contain events from both the source-level and target-level models.
Before presenting the proof sketch of our compilation-correctness results,
let us give a brief introduction to mixed execution graphs.

Informally speaking,
a mixed execution graph is the
superposition of two execution graphs:
one called \textit{source graph},
which is associated with a source program~$\prog$;
and
one called \textit{target graph},
which is associated with the compilation of~$\prog$.
The key feature of a mixed execution graph is that
it captures the fact that source and target graphs
share the same overall structure.
Indeed,
because a compilation scheme preserves the
control flow of the source program and
changes only how memory operations are mapped
to operations in the target language,
for every execution graph of the compiled program,
one can always construct an execution graph of the
source program that preserves much of the structure
of the target graph,
including its primitive relations~$\po$,~$\rf$, and~$\mo$.
The only mismatches between these graphs come from how
one memory operation from the source language might be
mapped to
zero, one, or multiple memory operations from the target language.

To account for these mismatches,
nodes in a mixed graph,
called \textit{mixed nodes},
carry events from both source and target models.
Events from the two models however cannot be
arbitrarily assembled in a mixed node:
the source-level events in a mixed node correspond
to the events of a single source instruction
and the target-level events correspond to the
events emitted by the snippet of target-level language
produced by the mapping of this instruction.
Therefore,
the range of mixed nodes is fixed and
determined by the underlying compilation scheme.

Mixed graphs form a very convenient tool for proving
compilation-correctness results
because they allow
one to work with the execution graphs from both the source
program and its compiled version at the same time, and
because they allow one to forget about the compilation
scheme which is ultimately encoded in the set of permissible
mixed nodes.
Moreover, it is possible to lift the consistency conditions from
the models of source and target languages to this mixed-graph structure.
Both models can thus be defined on the same structure,
thereby allowing one to formally reason about statements of the kind
``\textit{one model is stronger than the other}''.
In fact, the main convenience of mixed execution graphs is
precisely to allow one to formulate the compilation correctness result
as a statement in this fashion:
``\textit{in a mixed execution graph with nodes taken from a well-chosen set,
if the consistency conditions of the target model hold,
then so do the consistency conditions of the source model}''.
The set of nodes has to be well chosen so as to correctly reflect
the compilation scheme begin considered.

\input{figure-mixed-graph-example}

To give an illustration of mixed execution graphs,
let us consider the example
depicted in~\fref{fig:mixed:graph:example}.
We refer the reader to the
Appendix~\appendixref{subsection:mixed:execution:graphs:app}
for a complete exposition of mixed execution graphs
and for a more thorough explanation of this example.
The nodes are depicted as domino-shaped boxes
where the first part contains \rcelevenext events
and the second part contains \intelext events.
There are two types of nodes in this example:
one captures how a \sc write is compiled to a
plain write followed by a memory fence;
the other one captures
how a \sc read is compiled to a plain read.
In this simple example, it is easy to see how
a \rcelevenext graph~$\graphA$ and a
a \intelext graph~$\graphB$ can be recovered from
the mixed structure.
We wish to argue that the behavior represented by the mixed graph
is disallowed in~$\graphA$
because all access modes are \sc.
In other words,
we wish to argue that~$\graphA$ is inconsistent.
If compilation is correct,
then~$\graphB$ should also be inconsistent.
Thanks to the mixed graph structure,
we can carry out both proofs in the same graph:
$\graphA$ is inconsistent because the
cycle~$\tuple\evA{\evB,\evC,\evD}$
contradicts~\axiomSC,
and
$\graphB$ is inconsistent because the same
cycle contradicts~\axiomExternal.

\subsection{Compilation Correctness - Proof Sketch}
\label{subsection:proof:sketch}

The overall structure of our proofs
is depicted by the following diagram:
\begin{center}
\includegraphics{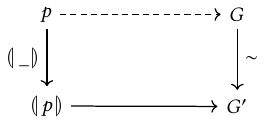}
\end{center}

It illustrates the first step
of a two-steps strategy to prove that~$\compile\_$ is correct.

This first step consists of showing that,
for every program~$\prog$,
for every execution graph~$\graphB$
associated with~$\compile\prog$,
there exists a graph~$\graphA$ associated with~$\prog$,
such that~$\graphA$ \textit{is simulated by}~$\graphB$~%
\appendixref{def:gsim:app},
noted~$\graphA\gsim\graphB$,
which means that~$\graphA$ and~$\graphB$
can be merged into a mixed graph~$\graphM$.
This first step is accomplished by induction over
the construction of the graph~$\graphB$.
Intuitively,
because the compiled program~$\compile\prog$
preserves much of the structure of~$\prog$,
it is possible to replay the pool-reduction steps
from~$\compile\prog$ and yield a graph~$\graphA$
that satisfies the desired properties.

The second step is then to show that,
if~$\graphM$ is \intelext-consistent,
then it is \rcelevenext-consistent,
for notions of \intelext-consistency
and \rcelevenext-consistency
adapted to mixed graphs~%
\appendixref{def:mixed:intelext:consistent:app,%
def:mixed:rcelevenext:consistent:app}.
The consistency of a mixed graph
holds iff each of its constituent graphs
is consistent,
a property we call \textit{Transfer Principle}~%
\appendixref{thm:transfer:app}.
It follows from this principle that
the second step is equivalent to the
proof that,
if~$\graphB$ is \intelext-consistent,
then~$\graphA$ is \rcelevenext-consistent.
This is sufficient to conclude the proof.

\subsection{Compiler Optimizations}
\label{subsection:optimizations}

We now study the compiler optimizations discussed by~\citet{rc11}.
We wish to determine under which conditions they are sound in \rcelevenext.
As previously stated,
our model validates all thread-local optimizations.
The only optimization that is only valid under additional conditions
is sequentialization, which is a global transformation.

Following~\citet{rc11},
we formalize a compiler optimization as a \textit{program transformation}:
a mapping that takes and produces programs in the source language,
which, in our case, is the language~\rcelevenextlang.
When discussing a given transformation,
we use the notation~$\transf\progA\progB$ to express that
$\progB$ can be obtained by applying the transformation to~$\progA$.

A program transformation is sound,
if applying this transformation does not introduce new behaviors.
Formally speaking,
this means that,
if~$\transf\progA\progB$ holds,
then the set of behaviors of~$\progB$
is a subset of the set of behaviors of~$\progA$,
that is,~$\interp\progB\subseteq\interp\progA$.

To prove the soundness of a program transformation,
we usually resort to its natural generalization
to the level of execution graphs:
a transformation that applies to events in
an execution graph rather than to instructions.
In the transformations considered here,
this generalization is straightforward.
We use the notation~$\transf\graphA\graphB$
to express that~$\graphB$ can be obtained by
applying the transformation to~$\graphA$.
The property that allows us to shift our attention to the graph
transformation when proving soundness of a program
transformation is the following:
if~$\transf\progA\progB$,
and if~$\graphB$ is an execution graph associated
with~$\progB$, then there exists an execution
graph~$\graphA$ associated with~$\progA$
such that~$\transf\graphA\graphB$.
Under this property,
to show the soundness of the program transformation,
it suffices to show (1) that,
if~$\graphB$ is \rcelevenext-consistent,
then so is~$\graphA$;
and (2) that,
if~$\graphB$ is racy,
then so is~$\graphA$.

\subsubsection{Register Promotion.}

Register promotion replaces accesses
to a memory location with accesses to a register,
provided that this location is accessed by only one thread
and that this location is not
accessed via an inline-assembly read-modify-write.
At the level of execution graphs,
the transformation~$\transf\graphA\graphB$
removes all the accesses to a location~$\locx$ in~$\graphA$,
provided that these accesses are related by~$\graphA.\po$
and that their intersection with~$\RMWMd\tso$ is empty.
Avoiding~$\RMWMd\tso$ is necessary,
because RMWs act as barriers in~\intel.
Intuitively,
this transformation is correct because
a consistency-violating cycle in~$\graphA$ involving more than one thread
must not contain accesses to~$\locx$
(because~$\locx$ is never shared between two threads),
so such a cycle would still exist in~$\graphB$.

\subsubsection{Strengthening.}
Strengthening replaces an access mode
with a stronger one with respect to the ordering of access modes
(\fref{fig:cpp:x86:diagram}).
Definitions in \rcelevenext are \textit{monotonic}:
only upward-closed ranges of the form ``$\sqsupseteq\md$'' occur.
\footnote{%
Sets of the form~${S}^{\md}$, for~$\md\in\{\sc,\tso\}$,
can be rewritten as~${S}^{\sqsupseteq\md}$,
and sets of the form~${S}\setminus\WriteMd\nt$
can be rewritten as~%
$({S}\setminus\WriteE)\cup({S}\cap\WriteAtLeast\na)$.
}
The correctness of this transformation is thus trivial,
because, every edge of the original graph is preserved.

\subsubsection{Deordering and Merging.}
Deordering transforms sequential composition
into parallel composition:
$\transf{\Seq\cmdA\cmdB}{\inarrII\cmdA\cmdB}$.
Merging transforms two consecutive instructions into one:
$\transf{\Seq\cmdA\cmdB}{\cmd''}$.
\citet[Table 1 and Figure 11]{rc11}
defines the pairs of deorderable instructions
and mergeable instructions permitted in \rceleven.
Both transformations remain valid in~\rcelevenext
when restricted to the same deorderable
and mergeable pairs of instructions.
Intuitively,
the correctness argument relies on the remark
that these transformations have no effect on~$\ppoasm$.
Therefore,
the additional \axiomCoherenceII condition of our extended model
does not pose a risk to the correctness of these optimizations,
because cycles in~$\ppoasm\cup\eco$ cannot be undone
by deordering and merging.

\subsubsection{Sequentialization.}

\begin{figure}[t]
\[\begin{array}{@{}c@{}}
\begin{array}{@{}ccc@{}}
\inarrIII{
  \AsmNTWritePL\locx{1}
}{
  \ReadPL\rlx\regA\locx\outcomeI{1} \\
  \WritePL\rel\locy{1}
}{
  \ReadPL\acq\regB\locy\outcomeI{1} \\
  \ReadPL\rlx\regC\locx\outcomeI{0} \\
}
&
\rightsquigarrow
&
\inarrII{
  \AsmNTWritePL\locx{1}             \\
  \ReadPL\rlx\regA\locx\outcomeC{1} \\
  \WritePL\rel\locy{1}
}{
  \ReadPL\acq\regB\locy\outcomeC{1} \\
  \ReadPL\rlx\regC\locx\outcomeC{0} \\
}
\end{array}
\\[7mm]
\begin{array}{@{}lr@{\hspace{4.2cm}}}
\inarrIII{
  \AsmWritePL\locx{1}
}{
  \AsmReadPL\regA\locx\outcomeI{1} \\
  \AsmReadPL\regB\locy\outcomeI{0} \\
}{
  \AsmWritePL\locy{1}              \\
  \AsmMFencePL                     \\
  \AsmReadPL\regC\locx\outcomeI{0} \\
}
&
\rightsquigarrow
\end{array}
\\[6mm]
\begin{array}{@{\hspace{7.5cm}}r@{}}
\inarrII{
  \AsmWritePL\locx{1}              \\
  \AsmReadPL\regA\locx\outcomeC{1} \\
  \AsmReadPL\regB\locy\outcomeC{0} \\
}{
  \AsmWritePL\locy{1}              \\
  \AsmMFencePL                     \\
  \AsmReadPL\regC\locx\outcomeC{0} \\
}
\end{array}
\end{array}\]
\caption{Counterexamples showing the unsoundness of sequentialization in \rcelevenext.}
\Description{}
\label{fig:sequentialization:counterexample}
\end{figure}

Sequentialization merges two threads into one
by interleaving their instructions.
\Cref{fig:sequentialization:counterexample}
depicts two counterexamples showing the
unsoundness of sequentialization in~\rcelevenext.

Sequentialization is unsound because, when merging two threads,
an external~$\rf$ edge might become internal.
Because internal~$\rf$ edges are not included
in~$\porceleven$, in~$\ppoasm$, or in~$\eco$,
exchanging a~$\rfe$ edge for a~$\rfi$ edge might undo
cycles in~$\ppoasm\disj\eco$, in~$\hb\seq\refC\eco$, or in~$\psc$.

The omission of~$\rfi$ edges from~$\porceleven$,~$\ppoasm$, and~$\eco$,
is necessary because non-temporal stores break
release-acquire synchronization.
Moreover,
the omission of~$\rfi$ in the statement of \axiomCoherenceII
is inherited from \intelext, which also omits~$\rfi$ edges
in the statement of \axiomExternal.
For this reason, sequentialization is also
unsound in plain \intelext~\cite{promising-semantics}.

Because sequentialization is unsound in \intelext,
its support is incompatible with
\hyperlink{P4}{Property P4}.
If we ignore \hyperlink{P4}{Property P4},
then there are two approaches to add support for sequentialization:
(1) to relax the model so as to allow the behavior of the programs on the
left-hand side of \cref{fig:sequentialization:counterexample},
or
(2) to make the model stronger than, or incomparable to,
\rcelevenext so as to
disallow the behavior of the programs on the
right-hand side of \cref{fig:sequentialization:counterexample}.
The first approach leads to a lost of reasoning principles,
whereas the second approach
invalidates the straightforward identity map as
a sound compilation scheme for inline assembly
(\hyperlink{P1}{Property P1}).
Therefore,
instead of aiming to support sequentialization
for the price of abandoning
\hyperlink{P4}{Property P4},
we investigate conditions under which
sequentialization is sound in \rcelevenext as is.

We call a~$\rfe$ edge
of a \rcelevenext-inconsistent graph~$\graph$
\textit{problematic}
if sequentialization transforms~$\graph$
into a \rcelevenext-consistent graph~$\graphB$.
We note that a problematic edge must contain at least
one inline-assembly event.
Indeed,
because~%
$\settorel{\EventE\setminus\WriteMd\nt}\seq\rfi$
is included in~$\porceleven$,
transforming a~$\graphA.\rfe$ edge between plain \rceleven events
into~$\graphB.\rfi$ makes this edge part of~$\graphB.\hb$,
so it cannot undo a cycle in~$\hb\seq\eco$.
Such a transformation cannot undo a cycle in~$\ppoasm\disj\eco$ either,
because,
by definition of~$\ppoasm$,
every edge~%
$\pair\evA\evB\in\graph.\rfe\seq\settorel{\ReadMd{\neq\tso}}$
that is part of a cycle in~$\ppoasm\disj\eco$
must be followed by an edge~%
$\pair\evB\evC\in\po\seq\settorel{\RMWMd\tso\cup\FenceAtLeast\stf}$,
therefore~$\pair\evA\evC\in\graphB.\ppoasm$.

When we consider two threads,
a sufficient purely syntactic condition
to rule out the existence of such problematic~$\rf$ edges
is the following:~%
(1) if one thread includes
plain \rceleven reads then the addresses of all these accesses
and the addresses of all locations modified by the other thread
using inline assembly should be statically known and disjoint,
and~%
(2) if one thread includes
inline-assembly reads then the addresses of all these accesses
and the addresses of all locations modified by the other thread
(using inline assembly or not) should be statically known and disjoint.
We call this condition
\textit{No Interaction Through Inline Assembly} (NITIA).
Notice that,
thanks to the inclusions~%
$\settorel{\RMWE}\seq\rfi\subseteq\porceleven\cap\ppoasm$
and~%
$\rfi\seq\settorel{\RMWE}\subseteq\porceleven\cap\ppoasm$,
read-modify-writes can be ignored when
checking the NITIA condition.
Refining the statement of sequentialization to
require this condition to hold when merging two threads
leads to a sound optimization.
We prove this claim in the
Appendix~\appendixref{thm:nitia:sequentialization:app}.

Another possible refinement of sequentialization is
to add a~$\sc$ fence between the threads to be merged.
Inserting such a fence
imposes the constraint that the instructions from
one thread are ordered with respect to
the instructions from the other thread.
In retrospect,
with the NITIA-refinement of sequentialization,
threads can be arbitrarily interleaved.
We prove soundness of this second version of sequentialization
in the Appendix~\appendixref{thm:fence:sequentialization:app}.

\subsection{Data-Race Freedom}
\label{subsection:drf}
Informally stated,
the data-race-freedom property posits that,
if a program~$\prog$ has races only on \sc accesses,
then~$\prog$ can exhibit only sequentially consistent behaviors.
This property enforces the reasoning principle that,
to recover the relative simplicity of sequential consistency,
it suffices to show the absence of races on non-\sc accesses.

Because the notion of a race,
as introduced in \defref{def:data:race},
applies to execution graphs,
not to programs,
to formalize this statement
we must define what it means for
a program \emph{to have races only on \sc accesses},
that is, to be \emph{data-race free}:
\begin{definition}[Data-Race Free]
A program~$\prog$
\emph{has races only on \sc accesses},
or, is \emph{data-race free},
if
every \SC-consistent execution graph~$\graph$
associated with~$\prog$
has races only on~\sc accesses:
\[
\prog\;\textit{is data-race free}\;\iff\;
\begin{array}{@{\,}l@{}}
\forall\,\graph,\,\mo,\,\rf,\,\evA,\,\evB.
\\
\quad\left(\begin{array}{@{\,}l@{\,}}
    \poolstep*{\toPool\prog}\initGraph\_\graphA
  \\
    \tuple\graphA{\rf,\mo}\;\textit{is \SC-consistent}
  \\
    \pair\evA\evB\;\textit{forms a data race}
\end{array}\right)
\implies
\dotMode\evA = \dotMode\evB = \sc
\end{array}
\]
\end{definition}

The definition relies on the notion of \SC-consistency,
captured by a single condition:
the acyclicity of~$\po\disj\rf\disj\mo\disj\rb$.
The restriction to \SC-consistent graphs
strengthens the reasoning principle enforced by data-race freedom.
If, for example,
the graphs were assumed to be \rcelevenext-consistent,
then the resulting property would offer no benefit over
\rcelevenext itself.

Finally, data-race freedom is formally stated as follows:
\begin{theorem}[Data-Race Freedom]
\label{thm:drf}
$\forall\,\prog.\;
\prog\;\text{is \textit{data-race free}}\implies
\interp[\rcelevenext]\prog = \interp[\SC]\prog$
\end{theorem}
A detailed proof of this theorem can
be found in the Appendix~\appendixref{subsection:drf:app}.

%% file: figure-mixed-graph-example.tex
\begin{figure}[t]
\includegraphics{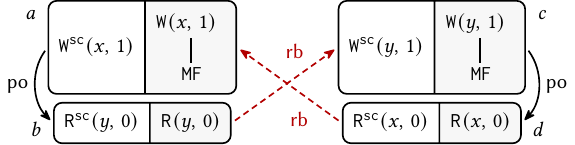}
\caption{Example of a mixed execution graph.}
\Description{}
\label{fig:mixed:graph:example}
\end{figure}

%% file: related.tex
\section{Related Work}
\label{section:related}

To the extent of our knowledge,
we are the first authors to consider the problem of
extending \cpp's memory model with support for inline assembly.
In the following paragraphs,
we discuss related work on topics that we covered in this paper.

\paragraph{Models of \intel.}

\citet{sewell-al-10}
introduce an operational model of \intel
that,
according to the documented tests,
agrees with the behavior of actual \intel machines
and is proven to be equivalent to the axiomatic formulation
of total store order (TSO)~\cite{sindhu-92}.
Such a model is devoid of the ambiguity
that is often present in
the documentation of multiprocessors
written in informal prose.
An interesting application of the model is to explain
the correctness of an optimization that was the
subject of a famous discussion in the~\citet{spinunlock-opt}.
In this paper,
we rely on \citet{raad-22}'s \intelext,
an extension of \intel with support for
(1) non-temporal stores,
(2) store fences,
and (3) reads and writes to the full range of Intel's
\textit{memory types} (\textit{uncacheable},
\textit{write-combined}, and \textit{write-through}).
More specifically,
we rely on the  axiomatic formulation of \intelext,
which formulation is included in the
Appendix~\appendixref{def:intelext:consistent:app}.

\paragraph{Models of \cpp.}

\citet{c11mm} introduce the first formal memory model of \cpp
as a formalization of the \cpp standard~\cite{cpp11}
mechanized in Isabelle/HOL~\cite{isabelle}.
\citet{rc11} however identify several issues with this model.
They introduce \rceleven (for \textit{Repaired C11})
in an attempt to repair these flaws.
Indeed, \citet{rc11} identify at least four problems with the
original model of~\citet{c11mm}:
(1) the proposed compilation
schemes~\cite{batty-al-12,sarkar-al-12} to \power is unsound;
(2) the semantics of \sc fences is too weak,
the authors show that placing \sc fences between every memory
access is not sufficient to enforce only sequentially consistent behaviors,
and they argue that \sc fences are not \textit{cumulative};
(3) \textit{out-of-thin-air} behaviors are allowed
even though they cannot be observed in any actual hardware;
(4) the model lacks \emph{monotonicity}~\cite{vafeiadis-et-al-15}.
The \rceleven model fixes these issues with the Axiom~\axiomSC,
which weakens the semantics of programs mixing \sc and non-\sc accesses
so that the compilation schemes to \power are sound and which strengthens the
semantics of \sc fences;
and with the Axiom~\axiomNTA,
which disallows out-of-thin-air behaviors.
The latter axiom has the undesired effect
of also disallowing \textit{load buffering} behaviors,
which can be observed in actual hardware.

\paragraph{Multi-language semantics.}

Devising a model for \cpp with inline assembly
can be framed as a problem of combining the semantics of two
different languages: \cpp and the assembly language of the
underlying hardware architecture.
We identify some works that propose general solutions to the
problem of specifying \textit{multi-language semantics}.
\citet{sammler-al-23}
introduce DimSum, a generic framework to reason about
programs written in different languages.
Inspired by process calculi,
one of the key ideas is to consider the semantics
of a program as a labeled transition system
where nodes represent the (global) state and transitions are labeled by events.
The semantics of a program is written as a refinement statement
that accounts for both \textit{demonic non-determinism}
(the usual flavor of non-determinism)
and \textit{angelic non-determinism}~\cite{floyd-nondet-67},
which is motivated by situations where
the representation of a value in one language
matches the representation of multiple different values in
another language.
This framework is inadequate for our purposes because,
as it stands, it is limited to sequential languages.
Moreover, there is also a difference in the nature of our works:
whereas~\citet{sammler-al-23} concentrate on a general framework
to define the semantics programs written in different languages,
with special attention on how the memory representation differs
in each of these languages,
our focus is rather to underpin the exact (consistency)
semantics of programs combining two specific languages, \cpp and assembly.
\citet{goens-al-23} study the question of devising memory models for
\textit{heterogeneous processors}, processors that mix
CPUs and GPUs and allow them to share memory.
Their contribution is the introduction of the notion of a
\textit{compound memory model},
a way to combine the different memory models from each of
devices sharing memory.
As the authors put it,
``\textit{a compound memory model is not a new memory model}'',
in the sense that
threads from devices abiding by different memory models
continue to adhere to these models.
This is in contrast with our work,
where (1) our extended model constitutes a new model and
(2) single threads can mix accesses
from two different models, \rceleven and \intelext.

\paragraph{Compilation-correctness proofs.}

\citet{rc11} prove the correctness of compilation schemes
from \rceleven to several architectures (\intel, \power, and \arm*).
\citet{imm-19}
introduce the idea of an \textit{intermediate memory model} (IMM),
a model to which high-level languages, such as \cpp,
can be mapped and from which low-level code can be produced
according to compilation schemes proven correct once and for all.
The authors argue that IMM is useful for structuring proofs
of correctness compilation,
because, for example,
in the situation where one has to establish the correctness
of compilation schemes of a language to~$N$ architectures,
instead of proving~$N$ results,
one could instead prove correctness of
a mapping from this language to IMM
(assuming that the mappings exists).
Such a proof would still be a proof of compilation
correctness (from the given language to IMM);
we argue that our idea of mixed execution graphs
would be valuable in this
compilation-correctness-proof
effort.
\citet{kokologiannakis-al-23} develop Kater,
a tool that automates reasoning about the metatheory of memory models.
The tool can decide the inclusion between two relations in an
execution graph and it is possible, even though intricate,
to formulate compilation-correctness statements in this fashion.
At the start of our project,
the tool was unfit to our purposes because
the notion of events comes as a built-in, thereby
precluding its use with new types of events such
as non-temporal stores and store fences.
The tool has since then been extended with support for
introducing user-defined sets of events.
However,
at the time of writing,
this feature lacks a comprehensive documentation
and
the tool lacks a specification of
the facts that it takes as assumptions.

%% file: conclusion.tex
\section{Conclusion}
\label{section:conclusion}

In this paper, we have presented a formal model for \ccpp with inline \intel
assembly as an extension of the \rceleven formal consistency model for \ccpp.
One can similarly try to extend \rceleven with inline assembly for other hardware
platforms, such as \arm**.  Doing so is expected to involve a few more challenges,
since the \arm** model makes use of syntactic dependencies between instructions,
which do not have an analogue in the \ccpp setting and are not guaranteed to be
preserved by compilers.   Another possible extension of our work would be to
model the persistency semantics of architectures over non-volatile memory.
We think that both extensions are worth exploring and leave them for future
work.

%% file: acknowledgments.tex
\begin{acks}
Paulo Emílio de Vilhena is supported by the UKRI Future Leaders Fellowship
MR/V024299/1.
Ori Lahav is supported by
the European Research Council under the European Union's Horizon 2020
research and innovation programme (grant agreement No.~851811)
and by the Israel Science Foundation (grant No.~814/22).
Viktor Vafeiadis is supported by
the European Research Council under the European Union’s Horizon 2020
research and by innovation programme (grant agreement No.~101003349).
Azalea Raad is supported by the UKRI Future Leaders Fellowship MR/V024299/1,
by the EPSRC grant EP/X037029/1, and by VeTSS.

\end{acks}

%% file: appendix.tex
\newpage

\appendix %

\title{Extending the C/\cpp Memory Model with Inline Assembly -- Technical Appendix}

\section{Construction of Execution Graphs}
\label{section:execution:graphs:app}

\input{figure-cpp-op-semantics-app.tex}

\section{Models}
\label{section:models:app}

This section contains the definitions of
\intelext~(\sref{subsection:models:intelext:app}),
\rceleven~(\sref{subsection:models:rceleven:app}), and
\rcelevenext~(\sref{subsection:models:rcelevenext:app}).
Each of these models includes
the definition of a set of labels,
the definition of its consistency conditions,
and
the syntax of a programming language
whose semantics is given by the model.
For every such language, the construction of the set of execution graphs
associated with a program is essentially the same as the one presented
in \sref{section:execution:graphs:app}.

\subsection{\intelext}
\label{subsection:models:intelext:app}

\begin{figure}[H]
\[\begin{array}{@{}r@{\;}r@{\;}l@{}}
  \typeExpr\ni\expr & ::=
               & \intn\;(\in\NN)
          \mid   \reg\;(\in\typeReg)
          \mid   \loc\;(\in\typeLoc\triangleq\NN)
          \mid   \Plus\expr\expr
          \mid   \Sub\expr\expr
          \mid   \Times\expr\expr
\\[.5mm]
  \typeCmd\ni\cmd & ::=
               & \ReadPL{}\reg\expr
         \mid    \WritePL{}\expr\expr
         \mid    \RMWPL{}\reg\expr\expr\expr
         \mid    \MFencePL
         \mid    \SFencePL
         \mid    \NTWritePL\expr\expr
\\
        &\mid &  \IfThen\expr\cmd
         \mid    \While\expr\cmd
         \mid    \Seq\cmd\cmd
         \mid    \Skip
\end{array}\]
\caption{Syntax of \intelextlang.}
\Description{}
\label{fig:intelext:syntax:app}
\end{figure}

\begin{figure}[H]
\[\begin{array}{l}
\typeLabel\intelext\ni\evA
::=
  \Write{}\loc\intn
\mid
  \Read{}\loc\intn
\mid
  \RMW[]\loc\intn{\intm^{?}}
\mid
  \MFence
\mid
  \SFence
\mid
  \NTWrite\loc\intn
\end{array}\]
\caption{Set of labels of \intelext.}
\Description{}
\label{fig:intelext:labels:app}
\end{figure}

\begin{definition}[\intelext-Consistency]
\label{def:intelext:consistent:app}
An execution graph
$\tuple\graphA{\rf,\mo}$ is
\emph{\intelext-consistent} if the
conditions
\begin{itemize}
\item \axiomInternalDef\hfill
        \textnormal{(\RefTirName\nameInternal)}
\item $\acyc\ob$
        \hfill
        \textnormal{(\TirName\nameExternal)}
\end{itemize}
hold, where the relations~$\ob$,~$\rb$,~$\pb$, and~$\ppo$
are defined as follows:

\[\begin{array}{@{}lr@{}}
\begin{array}{@{}r@{\;\;}c@{\;\;}l}
\ob
  &\eqdef&
    \ppo\disj
    \rfe\disj\moe\disj\rbe
\\[2mm]
\rb
  &\eqdef&
    (\inv\rf\seq\mo)\setminus\settorel\EventE
\end{array}
&
\begin{array}{@{}r@{\;}c@{\;\;}l@{}}
\ppo
  &\eqdef&
    \po\seq\settorel{\RMWE\cup\MFence\cup\SFenceE}
\\
  &\;\cup&
    \settorel{\ReadE\cup\RMWE\cup\MFence}\seq\po
\\
  &\;\cup&
    \settorel{\SFenceE}\seq\po\seq\settorel{\EventE\setminus\ReadE}
\\
  &\;\cup&
    \settorel\WriteE\seq\po\seq\settorel\WriteE
\\
  &\;\cup&
    \settorel{\WriteE\cup\NTWriteE}\seq
    \perloc\po\seq
    \settorel{\WriteE\cup\NTWriteE}
\end{array}
\end{array}\]

\end{definition}

\begin{definition}[Final State]
\label{def:intelext:finalSt:app}
The \emph{final state} of an execution graph~$\tuple\graph{\rf,\mo}$
is a partial function~$\wtt\finalSt{\typeArrow\typeLoc\NN}$
that maps a location~$\loc$ to the value of
the~$\mo$-maximal (write) event on~$\loc$:
\[
\finalStTypical(\loc)=\intn
  \iff
    \exists\,\evA\neq\initEv\_.\;
    \land\,\left\{\begin{array}{@{\;}l}
      \graphA.\lab(\evA)\in\atloc\WriteE\loc\cup\atloc\RMWSE\loc\\
      \Write{}\loc\intn=\maximum{\mo_\loc}
    \end{array}\right.
\]
\end{definition}

\begin{definition}[\intelextlang Semantics]
\label{def:intelextlang:semantics:app}
The semantics of a \intelextlang program~$\prog$ is defined as the
set of final states to which an initial memory
(where every location initially stores~$0$)
can be updated:
\[
\store\in\interp[\intelext]\prog
  \iff
    \exists\,\graphA,\,\rf,\,\mo.\;
    \land\,\left\{\begin{array}{@{\;}l}
      \poolstep*{\toPool\prog}\initGraph\emptyset\graph \\
      \tuple\graph{\rf,\mo}\;\textit{is \intelext-consistent}\\
      \store = \finalSt(\graphA)
    \end{array}\right.
\]
\end{definition}

\subsection{\rceleven}
\label{subsection:models:rceleven:app}

\input{figure-cpp-syntax-app}

\begin{figure}[H]
\[\begin{array}{l}
\typeLabel\rceleven\ni\evA::=
  \Write\md\loc\intn
\mid
  \Read\md\loc\intn
\mid
  \Fence\md
\mid
  \RMW[\md]\loc\intn{\intm^{?}}
\end{array}\]
\caption{Set of labels of \rceleven.}
\Description{}
\label{fig:rceleven:labels:app}
\end{figure}

\begin{definition}[\rceleven-Consistency]
\label{def:rceleven:consistent:app}
\input{rc11-consistency-app}
\end{definition}

The following definition of data race diverges from
the one presented in the main text~(\defref{def:rceleven:behaviors})
by generalizing~$\hb$ to an arbitrary relation~$R$.
In the main text, this generalization is not necessary,
because only the notion of~$\raceS\hb$ is used.
Here, it is necessary to avoid redundancy
in~\sref{subsection:drf:app}.

\begin{definition}[Data Race]
\label{def:rcelevenext:data:race:app}
Let~$\graph$ be an execution graph,
and~$R$ be a relation on the set of events~$\graph.\EventE$.
A pair of events~$\pair\evA\evB$ forms a
\emph{data race with respect to $R$},
or a \emph{\race{R}} for short,
if the following conditions hold:
\begin{itemize}
\item $\evA \neq \evB$
\item $\dotLoc\evA = \dotLoc\evB$
\item $\{\evA, \evB\}\cap(\WriteE\cup\RMWSE) \neq \emptyset$
\item $\pair\evA\evB\notin\transC{R}$ and~$\pair\evB\evA\notin\transC{R}$
\end{itemize}
The set of $R$-races of~$\graph$ is noted~$\graph.\raceS{R}$
or simply~$\raceS{R}$, when~$\graph$ can be easily inferred.
\end{definition}

\begin{definition}[\rceleven-Behaviors]
\label{def:rceleven:behaviors:app}
\begin{mathpar}
\hspace{-10pt}%
    \inferrule[]{
      \poolstep*{\toPool\prog}\initGraph\emptyset\graphA
        \\\\
      \tuple\graphA{\rf,\mo}\;\textit{is}\;
        \textit{\rceleven-consistent}
    }{
      \progbeh\prog{\tuple\graphA{\rf,\mo}}
    }
\and
    \inferrule[]{
      \poolstep*{\toPool\prog}\initGraph\_\graphA
        \\\\
      \tuple\graphA{\rf,\mo}\;\textit{is}\;
        \textit{\rceleven-consistent}
        \\\\
      \pair\evA\evB\in\raceS\hb \;\land\; \na\in\{\dotMode\evA,\dotMode\evB\}
    }{
      \progbeh\prog\UB
    }
\end{mathpar}
\end{definition}

\begin{definition}[\rcelevenlang Semantics]
\label{def:rcelevenlang:semantics:app}
The semantics of a \rcelevenlang program~$\prog$ is defined as
its set of final states:
\[
\store\in\interp[\rceleven]\prog
  \iff
\progbeh\prog\UB
\;\lor\;
\exists\,\graphA,\,\rf,\,\mo.\;
\progbeh\prog{\tuple\graphA{\rf,\mo}}
  \,\land\,
\store = \finalSt(\graph)
\]
\end{definition}

\subsection{\rcelevenext}
\label{subsection:models:rcelevenext:app}

\begin{figure}[H]
\[\begin{array}{@{}r@{\;}r@{\;}l@{}}
  \typeCmd\ni\cmd ::=
             \ldots
    & \mid & \AsmWritePL\expr\expr
      \mid   \AsmReadPL\reg\expr
      \mid   \AsmRMWPL\reg\expr\expr\expr
      \mid   \AsmMFencePL
  \\
    & \mid & \AsmNTWritePL\expr\expr
      \mid   \AsmSFencePL
  \\
  \typeMode\ni\md ::=
             \ldots
    & \mid & \nt
      \mid   \stf
      \mid   \tso
\end{array}\]
\caption{Syntax of \rcelevenextlang.}
\Description{}
\label{fig:cpp:x86:syntax:app}
\end{figure}

\input{figure-cpp-intelext-diagram-app}

\begin{definition}[\rcelevenext-Consistency]
\label{def:rcelevenext:consistent:app}
\input{rc11-ex86-consistent-app}
\end{definition}

\begin{definition}[\rcelevenext - Behaviors]
\label{def:rcelevenext:behaviors:app}
Analogous to~\defref{def:rceleven:behaviors:app}.
\end{definition}

\begin{definition}[\rcelevenextlang Semantics]
\label{def:rcelevenextlang:semantics:app}
Analogous to~\defref{def:rcelevenlang:semantics:app}.
\end{definition}

\section{Metatheory}

\subsection{Data-Race Freedom}
\label{subsection:drf:app}

\paragraph{Simplifying Assumption.}
For simplicity, we ignore~$\RMWE$ accesses in this subsection.
An idea to overcome this limitation is
to replace our formulation of \rcelevenext-consistency
with one that (like the original formulation of \rceleven~\cite{rc11})
models a~$\RMWE$ access as either a read
or a pair of a write and a read related by a~$\rmw$ relation.

\begin{definition}[\SC-Consistency]
\label{sc:consistent}
An execution graph~$\tuple\graph{\rf,\mo}$ is \emph{\SC-consistent}
if the following condition holds:
\begin{itemize}
\item \axiomSCCoherenceDef\hfill\textnormal{(\RefTirName\nameSCCoherence)}
\end{itemize}
\end{definition}

\begin{definition}[Racy]
\label{sc:racy:app}
Let~$R$ be a relation on events.
An execution graph~$\graph$ is~\emph{\racy{R}} if
it contains a pair of events~$\pair\evA\evB$
that forms a~\race{R} and for which
either~$\evA$ or~$\evB$ is not $\sc$:
\[
\graph\;\textit{is \racy{R}}
\;\iff\;
\exists\,\evA,\,\evB.\;
    \pair\evA\evB\in\raceS{R}
  \land
    (\dotMode\evA\neq\sc\lor\dotMode\evB\neq\sc)
\]
\end{definition}

\begin{definition}[Data-Race Free]
Let~$R$ be a relation on events.
A program~$\prog$ is
\emph{data-race free with respect to~$R$},
or simply~\emph{is \racefree{R}},
if every \SC-consistent execution graph~$\graph$
associated with~$\prog$ is not~\racy{R}:
\[
\prog\;\textit{is \racefree{R}}\;\iff\;
\begin{array}{@{\,}l@{}}
\forall\,\graph,\,\mo,\,\rf,\,\evA,\,\evB.
\\
\quad\left(\begin{array}{@{\,}l@{\,}}
    \poolstep*{\toPool\prog}\initGraph\_\graphA
  \\
    \tuple\graphA{\rf,\mo}\;\textit{is \SC-consistent}
  \\
    \pair\evA\evB\in\raceS{R}
\end{array}\right)
\implies
\dotMode\evA = \dotMode\evB = \sc
\end{array}
\]
\end{definition}

We show that our extended model \rcelevenext
enjoys the data-race-freedom property;
that is,
the semantics assigned by \rcelevenext to a
\racefree\hb program coincides with \SC:
\begin{theorem}[\rcelevenext - Data-Race Freedom]
\label{thm:drf:app}
\[\forall\,\prog.\;
\prog\;\text{is \racefree\hb}\implies
\interp[\rcelevenext]\prog = \interp[\SC]\prog\]
\end{theorem}
\begin{proof}
From~\cororef{coro:hb:racefree:po:rf:racefree},
it follows that~$\prog$ is~\racefree{(\po\disj\onsc\rf)}.
Let~$\graph$ be a \rcelevenext-consistent graph associated with~$\prog$.
By \lemmaref{lemma:non:racy},
the graph is not~\racy{{(\po\disj\onsc\rf)}}.
Then, by \lemmaref{lemma:sc:consistent},
we conclude that~$\graph$ is~\SC-consistent.
\end{proof}

We now state and prove the lemmas on which the proof
of~\thmref{thm:drf:app} relies.

\begin{tcolorbox}[colback=gray!10,     %
                  colframe=black,      %
                  width=\textwidth,    %
                  arc=1mm, auto outer arc,
                 ]
\paragraph{Remark}
A \rcelevenext-consistent (or \rceleven-consistent)
that is not~\racy\hb is not necessarily \SC-consistent:
\input{counterexample-rc11-non-racy-non-SC}
\end{tcolorbox}

\paragraph{Notation.}
The relation~$\coerce{R}{A}$ denotes the restriction of~$R$
to a set~$A$:~%
$\coerce{R}{A}\eqdef\settorel{A}\seq{R}\seq\settorel{A}$.
We use the abbreviation~$\onsc{R}$ for the restriction of~$R$
to~$\EventMd\sc$:~$\onsc{R}\eqdef\coerce{R}{\EventMd\sc}$.
The graph~$\coerce\graph{A}$
is the restriction of~$\graph$ to~$A$:~%
$\coerce\graph{A}\eqdef\graph.\{\EventE:=(\graph.\EventE\cap{A})\}$.

\begin{definition}[$R$-closed]
A set~$A$ is \emph{closed with respect to a relation~$R$},
or simply~\emph{$R$-closed},
if the inclusion~%
$R\seq\settorel{A}\subseteq\settorel{A}\seq{R}$
holds.
\end{definition}

\begin{lemma}
\label{lemma:restriction}
Let~$\tuple\graph{\rf, \mo}$ be a
\rcelevenext-consistent (resp.~\SC-consistent)
graph associated with a program~$\prog$,
and
let~$D\subseteq\graph.\EventE$ be a~%
$(\po\disj\rf)$-closed set.
The graph~$\tuple{\coerce\graph{D}}{\coerce\rf{D},\coerce\mo{D}}$ is
\rcelevenext-consistent (resp.~\SC-consistent).
Moreover, the graph~$\coerce\graph{D}$ is associated with~$\prog$,
and~$\coerce{\graph.\hb}{D}$
(the restriction of~$\graph.\hb$ to~$D$)
is included in~$\coerce\graph{D}.\hb$
(the happens-before relation derived
from~$\coerce\po{D}$ and~$\coerce\rf{D}$).
\end{lemma}
\begin{proof}
Because~$D$ is~$(\po\disj\rf)$-closed,
the following inclusion holds:
\[\begin{array}{rcl}
\coerce{\graph.\sw}{D}
  &\subseteq&
\settorel{D}\seq
\settorel{\EventAtLeast\rel}\seq
     \refC{(
       \settorel\FenceE\seq
       \po
     )}\seq
    \settorel{\WriteAtLeast\rlx}\seq
    \transC\rf\seq
    \settorel{\ReadAtLeast\rlx}\seq
    \refC{(\po\seq{\settorel\FenceE})}\seq
    \settorel{\EventAtLeast\acq}
\seq\settorel{D}
\\
  &\subseteq&
\settorel{\EventAtLeast\rel}\seq
     \refC{(
       \settorel\FenceE\seq
       \coerce\po{D}
     )}\seq
    \settorel{\WriteAtLeast\rlx}\seq
    \transC{\coerce\rf{D}}\seq
    \settorel{\ReadAtLeast\rlx}\seq
    \refC{(\coerce\po{D}\seq{\settorel\FenceE})}\seq
    \settorel{\EventAtLeast\acq}
\\
  &\subseteq&
\coerce\graph{D}.\sw
\end{array}\]
It is then easy to see that~%
$\coerce{\graph.\hb}{D}\subseteq\coerce\graph{D}.\hb$.
\end{proof}

The restriction of~$\rf$ and~$\mo$ to~$D$ in the statement of
\lemmaref{lemma:restriction} is not necessary.
However, these restrictions result in a stronger statement,
because \rcelevenext-consistency (resp. \SC-consistency)
is monotonic with respect to both~$\rf$ and~$\mo$.

\begin{lemma}
The set~$\dom{\reftransC{R}\seq\settorel{A}}$ is $R$-closed,
for any set~$A$ and relation~$R$.
\end{lemma}
\begin{proof}
Let~$D$ stand for~$\dom{\reftransC{R}\seq\settorel{A}}$,
and let~$\pair\evA\evB\in{R}\seq\settorel{D}$.
It suffices to show that~$\evA\in{D}$.
There exists~$\evC$ such that~%
and~$\pair\evB\evC\in\reftransC{R}\seq\settorel{A}$.
Therefore~$\pair\evA\evC\in{R}\seq\reftransC{R}\seq\settorel{A}
\subseteq\reftransC{R}\seq\settorel{A}$,
which implies that~$\evA\in{D}$.
\end{proof}

\begin{corollary}
\label{coro:po:rf:closed}
Let~$\graph$ be a graph,
and~$D$ be a subset of~$\graph.\EventE$.
The set~$\dom{\reftransC{(\po\disj\rf)}\seq\settorel{D}}$
is $(\po\disj\rf)$-closed.
\end{corollary}

\begin{lemma}
\label{lemma:rf:incl:onsc:rf}
Let~$\prog$ be a \racefree\hb program
and let~$\tuple\graph{\rf, \mo}$
be a \SC-consistent graph associated with~$\prog$.
The relation~$\rf$ is included in~$\transC{(\po\disj\onsc\rf)}$.
\end{lemma}

\begin{proof}
We introduce the following notation:
\[\begin{array}{rcl}
R &\triangleq& \transC{(\po\disj\onsc\rf)}\\
S &\triangleq& \rf\setminus{R}
\end{array}\]

Suppose by contradiction that~$S$ is non-empty.
Let~$\pair\evA\evB$ be a pair in~$S$
for which~$\evB$ is minimal with respect to~$\po\disj\rf$.
(Such a pair exists because $\po\disj\rf$ is acyclic in~$\graph$.)

\paragraph{Claim~1}
\hypertarget{lemma:rf:incl:onsc:rf:C1}{}
The pair~$\pair\evA\evB$ belongs to~$\rfe$.
\begin{proof}
If~$\pair\evA\evB\in\rfi$,
then~$\pair\evA\evB\in\po\subseteq{R}$.
\end{proof}

Let~$D$
be the set~%
$\dom{\reftransC{(\po\disj\rf)}\seq\settorel{\{\evA,\evB\}}}$.

\paragraph{Claim~2}
\hypertarget{lemma:rf:incl:onsc:rf:C2}{}
If~$\evC$ belongs to~$D$, then~%
$\evC\in\dom{\reftransC{(\po\disj\onsc\rf)}\seq\settorel{\{\evA,\evB\}}}$.
\begin{proof}
Suppose by contradiction that there exists~$\evC\neq\evA,\evB$
such that~$\evC\in\dom{\reftransC{(\po\disj\rf)}\seq
\settorel{\{\evA,\evB\}}}$
but~%
$\evC\notin\dom{\reftransC{(\po\disj\onsc\rf)}\seq\settorel{\{\evA,\evB\}}}$.
Then it must be the case that~%
\[\evC\in\dom{
\reftransC{(\po\disj\rf)}\seq
{S}\seq
\reftransC{(\po\disj\rf)}\seq
\settorel{\{\evA,\evB\}}}.\]
However, this implies the existence of a pair in~%
$S\setminus\{\pair\evA\evB\}$
that contradicts the minimality of~$\evB$ with respect to~$\po\disj\rf$.
\end{proof}

\paragraph{Claim~3}
\hypertarget{lemma:rf:incl:onsc:rf:C3}{}
The graph~%
$\tuple{\coerce\graph{D}}{\coerce\rf{D}, \coerce\mo{D}}$
is not \racy{(\coerce\graph{D}.\hb)}, \SC-consistent,
and associated with~$\prog$.
\begin{proof}
That the graph is \SC-consistent and associated with~$\prog$
is a direct consequence from \lemmaref{lemma:restriction}.
Let~$\evC,\evD\in{D}$ be a pair of distinct events
acting on the same location.
Because~$\graph$ is not \racy{\graph.\hb},
it follows that~$\pair\evC\evD\in\graph.\hb\disj\inv{\graph.\hb}$.
It follows from~\lemmaref{lemma:restriction}
that~$\pair\evC\evD\in\coerce\graph{D}.\hb\disj\inv{\coerce\graph{D}.\hb}$.
\end{proof}

\paragraph{Claim~4}
\hypertarget{lemma:rf:incl:onsc:rf:C4}{}
The events~$\evA$ and~$\evB$ are $\coerce\po{D}$-maximal.
\begin{proof}
Suppose by contradiction that there exists~$\evC\in{D}$
such that~$\pair\evA\evC\in\po$.
Because~$\evC\in{D}$,
either~$\pair\evC\evA\in\reftransC{(\po\disj\rf)}$
or~$\pair\evC\evB\in\reftransC{(\po\disj\rf)}$.
If~$\pair\evC\evA\in\reftransC{(\po\disj\rf)}$,
then~$\pair\evA\evC$ violates
the acyclicity of~$\po\disj\rf$.
If~$\pair\evC\evB\in\reftransC{(\po\disj\rf)}$,
then, from \hyperlink{lemma:rf:incl:onsc:rf:C2}{Claim~2},
it follows that
$\pair\evC\evB\in\reftransC{(\po\disj\onsc\rf)}$.
But then~$\pair\evA\evB\in\po\seq\reftransC{(\po\disj\onsc\rf)}\subseteq{R}$,
a contradiction to~$\pair\evA\evB\in{S}$.

Suppose by contradiction that there exists~$\evC\in{D}$
such that~$\pair\evB\evC\in\po$.
If~$\pair\evC\evA\in\reftransC{(\po\disj\rf)}$,
then~$\tuple\evA{\evB,\evC}$ violates
the acyclicity of~$\po\disj\rf$.
Analogously,
if~$\pair\evC\evB\in\reftransC{(\po\disj\rf)}$,
then~$\tuple\evB\evC$ violates
the acyclicity of~$\po\disj\rf$.
\end{proof}

Let~$\evE\in{D}$ be the write event
that immediately precedes~$\evA$ in~$\mo$,
that is, $\pair\evE\evA\in\imm\mo$.
We define the graph~$\graphB$, and the relations~$\rf'$ and~$\rb'$
as follows:
\[\begin{array}{rcl}
\graphB &\eqdef&

  \coerce\graph{D}.\{
    \lab :=
      \coerce\graph{D}.\lab[\evB\mapsto\Read{\md_2}\locx\val]
  \}
\;\;\text{where}\;
\left\{\begin{array}{@{}l@{}}
\coerce\graph{D}.\lab(\evE) = \Write{\md_1}\locx\val\\
\coerce\graph{D}.\lab(\evB) = \Read{\md_2}\locx\_
\end{array}\right.
\\
\rf'&\eqdef&
  (\coerce\rf{D}\setminus\{\pair\evA\evB\})\cup\{\pair\evE\evB\}
\end{array}\]

\paragraph{Claim~5}
\hypertarget{lemma:rf:incl:onsc:rf:C5}{}
The graph~%
$\tuple\graphB{\rf', \coerce\mo{D}}$
is \SC-consistent, and associated with~$\prog$.
\begin{proof}
Suppose, by contradiction, that there is a cycle~$C$
in~%
$\coerce\po{D}\disj\rf'\disj\coerce\mo{D}\disj\rb'$.
Let~$\rb'$ denote the derived reads-before relation~$\graphB.\rb$.
Because~$\rb'\setminus\coerce\rb{D} = \{\pair\evB\evA\}$
and~$\rf'\setminus\coerce\rf{D} = \{\pair\evE\evB\}$,
the only edges
that could possibly be in~$\cycle$
but not in~$\coerce\graph{D}$ are~$\pair\evB\evA$
and~$\pair\evE\evB$.
If~$C$ includes the~$\rf'$ edge~$\pair\evE\evB$,
then it must be followed by a~$\rb'$ edge,
because $\evB$ is a $\coerce\po{D}$-maximal read event.
By definition of the reads-before relation,
it follows that~$\rf'\seq\rb'$ is included
in~$\graphB.\mo$,
which coincides with~$\coerce\mo{D}$.
The sequence~%
$\settorel\evE\seq\rf'\seq\settorel\evB\seq\rb'$
can thus be exchanged with~%
$\settorel\evE\seq\coerce\mo{D}$.
Moreover,
if~$C$ includes the~$\rb'$ edge~$\pair\evB\evA$,
then it must be followed by a~$\coerce\mo{D}$
edge because
$\evA$ is a $\coerce\po{D}$-maximal write event.
The sequence~%
$\settorel\evB\seq\rb'\seq\settorel\evA\seq\coerce\mo{D}$
can thus be exchanged with~%
$\settorel\evB\seq\coerce\rb{D}$.
Performing all these exchanges yields a cycle in~%
$\coerce\po{D}\disj\coerce\rf{D}\disj\coerce\mo{D}\disj\coerce\rb{D}$,
a contradiction to the \SC-consistency of~$\coerce\graph{D}$.

\end{proof}

\paragraph{Claim~6}
\hypertarget{lemma:rf:incl:onsc:rf:C6}{}
The graph~%
$\tuple\graphB{\rf', \coerce\mo{D}}$
is \racy{(\graph'.\hb)}.
\begin{proof}
The events~$\evA$ and~$\evB$ are distinct,
they act on the same location,
at least one of them is a write event.
Moreover,
$\minimum{\dotMode\evA, \dotMode\evB}\neq\sc$,
because otherwise~$\pair\evA\evB\in{R}$.
Finally,
because both are~$\coerce\po{D}$-maximal,
and because~$\evA\notin\dom{\rf'}$,
it follows that~$\settorel{\{\evA,\evB\}}\seq\graphB.\hb\subseteq
\settorel{\{\evA,\evB\}}\seq\transC{(\coerce\po{D}\disj\rf')}
=\emptyset
$.
Therefore,
the pair~$\pair\evA\evB$ forms a \race{(\graph'.\hb)}.
\end{proof}

From
\hyperlink{lemma:rf:incl:onsc:rf:C5}{Claims~5}
and
\hyperlink{lemma:rf:incl:onsc:rf:C6}{6},
it follows that
the graph~$\tuple\graphB{\rf', \coerce\mo{D}}$
contradicts the assumption that~$\prog$
is~\racefree\hb.
\end{proof}

\begin{corollary}
\label{coro:hb:racefree:po:rf:racefree}
If a program~$\prog$ is~\racefree\hb,
then it is~\racefree{{(\po\disj\onsc\rf)}}.
\end{corollary}

\begin{lemma}
\label{lemma:sc:consistent}
Let~$\graph$ be a \rcelevenext-consistent graph.
If~$\graph$ is not \racy{{(\po\disj\onsc\rf)}},
then~$\graph$ is \SC-consistent.
\end{lemma}
\begin{proof}
Because~$\graph$ is not \racy{{(\po\disj\onsc\rf)}},
the following inclusions hold:
\[\begin{array}{@{}l@{\;}c@{\;}r@{\;}c@{\;\;}l@{}}
\rf &\subseteq&
  \onsc\rf&\disj&\transC{(\po\disj\onsc\rf)}
\\
\mo &\subseteq&
  \onsc\mo&\disj&\transC{(\po\disj\onsc\rf)}
\\
\rb &\subseteq&
  \onsc\rb&\disj&\transC{(\po\disj\onsc\rf)}
\end{array}\]

We can use these inclusions to show that
the violation of \SC-consistency yields
the following chain of implications:
\[\begin{array}{@{}l@{}}
\cyc{\po\disj\rf\disj\mo\disj\rb}
\\
\begin{array}{@{\implies}l@{}}
\cyc{
\po\disj\onsc\rf
  \disj
\onsc\mo
  \disj
\onsc\rb
  \disj
\transC{(\po\disj\onsc\rf)}
}
\\
\cyc{
\po
  \disj
\onsc\rfe
  \disj
\onsc\moe
  \disj
\onsc\rbe
}
\end{array}
\end{array}\]

The internal edges of a cycle in~%
$\po
  \disj
\onsc\rfe
  \disj
\onsc\moe
  \disj
\onsc\rbe$
are separated by external edges
whose domain and codomain is included in~$\EventMd\sc$.
Consequently,
the domain and codomain of
the internal edges~$\po$
must also be included in~$\EventMd\sc$.
It follows that
the relation~%
$\onsc\po
  \disj
\onsc\rfe
  \disj
\onsc\moe
  \disj
\onsc\rbe$
is cyclic,
a contradiction to Condition~\axiomSC.
\end{proof}

\begin{lemma}
\label{lemma:non:racy}
If a program~$\prog$ is \racefree{{(\po\disj\onsc\rf)}},
then
every \rcelevenext-consistent graph associated with~$\prog$
is not \racy{{(\po\disj\onsc\rf)}}.
\end{lemma}
\begin{proof}

\newcommand{\D}{D}
\newcommand{\Dab}{D_{\evA\evB}}
\newcommand{\Da}{D_\evA}
\newcommand{\Db}{D_\evB}

Let~$\tuple\graph{\rf,\mo}$ be a \rcelevenext-consistent
graph associated with~$\prog$.
Suppose by contradiction that~$\graph$ is
\racy{{(\po\disj\onsc\rf)}}.
In other words, suppose that the set~%
$S\triangleq
\raceS{\po\disj\onsc\rf}
\setminus{(\EventMd\sc\times\EventMd\sc)}$
is non-empty.
Let~$\pair\evA\evB$ be a pair in~$S$
that is minimal with respect to~$\po\cup\rf$,
that is:
\[
\forall\,\evC,\,\evD.\;
\pair\evC\evD\in{S}
\implies
\evC,\evD\in
\dom{\reftransC{(\po\cup\rf)}\seq\settorel{\{\evA,\evB\}}}
\implies
\pair\evC\evD\in\{\pair\evA\evB,\pair\evB\evA\}.
\]

Let~$D$ be the set~%
$\dom{\reftransC{(\po\cup\rf)}\seq\settorel{\{\evA,\evB\}}}$.

\paragraph{Claim~1}
\hypertarget{lemma:non:racy:C1}{}
The events~$\evA$ and~$\evB$ belong to different threads.
\begin{proof}
If either~$\pair\evA\evB\in\po$ or~$\pair\evB\evA\in\po$,
then~$\pair\evA\evB$ would not form a~%
\race{(\po\disj\onsc\rf)}.
\end{proof}

\paragraph{Claim~2}
\hypertarget{lemma:non:racy:C2}{}
$
\reftransC{(\po\disj\rf)}\seq
\settorel{\EventE\setminus\{\evA,\evB\}}\seq
(\po\disj\rf)\seq
\settorel{\{\evA,\evB\}}
\subseteq
\transC{(\po\disj\onsc\rf)}
$
\begin{proof}
Suppose by contradiction that~%
\[
\reftransC{(\po\disj\rf)}\seq
\settorel{\EventE\setminus\{\evA,\evB\}}\seq
(\po\disj\rf)\seq
\settorel{\{\evA,\evB\}}
\setminus
\transC{(\po\disj\onsc\rf)}
\neq
\emptyset.
\]

Then,
there must be an edge~%
$\pair\evC\evD\in\reftransC{(\po\disj\rf)}$,
such that~%
\[\evD\in
\dom{\settorel{\EventE\setminus\{\evA,\evB\}}\seq
(\po\disj\rf)\seq
\settorel{\{\evA,\evB\}}},\]
and~$\pair\evC\evD\notin\transC{(\po\disj\onsc\rf)}$.

We claim that~%
$\evD\in\dom{(\po\disj\onsc\rf)\seq\settorel{\{\evA,\evB\}}}$.
If~$\evD\in\dom{(\po\disj\rfi)\seq\settorel{\{\evA,\evB\}}}$,
then the assertion follows immediately.
Moreover, if~%
$\evD\in\dom{\rfe\seq\settorel{\{\evA,\evB\}}}$, and if~%
$\evD\notin\dom{(\po\disj\onsc\rf)\seq\settorel{\{\evA,\evB\}}}$,
then it must be the case that~%
$\evD\notin\dom{(\rfe\setminus\onsc\rf)\seq\settorel{\{\evA,\evB\}}}$.
However, in this case,
either the pair~$\pair\evD\evA$ or the pair~$\pair\evD\evB$
would contradict the minimality of~$\pair\evA\evB$.

We now proceed by induction on the number
of~$(\po\disj\rf)$ steps between~$\evC$ and~$\evD$.
In the base case, when~$\evC=\evD$,
the desired conclusion follows from the previous paragraph.
In the inductive case, we have~$\evE$ such that~%
$\pair\evC\evE\in(\po\disj\rf)$
and~%
$\pair\evE\evD\in\transC{(\po\disj\onsc\rf)}$.
We wish to prove that~$\pair\evC\evE\in\transC{(\po\disj\onsc\rf)}$.
If not, then
the pair~$\pair\evC\evE$
forms a \race{(\po\disj\onsc\rf)}.
The minimality of~$\pair\evA\evB$ implies that~%
$\pair\evC\evE\in\{\pair\evA\evB,\pair\evB\evA\}$.
However,
because~%
$\evE\in\dom{\transC{(\po\disj\onsc\rf)}\seq\settorel{\{\evA,\evB\}}}$,
this contradicts either~\axiomNTA
or the fact that~$\pair\evA\evB\in{S}$.
\end{proof}

\paragraph{Claim~3}
\hypertarget{lemma:non:racy:C3}{}
$
\settorel{\{\evA,\evB\}}\seq
\transC{(\po\disj\rf)}\seq
{(\po\disj\rf)}\seq
\settorel{\{\evA,\evB\}}
=
\emptyset$.
\begin{proof}
The relation~$R \triangleq\settorel{\{\evA,\evB\}}\seq
\transC{(\po\disj\rf)}\seq
{(\po\disj\rf)}\seq
\settorel{\{\evA,\evB\}}$
is included in~
\[\{\pair\evA\evA, \pair\evA\evB, \pair\evB\evA, \pair\evB\evB \}.\]

The inclusions~$\pair\evA\evA\in{R}$ and~$\pair\evB\evB\in{R}$
contradict \axiomNTA.
Therefore, if~$R$ is non-empty,
then it must be the case that either~$\pair\evA\evB\in{R}$
or~$\pair\evB\evA\in{R}$.
Suppose by contradiction, and without loss of generality,
that~$\pair\evA\evB\in{R}$.
Then, there exists~$\evC$ such that~%
$\pair\evA\evC\in\transC{(\po\disj\rf)}$
and~%
$\pair\evC\evB\in{(\po\disj\rf)}$.
If~$\evC\in\{\evA,\evB\}$,
then either~$\pair\evA\evC$
or~$\pair\evC\evB$
contradicts \axiomNTA.
If~$\evC\notin\{\evA,\evB\}$,
then, by \hyperlink{lemma:non:racy:C2}{Claim~2},
it follows that~$\pair\evA\evB\in\transC{(\po\disj\onsc\rf)}$,
a contradiction to the fact that~$\pair\evA\evB\in{S}$.
\end{proof}

\paragraph{Claim~4}
\hypertarget{lemma:non:racy:C4}{}
The events~$\evA$ and~$\evB$ are~$\coerce\po{D}$-maximal.

If~$\pair\evA\evB\in\rf$ (resp.~$\pair\evB\evA\in\rf$),
then~$\evB$ (resp.~$\evA$) is~%
$(\coerce\po\D\disj\coerce\rf\D)$-maximal.

If~$\pair\evA\evB\notin\rf\disj\inv\rf$,
then~$\evA$ and~$\evB$ are~%
$(\coerce\po\D\disj\coerce\rf\D)$-maximal.
\begin{proof}
Suppose by contradiction,
and without loss of generality,
that~$\evA$ is not~$\coerce\po\D$-maximal,
and let~$\evC$ be an event in~$\D$
such that~$\pair\evA\evC\in\po$.
Then it must be the case that
either~$\pair\evC\evA\in\reftransC{(\po\disj\rf)}$
or~$\pair\evC\evB\in\reftransC{(\po\disj\rf)}$.
The first case contradicts \axiomNTA.
If~$\pair\evC\evB\in\reftransC{(\po\disj\rf)}$,
then~$\pair\evC\evB\in\transC{(\po\disj\rf)}$,
otherwise~$\evC=\evB$ and~$\pair\evA\evB\in\po$,
a contradiction to~\hyperlink{lemma:non:racy:C1}{Claim~1}.
The inclusion~$\pair\evC\evB\in\transC{(\po\disj\rf)}$,
however,
contradicts
\hyperlink{lemma:non:racy:C3}{Claim~3},
because then~$
\pair\evA\evB\in
  \settorel\evA\seq
  \po\seq
  \settorel\evC\seq
  \transC{(\po\disj\rf)}\seq
  \settorel\evB
\subseteq
  \settorel{\{\evA,\evB\}}\seq
  \transC{(\po\disj\rf)}\seq(\po\disj\rf)\seq
  \settorel{\{\evA,\evB\}}
$.

If~$\pair\evA\evB\notin\rf\disj\inv\rf$,
then we prove, without loss of generality,
that~$\evA$ is~%
${(\coerce\po\D\disj\coerce\rf\D)}$-maximal.
It suffices to show that~%
$\settorel\evA\seq\transC{(\coerce\po\D\disj\coerce\rf\D)}$
is empty.
Suppose by contradiction that there exists~$\evC\in\D$
such that~$\pair\evA\evC\in(\coerce\po\D\disj\coerce\rf\D)$
and~$\pair\evC\_\in\reftransC{(\coerce\po\D\disj\coerce\rf\D)}$.
Because~$\evA$ is~$\coerce\po\D$-maximal,
the edge~$\pair\evA\evC$ belongs to~$\coerce\rf\D$.
Moreover, because~$\evC\in\D$
it is the case that either~%
$\pair\evC\evA\in\reftransC{(\po\disj\rf)}$
or~$\pair\evC\evB\in\reftransC{(\po\disj\rf)}$.
The first case contradicts \axiomNTA.
In the second case,
the edge~$\pair\evC\evB$ must also belong
to~$\transC{(\po\disj\rf)}$;
otherwise, the events~$\evC$ and~$\evB$ would coincide,
and therefore~$\pair\evA\evB\in\rf$,
a contradiction to the assumption
that~$\pair\evA\evB\notin\rf\disj\inv\rf$.
The inclusion~$\pair\evC\evB\in\transC{(\po\disj\rf)}$,
however, contradicts
\hyperlink{lemma:non:racy:C3}{Claim~3},
because then~$
\pair\evA\evB\in
  \settorel\evA\seq
  \rf\seq
  \settorel\evC\seq
  \transC{(\po\disj\rf)}\seq
  \settorel\evB
\subseteq
  \settorel{\{\evA,\evB\}}\seq
  \transC{(\po\disj\rf)}\seq(\po\disj\rf)\seq
  \settorel{\{\evA,\evB\}}
$.

Finally, we prove, without loss of generality,
that, if~$\pair\evA\evB\in\rf$,
then~$\evB$ is~$(\coerce\po\D\disj\coerce\rf\D)$-maximal.
Because~$\evB$ is a~$\coerce\po\D$-maximal read event,
the relation~$\settorel\evB\seq(\coerce\po\D\disj\coerce\rf\D)$
is empty.
Therefore, the relation~%
$\settorel\evB\seq\transC{(\coerce\po\D\disj\coerce\rf\D)}$
is empty,
which conclusion finishes the proof.
\end{proof}

We now introduce the following sets:
\[\begin{array}{@{}r@{\;\;}c@{\;\;}l@{}}
\Dab &\triangleq& \D\setminus\{\evA,\evB\}
\\
\Da &\triangleq& \D\setminus\{\evA\}
\\
\Db &\triangleq& \D\setminus\{\evB\}
\end{array}\]

\paragraph{Claim~5}
\hypertarget{lemma:non:racy:C5}{}
The set~$\Dab$ is~$(\po\disj\rf)$-closed.

If~$\pair\evA\evB\in\rf$ (resp.~$\pair\evB\evA\in\rf$),
then~$\Db$ (resp.~$\Da$) is~$(\po\disj\rf)$-closed.

If~$\pair\evA\evB\notin\rf\disj\inv\rf$,
then the sets~$\Da$ and~$\Db$ are both~$(\po\disj\rf)$-closed.
\begin{proof}
To show that~$\Dab$ is~$(\po\disj\rf)$-closed,
it suffices to prove that the inclusion~%
$(\po\disj\rf)\seq\settorel\Dab\subseteq
\settorel\Dab\seq(\po\disj\rf)$ holds.
Suppose by contradiction that it does not,
and let~$\pair\evC\evD$ be a pair such that~%
$\pair\evC\evD\in(\po\disj\rf)\seq\settorel\Dab$
and
$\pair\evC\evD\notin\settorel\Dab\seq(\po\disj\rf)$.
Because~$\D$ is~$(\po\disj\rf)$-closed
(by definition),
and because~$\Dab\subseteq\D$,
it follows that~%
$\pair\evC\evD\in\settorel\D\seq(\po\disj\rf)$.
We thus conclude that~$\evC\in\{\evA,\evB\}$,
because~%
$\pair\evC\evD\in
(\settorel\D\seq(\po\disj\rf))
\setminus
(\settorel\Dab\seq(\po\disj\rf))
\subseteq
\settorel{\{\evA,\evB\}}\seq(\po\disj\rf)
$.
Because~$\evD\in\Dab\subseteq\D$,
there exists~$\evE\in\{\evA,\evB\}$
such that~%
$\pair\evD\evE\in\reftransC{(\po\disj\rf)}$,
by definition of~$\D$.
In fact, it must be the case that~%
$\pair\evD\evE\in\transC{(\po\disj\rf)}$,
because~$\evD\notin\{\evA,\evB\}$.
We then reach a contradiction to
\hyperlink{lemma:non:racy:C3}{Claim~3},
because~$
\pair\evC\evE\in
  \settorel{\{\evA,\evB\}}\seq
  (\po\disj\rf)\seq
  \settorel\evD\seq
  \transC{(\po\disj\rf)}\seq
  \settorel{\{\evA,\evB\}}\seq
\subseteq
  \settorel{\{\evA,\evB\}}\seq
  \transC{(\po\disj\rf)}\seq(\po\disj\rf)\seq
  \settorel{\{\evA,\evB\}}
$.
The remaining claims follow a similar proof.
\end{proof}

\paragraph{Claim~6}
\hypertarget{lemma:non:racy:C6}{}
The graph~$\tuple{\coerce\graph\Dab}{\coerce\rf\Dab,\coerce\mo\Dab}$
is not \racy{{(\coerce\po\Dab\disj\onsc{\coerce\rf\Dab})}}.
\begin{proof}
Suppose by contradiction that~%
$\raceS{{(\coerce\po\Dab\disj\onsc{\coerce\rf\Dab})}}
\setminus(\EventMd\sc\times\EventMd\sc)$
is non-empty,
and let~$\pair\evC\evD$ be a pair in this set.
We claim that~$\pair\evC\evD\in{S}$.
If not, then~%
$\pair\evC\evD\in
\transC{(\po\disj\onsc\rf)}
\disj
\inv{(\transC{(\po\disj\onsc\rf)})}$.
Suppose without loss of generality
that~$\pair\evC\evD\in\transC{(\po\disj\onsc\rf)}$.
Because~$\Dab$ is~$(\po\disj\rf)$-closed,
and because~$\evC,\evD\in\Dab$,
it follows that~%
$\pair\evC\evD\in\transC{(\coerce\po\Dab\disj\onsc{\coerce\rf\Dab})}$,
a contradiction with the fact that~%
$\pair\evC\evD$ forms a~%
\race{(\coerce\po\Dab\disj\onsc{\coerce\rf\Dab})}.
Therefore, it must be the case that~$\pair\evC\evD\in{S}$.
Because~$\evC,\evD\in\D$,
it follows from the minimality of~$\pair\evA\evB$
that~$\pair\evC\evD\in\{\pair\evA\evB,\pair\evB\evA\}$,
a contradiction with the fact that~$\evC,\evD\in\Dab$.
\end{proof}

\paragraph{Claim~7}
\hypertarget{lemma:non:racy:C7}{}
The graph~$\tuple{\coerce\graph\Da}{\coerce\rf\Da,\coerce\mo\Da}$
is not \racy{{(\coerce\po\Da\disj\onsc{\coerce\rf\Da})}}.

Analogously,
the graph~$\tuple{\coerce\graph\Db}{\coerce\rf\Db,\coerce\mo\Db}$
is not \racy{{(\coerce\po\Db\disj\onsc{\coerce\rf\Db})}}.
\begin{proof}
Proof similar to \hyperlink{lemma:non:racy:C6}{Claim~6}.
\end{proof}

\paragraph{Claim~8}
\hypertarget{lemma:non:racy:C8}{}
The graph~$\tuple{\coerce\graph\Dab}{\coerce\rf\Dab,\coerce\mo\Dab}$
is \SC-consistent.
\begin{proof}
From \hyperlink{lemma:non:racy:C5}{Claim~5}
and \lemmaref{lemma:restriction},
it follows that~$\coerce\graph\Dab$ is \rcelevenext-consistent.
Finally,
thanks to~\hyperlink{lemma:non:racy:C6}{Claim~6}
and~\lemmaref{lemma:sc:consistent},
it follows that~$\coerce\graph\Dab$ is \SC-consistent.
\end{proof}

\paragraph{Claim~9}
\hypertarget{lemma:non:racy:C9}{}
If~$\pair\evA\evB\in\rf$~(resp.~$\pair\evB\evA\in\rf$),
then~$\tuple{\coerce\graph\Db}{\coerce\rf\Db,\coerce\mo\Db}$
(resp.~$\tuple{\coerce\graph\Da}{\coerce\rf\Da,\coerce\mo\Da}$)
is \SC-consistent.
\begin{proof}
Proof similar to \hyperlink{lemma:non:racy:C8}{Claim~8}.
\end{proof}

\paragraph{Claim~10}
\hypertarget{lemma:non:racy:C10}{}
If~$\pair\evA\evB\notin\rf\disj\inv\rf$,
then~$\tuple{\coerce\graph\Da}{\coerce\rf\Da,\coerce\mo\Da}$
and~$\tuple{\coerce\graph\Db}{\coerce\rf\Db,\coerce\mo\Db}$
are \SC-consistent.
\begin{proof}
Proof similar to \hyperlink{lemma:non:racy:C8}{Claim~8}.
\end{proof}

\paragraph{Claim~11}
\hypertarget{lemma:non:racy:C11}{}
The graph~$\tuple{\coerce\graph\D}{\coerce\rf\D,\coerce\mo\D}$
is \SC-consistent, and associated with~$\prog$.
\begin{proof}
It follows from~\lemmaref{lemma:restriction} that~%
$\coerce\graph\D$ is \rcelevenext-consistent, and associated with~$\prog$.
To prove that~$\coerce\graph\D$ is \SC-consistent,
suppose by contradiction that it is not,
and let~$C$ be a cycle that violates~\axiomSCCoherence.
Now, proceed by case disjunction on whether~%
$\pair\evA\evB\in
(\rf\disj\mo\disj\rb)
\disj
\inv{(\rf\disj\mo\disj\rb)}
$.
If it does not, then a cycle that violates~\axiomSCCoherence
must be included in~$\Dab$.
However, such a cycle would contradict
\hyperlink{lemma:non:racy:C6}{Claim~6}.
Let us consider the case in which~%
$\pair\evA\evB\in
(\rf\disj\mo\disj\rb)
\disj
\inv{(\rf\disj\mo\disj\rb)}
$,
and let us assume, without loss of generality,
that~$\pair\evA\evB\in(\rf\disj\mo\disj\rb)$.
We now consider the following cases:
\begin{itemize}
\item Case:~$\pair\evA\evB\in\rf$.\\
If~$\pair\evA\evB\in\rf$,
then, by~\hyperlink{lemma:non:racy:C9}{Claim~9},
the graph~$\coerce\graph\Db$ is \SC-consistent.
Therefore, the cycle~$C$ must include~$\evB$,
otherwise it would violate the \SC-consistency
of~$\coerce\graph\Db$.
Because~$\evB$ is~$\coerce\po\D$-maximal
(\hyperlink{lemma:non:racy:C4}{Claim~4}),
the only edge that can follow~$\evB$ in~$C$
is a~$\rb$ edge.
However, because~$\rf\seq\rb\subseteq\mo$,
the sequence~$\rf\seq\settorel{\evB}\seq\rb$
can be replaced with a~$\mo$ edge that avoids~$\evB$,
thereby yielding a cycle that violates
the \SC-consistency of~$\coerce\graph\Db$.

\item Case:~$\pair\evA\evB\in(\mo\disj\rb)$.\\
If~$\pair\evA\evB\in(\mo\disj\rb)$,
then, by~\hyperlink{lemma:non:racy:C10}{Claim~10},
both~$\coerce\graph\Da$ and~$\coerce\graph\Db$ are \SC-consistent.
Therefore, the cycle~$C$ must include~$\evA$ and~$\evB$,
otherwise it would violate the \SC-consistency
of either~$\coerce\graph\Da$ or~$\coerce\graph\Db$.
Because~$\evB$ is~$(\coerce\po\D\disj\coerce\rf\D)$-maximal
(\hyperlink{lemma:non:racy:C4}{Claim~4}),
the only edge that can follow~$\evB$ in~$C$
is a~$\mo$ edge.
However, because~$(\mo\disj\rb)\seq\mo\subseteq(\mo\disj\rb)$,
the sequence~$(\mo\disj\rb)\seq\settorel{\evB}\seq\mo$
can be replaced with a~$(\mo\disj\rb)$ edge that avoids~$\evB$,
thereby yielding a cycle that violates
the \SC-consistency of~$\coerce\graph\Db$.

\end{itemize}
\end{proof}

The graph~$\coerce\graph\D$ contradicts
the premise that~$\prog$ is~\racefree{(\po\disj\onsc\rf)},
because
the graph~$\coerce\graph\D$ is
(1)~\SC-consistent;
(2)~associated with~$\prog$;
and
(3)~contains a~\race{(\po\disj\onsc\rf)},
the pair~$\pair\evA\evB$.
\end{proof}

\subsection{Extension Property}

\begin{theorem}[\rcelevenext - Extension-I]
\label{thm:extension:I:rcelevenext:app}
\rceleven-consistency conditions are equivalent
to those of \rcelevenext in every execution graph
containing only \rceleven events.
\end{theorem}
\begin{proof}
In the absence of~$\intelext$ events,
the relation~$\porceleven$ is equivalent to~$\po$,
and
the relation~$\ppoasm$ is equivalent to
$
\po\seq\settorel{\FenceMd\sc}\disj
\settorel{\FenceMd\sc}\seq\po
$.
Therefore,
the only differences between the two models,
with respect to the derived relations,
is the definition of~$\eco$:
in \rceleven,
it is defined as~$\transC{(\rf\disj\mo\disj\rb)}$,
whereas,
in \rcelevenext,
it is defined as~$\transC{(\rfe\disj\mo\disj\rb)}$.
Let us use the names of these models as a prefix to
distinguish to which version of $\eco$
(or of any other relation defined on top of $\eco$)
we refer.

To complete the proof it is thus sufficient to show that
\axiomCoherenceII
and
\axiomCoherenceIII
are a consequence of
\rceleven-consistency, and that,
in the remaining \rceleven-consistency conditions,
the relations~$\rceleven.\eco$ and~$\rcelevenext.\eco$
can be used interchangeably.
That is, it suffices to show that the following assertions hold:
\begin{enumerate}
\item $\irr{\hb\seq\rceleven.\eco}\implies
       \irr{\po\seq\rb}$
\item $\acyc{\rceleven.\psc}\implies
       \acyc{\ppoasm\disj\rcelevenext.\eco}$
\item $\irr{\hb\seq{\refC{\rceleven.\eco}}}\iff
       \irr{\hb\seq{\refC{\rcelevenext.\eco}}}$
\item $\acyc{\rceleven.\psc}\iff\acyc{\rcelevenext.\psc}$
\end{enumerate}

\paragraph{Proof of Assertion~(1).}
Immediate from
$\rb\subseteq\eco$ and
$\po\subseteq\hb$
(which holds of \rceleven-events-only execution graphs).

\paragraph{Proof of Assertion~(2).}
We proceed by contradiction;
that is,
we show that,
if~$\ppoasm\disj\rcelevenext.\eco$ is cyclic, then so is~$\rceleven.\psc$:
\[\begin{array}{@{}l@{}}
\cyc{\ppoasm\disj\rcelevenext.\eco}
\\
\begin{array}{@{}r@{\;}c@{\;}l@{}}
  &\implies&
    \cyc{\po\seq\settorel{\FenceMd\sc}\disj\settorel{\FenceMd\sc}\seq\po
    \disj
    \rceleven.\eco}
\\
  &\implies&
    \cyc{\po\seq\settorel{\FenceMd\sc}\disj\settorel{\FenceMd\sc}\seq\po\disj
    \rceleven.\eco}
\\
  &\implies&
    \cyc{\rcelevenext.\eco} \;\;\absurd
\\
  &\lor&
    \cyc{
      \settorel{\FenceMd\sc}\seq\po\seq
      \rceleven.\eco\seq
      \po\seq\settorel{\FenceMd\sc}
    }
\\
  &\implies&
    \cyc{
      \settorel{\FenceMd\sc}\seq\hb\seq
      \rceleven.\eco\seq
      \hb\seq\settorel{\FenceMd\sc}
    }
\\
  &\implies&
    \cyc{
      \settorel{\FenceMd\sc}\seq\hb\seq
      \rceleven.\eco\seq
      \hb\seq\settorel{\FenceMd\sc}
    }
\\
  &\implies&
    \cyc{\rceleven.\psc}
\end{array}
\end{array}\]

\paragraph{Proof of Assertion~(3).}
Because~$\rcelevenext.\eco$
is included in~$\rceleven.\eco$,
the left-to-right implication is trivial.
The other direction follows by contradiction.
The proof exploits the equality~%
$\eco=\rf\disj(\mo\disj\rb)\seq\refC\rf$
and the inclusion~%
$\settorel{\EventE\setminus\WriteMd\nt}\seq\rfi\subseteq\hb$:
\[\begin{array}{@{}l@{}}
\;\notirr{\hb\seq\refC{\rceleven.\eco}}
\\
\begin{array}{@{}c@{\;\;}l@{}}
  \implies&
    \notirr{\overbrace{\hb\seq\refC\rfi}^{\subseteq\;\hb}}
      \;\lor\;
    \notirr{\hb\seq\refC\rfe}
\\
  \lor&
    \notirr{\hb\seq\refC{(\mo\disj\rb)}\seq\refC\rfe}
      \;\lor\;
    \underbrace{
    \notirr{\hb\seq\refC{(\mo\disj\rb)}\seq\refC\rfi}
    }_{\implies\;\notirr{\refC\rfi\seq\hb\seq\refC{(\mo\disj\rb)}}}
\\
  \implies&
    \notirr{\hb\seq\refC\rfe}
      \;\lor\;
    \notirr{\hb\seq\refC{(\mo\disj\rb)}\seq\refC\rfe}
\\
  \implies&
    \notirr{\hb\seq\refC{\rcelevenext.\eco}}
\end{array}
\end{array}\]

\paragraph{Proof of Assertion~(4).}
We show the following equality (which holds of \rceleven-events-only execution graphs):
\[\settorel{\FenceMd\sc}\seq\hb\seq\rceleven.\eco\seq\hb\seq\settorel{\FenceMd\sc}
=
\settorel{\FenceMd\sc}\seq\hb\seq\rcelevenext.\eco\seq\hb\seq\settorel{\FenceMd\sc}\]

Because~$\rcelevenext.\eco\subseteq\rceleven.\eco$,
it is easy to see that the relation on the right-hand side of the
equality is included
in the relation on the left-hand side.
Let us now show the inclusion in the other direction:
\[
\settorel{\FenceMd\sc}\seq\hb\seq\rceleven.\eco\seq\hb\seq\settorel{\FenceMd\sc}
  \subseteq
    \settorel{\FenceMd\sc}\seq\hb\seq\rcelevenext.\eco\seq\hb\seq\settorel{\FenceMd\sc}
\]

Exploiting the equality~%
$\rceleven.\eco=\rf\disj(\mo\disj\rb)\seq\refC\rf$
and the inclusion~%
$\settorel{\EventE\setminus\WriteMd\nt}\seq\rfi\subseteq\hb$,
it is then easy to see that
every edge~$\rfi$ in~$\rceleven.\eco$
can be merged into the~$\hb$ edge
that either precedes or succeeds~$\rceleven.\eco$.

\end{proof}

\begin{theorem}[\rcelevenext - Extension-II]
\label{thm:extension:II:rcelevenext:app}
\intelext-consistency conditions are equivalent
to those of \rcelevenext in every execution graph
containing only \intelext events.
\end{theorem}
\begin{proof}
The proof is split into two parts.
First we prove that \intelext-consistency implies \rcelevenext-consistency,
then we prove the converse:
\begin{enumerate}
\item %
\intelext-consistency~$\implies$ \rcelevenext-consistency.
\\

It is thus sufficient to prove that following conditions hold:
\begin{itemize}
\item \axiomCoherenceI. \\
The proof follows by contradiction.
We prove that,
in the absence of \rceleven events,
the violation of \axiomCoherenceI leads
to the violation of \intelext-consistency
or to the violation of \axiomCoherenceII
(which we show to hold in the next item):
\[\begin{array}{@{}r@{\;}l@{}}
  &\notirr{\hb\seq\refC\eco}
\\
  &\implies\;\,
    \notirr{\hb}
      \;\lor\;
    \notirr{\hb\seq\eco}
\\
  &\hrulefill
\\[3mm]
  &\notirr{\hb}
\\
  &\implies\;\,
    \cyc{
      \underbrace{\settorel{\RMWMd\tso\cup\ReadMd\tso}\seq\po}_{\subseteq\;\ppointel}
      \disj\rfe
    }
\\
  &\implies\;\,
    \cyc{\ppointel\disj\rfe\disj\moe\disj\rbe}
\\
  &\hrulefill
\\[3mm]
  &\notirr{\hb\seq\eco}
\\
  &\implies\;\,
    \notirr*{
      \begin{array}{@{}l@{}}
        \overbrace{
          \settorel{\codom\eco}\seq
          \porceleven\seq
           \refC{(
             \settorel{\FenceAtLeast\stf}\seq
             \po
           )}\seq
           \settorel{\WriteMd\tso}
         }^{\begin{array}{@{}l@{}}
        \subseteq\;\left(\begin{array}{@{}l@{}}
          \settorel{\RMWMd\tso\cup\ReadMd\tso}\seq\po               \;\cup\\
          \settorel{\WriteMd\tso}\seq
            \po\seq\settorel{\EventE\setminus\ReadE\setminus\WriteMd\nt}\;\cup\\
          \settorel{\WriteMd\nt}\seq
            \po\seq
            \settorel{\RMWMd\tso\cup\FenceAtLeast\stf}\seq
            \refC\po                                                    \;\cup\\
          \settorel{\WriteMd\nt}\seq\perloc\po
        \end{array}\right)
        \end{array}}
        \\
        \transC{(\rfe\seq\refC\po)}\seq
        \eco
      \end{array}}
\\
  &\implies\;\,
    \notirr{
      \ppoasm\seq
      \transC{(
        \rfe\seq\refC{(
          \underbrace{\settorel{\ReadMd\tso\cup\RMWMd\tso}\seq\po}_{
            \subseteq\;\ppoasm})}
      )}\seq
      \eco}
\\
  &\implies\;\,
    \cyc{\ppoasm\cup\eco}
\end{array}\]

\item \axiomCoherenceII. \\
It suffices to exploit the inclusion
$
  (\moi\disj\rbi)
    \subseteq
  \perloc\po\seq
  \settorel{\EventE\setminus\ReadE}
    \subseteq
  \ppointel
$
to show that the violation of \axiomCoherenceII leads
to the violation of \axiomExternal:
\[\begin{array}{@{~~~~~~~~}r@{\;}c@{\;}l@{}}
\cyc{\ppoasm\disj\eco}
  &\implies&
    \cyc{\ppoasm\disj\rfe\disj\mo\disj\rb}
\\
  &\implies&
    \cyc{
      (\underbrace{\ppoasm\disj\moi\disj\rbi}_{\subseteq\ppointel})
        \disj
       \rfe
        \disj
       \moe
        \disj
       \rbe
      }
\\
  &\implies&
    \cyc{\ppointel\disj\rfe\disj\moe\disj\rbe}
\end{array}\]

\item \axiomCoherenceIII. (Immediate from \axiomInternal.)

\item \axiomNTA. \\
We show that the negation of
\axiomNTA leads to a contradiction with \axiomExternal:
\[\begin{array}{@{}r@{\;}c@{\;}l@{}}
\cyc{\po\disj\rf}
  &\implies&
    \cyc{(\underbrace{\po\disj\rfi}_{\subseteq\;\po})\disj\rfe} \\
  &\implies&
    \cyc{(\underbrace{\settorel{\ReadMd\tso\disj\RMWMd\tso}\seq\po}_{\subseteq\,\ppointel})\disj\rfe} \\
  &\implies&
    \cyc{\ppointel\disj\rfe\disj\moe\disj\rbe}
\end{array}\]
\end{itemize}

\item %
\rcelevenext-consistency~$\implies$ \intelext-consistency. \\
Condition~\axiomExternal is an immediate consequence of~\axiomCoherenceII.
The proof of Condition~\axiomInternal is split into the three following
subconditions:
\begin{itemize}
\item $\irr{\po\seq\rfi}$. (Immediate from \axiomNTA.)

\item $\irr{\po\seq\moi}$. \\
The violation of this condition leads to a contradiction with
\axiomCoherenceII:
\[\begin{array}{@{}r@{\;}c@{\;\,}l@{}}
\notirr{\po\seq\moi}
&\implies&
  \notirr{
    (\underbrace{
       \perloc\po\seq
       \settorel{\EventE\setminus\ReadE}
     }_{\subseteq\;\ppoasm}
    )\seq
    \moi
  }
\\
&\implies&
  \notirr{\ppoasm\seq\eco}
\\
&\implies&
  \cyc{\ppoasm\disj\eco}
\end{array}\]
\item $\irr{\po\seq\rbi}$. (Immediate from \axiomCoherenceIII.)
\end{itemize}
\end{enumerate}
\end{proof}

\section{Compilation}

\subsection{Compilation Schemes}
\label{subsection:compilation:scheme:app}

We present two compilation schemes
from \rcelevenextlang to \intelextlang:
one that naturally extends the scheme studied by~\citet[Fig. 8]{rc11},
and a slightly more elaborated one that
maps \rlx writes to non-temporal stores,
and adds store fences
to the mapping of \rel/\sc writes and
to the mapping of \rel/\acqrel fences.
We prove that these schemes are correct (with respect to \rcelevenext)
in \sref{subsection:compilation:correctness:app}.

\begin{definition}[Compilation Scheme from \rcelevenextlang to \intelextlang]
\label{def:scheme:rcelevenext:app}
\[\begin{array}{@{}c@{}}
\begin{array}{@{}c@{~~}c@{~~}c@{}}
\begin{array}{@{}r@{}c@{\;}l@{}}
\compile{\WritePL\sc\exprA\exprB}
  &\eqdef&
    \Seq{\WritePL{}\exprA\exprB}{\MFencePL}
\\
\compile{\WritePL{\neq\sc}\exprA\exprB}
  &\eqdef&
    \WritePL{}\exprA\exprB
\\
\compile{\ReadPL\md\reg\expr}
  &\eqdef&
    \ReadPL{}\reg\expr
\end{array}
&
\begin{array}{@{}r@{}c@{\;}l@{}}
\compile{\FencePL\sc}
  &\eqdef&
    \MFencePL
\\
\compile{\FencePL{\neq\sc}}
  &\eqdef&
    \Skip
\\
\compile{\IfThen\expr\cmd}
  &\eqdef&
    \IfThen\expr{\compile\cmd}
\end{array}
&
\begin{array}{@{}r@{}c@{\;}l@{}}
\compile{\Seq\cmdA\cmdB}
  &\eqdef&
    \Seq{\compile\cmdA}{\compile\cmdB}
\\
\compile{\Skip}
  &\eqdef&
    \Skip
\\
\compile{\Asm[\cmd]}
  &\eqdef&
    \cmd
\end{array}
\end{array}
\\
\begin{array}{@{}c@{~~~~}c@{}}
\compile{\RMWPL\md\reg\exprOne\exprTwo\exprThree}
  \eqdef
    \RMWPL{}\reg\exprOne\exprTwo\exprThree
&
\compile{\While\expr\cmd}
  \eqdef
    \While\expr{\compile\cmd}
\end{array}
\end{array}\]
\end{definition}
\begin{definition}[Alternative Compilation Scheme]
\label{def:alt:scheme:rcelevenext:app}
Same as Def.~\ref{def:scheme:rcelevenext:app} except for the following cases:
\[\begin{array}{cc}
\begin{array}{r@{}c@{\;}l}
\compilealt{\WritePL{\sc}\exprA\exprB}
  &\eqdef&
    \Seq\SFencePL{
    \Seq{\WritePL{}\exprA\exprB}
        \MFencePL
    }
\\
\compilealt{\WritePL{\rel}\exprA\exprB}
  &\eqdef&
    \Seq\SFencePL{\WritePL{}\exprA\exprB}
\end{array}
&
\begin{array}{r@{}c@{\;}l}
\compilealt{\WritePL\rlx\exprA\exprB}
  &\eqdef&
    \NTWritePL\exprA\exprB
\\
\compilealt{\FencePL{\rel,\,\acqrel}}
  &\eqdef&
    \SFencePL
\end{array}
\end{array}\]
\end{definition}

\subsection{Mixed Execution Graphs}
\label{subsection:mixed:execution:graphs:app}

Our proofs of compilation correctness
(\cref{subsection:compilation:correctness:app})
rely on the novel notion of \textit{mixed execution graphs},
a type of execution graph whose nodes
contain events from both the source-level and target-level models.
Before presenting our compilation-correctness results,
let us give a brief introduction to mixed execution graphs.

Informally speaking,
a mixed execution graph is the
superposition of two execution graphs:
one called \textit{source graph},
which is associated with a source program~$\prog$;
and
one called \textit{target graph},
which is associated with the compilation of~$\prog$.
The key feature of a mixed execution graph is that
it captures the fact that source and target graphs
share the same overall structure.
Indeed,
because a compilation scheme preserves the
control flow of the source program and
changes only how memory operations are mapped
to operations in the target language,
for every execution graph of the compiled program,
one can always construct an execution graph of the
source program that preserves much of the structure
of the target graph,
including its primitive relations~$\po$,~$\rf$, and~$\mo$.
The only mismatches between these graphs come from how
one memory operation from the source language might be
mapped to
zero, one, or multiple memory operations from the target language.

To account for these mismatches,
nodes in a mixed graph,
called \textit{mixed nodes},
carry events from both source and target models.
Events from the two models however cannot be
arbitrarily assembled in a mixed node:
the source-level events in a mixed node correspond
to the events of a single source instruction
and the target-level events correspond to the
events emitted by the snippet of target-level language
produced by the mapping of this instruction.
Therefore,
the range of mixed nodes is fixed and
determined by the underlying compilation scheme.

Mixed graphs form a very convenient tool for proving
compilation-correctness results
because they allow
one to work with the execution graphs from both the source
program and its compiled version at the same time, and
because they allow one to forget about the compilation
scheme which is ultimately encoded in the set of permissible
mixed nodes.
Moreover, it is possible to lift the consistency conditions from
the models of source and target languages to this mixed-graph structure.
Both models can thus be defined on the same structure,
thereby allowing one to formally reason about statements of the kind
``\textit{one model is stronger than the other}''.
In fact, the main convenience of mixed execution graphs is
precisely to allow one to formulate the compilation correctness result
as a statement in this fashion:
``\textit{in a mixed execution graph with nodes taken from a well-chosen set,
if the consistency conditions of the target model hold,
then so do the consistency conditions of the source model}''.
The set of nodes has to be well chosen so as to correctly reflect
the compilation scheme begin considered.

To give an illustration of mixed execution graphs,
let us consider \rcelevenext as the source model,
\intelext as the target model,
and~$\compile\_$ (\defref{def:scheme:rcelevenext:app})
as the compilation scheme.

\input{figure-mixed-nodes-intelext-app}

\fref{fig:mnodes:intelext:app}
shows our choice for the set of permissible mixed nodes.
The nodes are depicted as domino-shaped boxes
where the left component stores \rcelevenext events
and the right component stores \intelext events.
We use the symbol~$\bot$ to denote an empty set of events
(in addition to its meaning as the \textit{none} element
of an option type).
It is easy to see how this definition mimics
the compilation scheme from~\defref{def:scheme:rcelevenext:app}.
Indeed, Node~\nodeRR
reflects how read instructions
are compiled to plain reads.
Node~\nodeWWMF
reflects how a \sc write
is compiled to a plain write followed by a memory fence.
Moreover, node~\nodeFB
reflects how fences weaker than \stf
are erased by the compilation scheme.
Nodes~\nodeWNTW and \nodeFSF
reflects the compilation of inline-assembly instructions.
Finally, nodes~\nodeRMWRMWS and~\nodeRMWRMWF
reflect the compilation of read-modify-writes.

To see an example of a mixed execution graph
constructed with these nodes,
consider the following program:
\[
\prog\eqdef
\left(\begin{array}{l@{~~~}||@{~~~}l}
\begin{array}{@{}l@{}}
\WritePL\sc\locx{1}\scol\\
\ReadPL\sc\reg\locy
\end{array}
&
\begin{array}{@{}l@{}}
\WritePL\sc\locy{1}\scol\\
\ReadPL\sc\reg\locx
\end{array}
\end{array}\right)
\]

This program implements the \textit{store-buffering litmus test} (SB).
SB is one of the simplest demonstrations of
non-sequentially consistent behaviors:
it would happen if both read instructions returned
the value~$0$.
The \rcelevenextlang program~$\prog$, however,
exhibits only sequentially consistent behaviors
because the access mode of all memory instructions is \sc.
The following mixed execution graph allows us to see
simultaneously
how \rcelevenext rules out SB in~$\prog$ and
how \intelext rules out SB in the compilation of~$\prog$,
the program~$\compile\prog$:
\begin{center}
\includegraphics{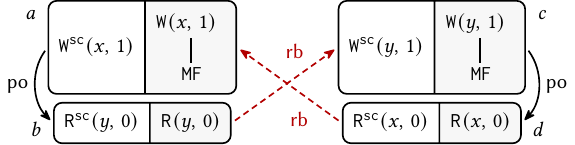}
\end{center}

To show that SB is ruled out (in both source and compiled programs),
we must show that both \rcelevenext and \intelext graphs are inconsistent.
Indeed,
both graphs are inconsistent because of the cycle~%
$\tuple\evA{\evB,\evC,\evD}$.
In the \intelext graph, this cycle violates Condition~\axiomExternal.
In the \rcelevenext graph, this cycle violates Condition~\axiomSC.
Here is a summary of the technical arguments sustaining these claims:
\[\begin{array}{@{}c@{}c}
\begin{array}{@{}l}
\text{Cycle in}\;\psc\,\text{:}
\\[2pt]
\begin{array}{@{\;}l}
  \bullet\;\pair\evA\evB,\,\pair\evC\evD\in
    {\REventMd\sc}\seq{\po\seq{\REventMd\sc}}
    \subseteq\pscBase
    \subseteq\psc
  \\[2pt]
  \bullet\;\pair\evB\evC,\,\pair\evD\evA\in
    {\REventMd\sc}\seq{\rb\seq{\REventMd\sc}}
    \subseteq\pscBase
    \subseteq\psc
\end{array}
\end{array}
&
\begin{array}{@{}l}
\text{Cycle in}\;\relExternalDef\,\text{:}
\\[2pt]
\begin{array}{@{\;}l}
  \bullet\;\pair\evA\evB,\,\pair\evC\evD\in\ppointel\\[2pt]
  \bullet\;\pair\evB\evC,\,\pair\evD\evA\in\rbe
\end{array}
\end{array}
\end{array}\]

At first glance,
the outlined arguments might seem
informal because the relations so specified
apply only to events of a specific model,
not to events of the mixed graph.
However,
we show that these arguments can be made valid:
in essence, it suffices to lift the relations from
source and target models to the structure of mixed graphs.
That is, when working with mixed graphs,
we manipulate custom versions of these relations
defined as relations on mixed nodes.

The following definition formally introduces mixed execution graphs
and its custom version of the relations~$\rf$ and~$\mo$:

\begin{definition}[Mixed Execution Graph]%
\label{def:mgraph:rcelevenext:intelext:app}
A~\emph{mixed execution graph} is a graph where every node,
called a~\emph{mixed node},
is a pair of a set of \rcelevenext events
and a set of \intelext events.
\fref{fig:mnodes:intelext:app}
depicts the set of mixed nodes allowed in a mixed graph.
The two basic relations are~$\rf$ and~$\mo$:
\begin{enumerate}
\item \emph{Reads-from ($\rf$).}
The reads-from relation is a surjective and functional
relation with domain and codomain specified as follows:

\includegraphics{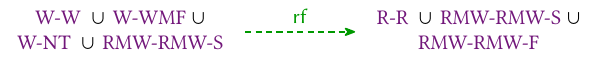}

\item \emph{Modification-order ($\mo$).}
The modification-order has domain and codomain specified as follows:

\includegraphics{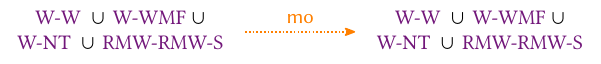}

\end{enumerate}
We introduce the following sets of mixed nodes:
\[\begin{array}{@{}l@{\hspace{2cm}}r@{}}
\begin{array}{@{}r@{\;}c@{\;\;}l@{}}
\WriteE
  &\eqdef&
    \textnormal\nodeWW    \,\cup\,
    \textnormal\nodeWWMF  \,\cup\,
    \textnormal\nodeWNTW
\\
\RMWExtE[]
  &\eqdef&
    \textnormal\nodeRMWRMWS \,\cup\,
    \textnormal\nodeRMWRMWF
\\
\ReadExtE[]
  &\eqdef&
    \textnormal\nodeRR
\\
\FenceExtE[]
  &\eqdef&
    \textnormal\nodeFMF   \,\cup\,
    \textnormal\nodeFSF   \,\cup\,
    \textnormal\nodeFB
\end{array}
&
\begin{array}{@{}r@{\;}c@{\;\;}l@{}}
\NTWriteExtE[]
  &\eqdef&
    \textnormal\nodeWNTW
\\
\SFenceExtE[]
  &\eqdef&
    \textnormal\nodeFSF
\\
\MFenceExtE[]
  &\eqdef&
    \textnormal\nodeFMF
\\
  &&
\end{array}
\end{array}\]
When applicable,
we annotate sets of mixed nodes with superscripts of the form
``$\sqsupseteq\md$'' (and variations of it)
to specify the range of access modes
of the \rceleven events in the left component of mixed nodes.
\end{definition}

Naturally,
reasoning at the level of mixed graphs and its corresponding
version of the relations~$\rf$ and~$\mo$ leads to
facts about mixed graphs and nodes;
to extract a result about the source and target models,
we provide a theorem that allows one to transfer
results between these structures.
For example, we prove that,
if the consistency conditions of \rcelevenext hold of a mixed graph,
then it also holds of the source graph.

Before we introduce this theorem,
let us clarify the notions of
\textit{source graph} and \textit{target graph}.
These concepts are not yet well defined,
because, given a mixed graph,
we have not explained how they can be obtained.
The missing piece of information is the notion of
\textit{source} and \textit{target projections}:
given a mixed graph~$\graphM$,
its source projection~$\lproj\graphM$
and target projection~$\rproj\graphM$
correspond to the source and target graphs
whose superposition is~$\graphM$.

\input{figure-projections-intelext-app}

\begin{definition}[Projections]
\label{def:graph:proj:app}
Let~$\graphM$ be a mixed execution graph.
The \emph{source} and \emph{target projections} of~$\graphM$,
noted~$\lproj\graphM$ and~$\rproj\graphM$,
are \rcelevenext and \intelext execution graphs.
The nodes of~$\lproj\graphM$ and~$\rproj\graphM$ correspond
to the first and second parts of~$\graphM$ nodes.
The edges of~$\lproj\graphM$ and~$\rproj\graphM$ 
are constructed through \emph{projection rules}.
A projection rule formalizes the correspondence
between edges in~$\graphM$
(appearing at the top of the rule)
and
the edges in~$\lproj\graphM$ and~$\rproj\graphM$
(appearing respectively at the bottom left
and at the bottom right of the rule).
\fref{fig:proj:intelext:app}
shows a selection of the projection rules.
\end{definition}

\begin{definition}[Mixed Execution Graph - \intelext-Consistency]
\label{def:mixed:intelext:consistent:app}
A mixed execution graph
$\tuple\graphM{\rf,\mo}$ is \emph{\intelext-consistent}
if the conditions from \defref{def:intelext:consistent:app} hold
when the sets and relations to their corresponding mixed-graph versions
as introduced in \defref{def:mgraph:rcelevenext:intelext:app},
and the~$\ppo$ relation with the following one:
\[\begin{array}{@{}r@{\;}c@{\;\;}l@{}}
\ppo
  &\eqdef&
    \settorel{\EventE\setminus\nodeFB}\seq
    \po\seq
    \settorel{\RMWE\cup\MFence\cup\SFenceE}
\\
  &\;\cup&
    \settorel{\ReadE\cup\RMWE\cup\MFence}\seq
    \po\seq
    \settorel{\EventE\setminus\nodeFB}
\\
  &\;\cup&
    \settorel{\EventE\setminus\nodeFB}\seq
    \refC\po\seq
    \settorel{\nodeWWMF}\seq
    \po\seq
    \settorel{\EventE\setminus\nodeFB}
\\
  &\;\cup&
    \settorel{\SFenceE}\seq\po\seq\settorel{\EventE\setminus\nodeFB\setminus\ReadE}
\\
  &\;\cup&
    \settorel{\WriteMd{\neq\nt}}\seq\po\seq\settorel{\WriteMd{\neq\nt}}
\\
  &\;\cup&
    \settorel{\WriteE}\seq
    \perloc\po\seq
    \settorel{\WriteE}
\end{array}\]
\end{definition}

\begin{definition}[Mixed Execution Graph - \rceleven-Consistency]
\label{def:mixed:rcelevenext:consistent:app}
A mixed execution graph
$\tuple\graphM{\rf,\mo}$ is \emph{\rcelevenext-consistent}
if the conditions from \defref{def:rcelevenext:consistent:app} hold
when we replace the sets and relations to their corresponding mixed-graph versions
as introduced in \defref{def:mgraph:rcelevenext:intelext:app}.
\end{definition}

\begin{theorem}[Transfer Principle]
\label{thm:transfer:app}
Let~$\tuple\graphM{\rf,\mo}$ be a mixed execution graph.
The consistency conditions from
\defref{def:mixed:rcelevenext:consistent:app}
hold of~$\graphM$ if, and only if,
the \rcelevenext-consistency conditions
(\defref{def:rcelevenext:consistent:app})
hold of~$\lproj\graphM$.
Analogously,
the consistency conditions from
\defref{def:mixed:intelext:consistent:app}
hold of~$\graphM$ if, and only if,
the \intelext-consistency conditions
(\defref{def:intelext:consistent:app})
hold of~$\rproj\graphM$.
\end{theorem}
\begin{proof}
The first claim,
that~$\graphM$ is \rcelevenext-consistent
iff~$\lproj\graphM$ is \rcelevenext-consistent,
is easy to see because every mixed node carries exactly
one \rcelevenext event.
Therefore, the projected nodes and relations
can be related by a one-to-one correspondence.
Since the consistency conditions
from~\defref{def:mixed:rcelevenext:consistent:app}
are essentially the same as those from \rcelevenext
(\defref{def:rcelevenext:consistent:app}),
this observation is sufficient to establish this claim.

The second claim,
that~$\graphM$ is \intelext-consistent
iff~$\rproj\graphM$ is \intelext-consistent,
is slightly more intricate to prove than the previous one
because of the nodes \nodeFB and \nodeWWMF,
which do not have a one-to-one correspondence
with the projected \intelext events;
and because of the differences between
the two versions of~$\ppo$ from
Definitions
\ref{def:intelext:consistent:app}
and
\ref{def:mixed:intelext:consistent:app}.
That the \intelext-consistency of~$\rproj\graphM$
implies that of~$\graphM$
follows from the fact that
the projection of every
edge in~$\graphM.\ppo$ is an edge in~$(\rproj\graphM).\ppo$.
Indeed, it is easy to see that the two
problematic types of nodes, \nodeFB and \nodeWWMF,
are correctly handled by the definition of~$\graphM.\ppo$:
nodes of type \nodeFB are excluded from~$\graphM.\ppo$,
and
nodes of type \nodeWWMF always have a trailing~$\po$
edge so that a memory fence is always between the
two endpoints of the resulting projected edge.
This concludes one direction of the logical equivalence.
To prove the converse,
that the \intelext-consistency of~$\graphM$
implies that of~$\rproj\graphM$,
it suffices to show that every $(\rproj\graphM).\transC\ppo$
edge in a cycle that violates \axiomExternal
is the projection of a $\graphM.\transC\ppo$ edge.
This condition can be easily checked;
the only non-trivial case is
when there is a memory-fence event between the endpoints of
a $(\rproj\graphM).\transC\ppo$ edge,
because this memory fence could be the projection
of either a $\nodeFMF$ node or a $\nodeWWMF$ node.
In both cases,
it is easy to see that the $(\rproj\graphM).\transC\ppo$
edge is the projection of an edge of type
$\graphM.(
\settorel{\EventE\setminus\nodeFB}\seq
\refC\po\seq
\settorel{\MFenceE\cup\nodeWWMF}\seq
\po\seq
\settorel{\EventE\setminus\nodeFB}
)$,
which is included in~$\graphM.\transC\ppo$.
\end{proof}

To conclude this discussion,
we introduce the notion of \textit{graph simulation},
an auxiliary concept for our upcoming compilation-correctness
proofs (\sref{subsection:compilation:correctness:app}):
\begin{definition}[Graph Simulation]
\label{def:gsim:app}
A \rcelevenext graph~$\graphA$ is \emph{simulated by}
a \intelext graph~$\graphB$,
noted~$\graphA\gsim\graphB$,
if there exists a mixed execution graph~$\graphM$
such that~$\lproj\graphM = \graphA$ and~$\rproj\graphM = \graphB$.
\end{definition}

\subsection{Compilation Correctness}
\label{subsection:compilation:correctness:app}

We state and prove correctness of the
compilation schemes from
Definitions~\ref{def:scheme:rcelevenext:app}
and~\ref{def:alt:scheme:rcelevenext:app}.
The statement of correctness is straightforward:

\begin{theorem}[Correctness of \defref{def:scheme:rcelevenext:app}]
\label{thm:correctness:rcelevenext:app}
For every program~$\prog$,
the set of behaviors of~$\compile\prog$
defined by \intelext
is included in the set of behaviors of~$\prog$
defined by \rcelevenext:
\[
\forall\prog.\;
\interp[\intelext]{\compile\prog}
\subseteq
\interp[\rcelevenext]{\prog}
\]
\end{theorem}

\begin{theorem}[Correctness of \defref{def:alt:scheme:rcelevenext:app}]
\label{thm:correctness:rcelevenext:alt:app}
Statement analogous to \thmref{thm:correctness:rcelevenext:app}
\end{theorem}

\subsubsection{Proof Sketch.}

The overall structure of our proofs
is depicted by the following diagram:
\begin{center}
\includegraphics{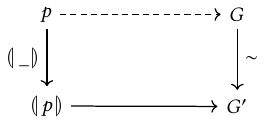}
\end{center}

It illustrates the first step
of a two-steps strategy to prove that~$\compile\_$ is correct.

This first step consists of showing that,
for every program~$\prog$,
for every execution graph~$\graphB$
associated with~$\compile\prog$,
there exists a graph~$\graphA$ associated with~$\prog$,
such that~$\graphA\gsim\graphB$.

The second step is then to show that,
if the consistency conditions
from the target model hold of a mixed execution graph~$\graphM$,
then the consistency conditions
from the source model also hold of~$\graphM$.
Finally,
by invoking Theorem~\ref{thm:transfer:app},
it follows that,
if~$\graphB$ is consistent
(with respect to the target model),
then~$\graphA$ is consistent
(with respect to the source model).
In particular,
this implies that
every final state of~$\compile\prog$
is a final state of~$\prog$,
which statement corresponds precisely
to the formulation of compilation correctness.

The first step is accomplished by induction over
the construction of the graph~$\graphB$.
Intuitively,
because the compiled program~$\compile\prog$
preserves much of the structure of~$\prog$,
it is possible to replay the pool-reduction steps
(Figures~\ref{fig:pool:reduction:partI:app}
and~\ref{fig:pool:reduction:partII:app})
from~$\compile\prog$ and yield a graph~$\graphA$
that satisfies the desired properties.

The second step is the crux of our proofs and
it is where we concentrate our attention.
Next, we discuss how to accomplish this
step in the case of \intelext,
first considering the standard compilation scheme
(\defref{def:scheme:rcelevenext:app})
and then the alternative one
(\defref{def:alt:scheme:rcelevenext:app}).

\subsubsection{Compilation to \intelext.}

\begin{lemma}[\rcelevenext-Weaker-Than-\intelext]
\label{lemma:rcelevenext:weaker:intelext}
Let~$\graphM$ be a mixed execution graph.
If~$\graphM$ is \intelext-consistent,
then~$\graphM$ is \rcelevenext-consistent.
\end{lemma}
\begin{proof}
\input{proof-compilation-intelext}

\end{proof}

\subsubsection{Alternative Compilation to \intelext.}

\input{figure-mixed-nodes-intelext-alt-app}

To apply our methodology of mixed execution graphs
to show the correctness of the alternative compilation scheme
(\defref{def:alt:scheme:rcelevenext:app}),
we need to complete the three following preliminary steps:
\begin{enumerate}
\item Define the set of mixed nodes that reflect
        the alternative compilation scheme.
\item State \rcelevenext-consistency and \intelext-consistency
        of mixed execution graphs containing this new set of nodes.
\item Prove the transfer principle for these new consistency
        definitions.
\end{enumerate}

\fref{fig:mnodes:intelext:alt:app}
depicts a selection of the updated mixed nodes
reflecting the alternative compilation scheme.
The complete set of nodes is the same as
\fref{fig:mnodes:intelext:app}
with the following exceptions:
(1) nodes of type~\nodeWWMF are replaced with~\nodeWSFWMF,
(2) nodes of type~\nodeFB with access mode~$\rel/\acqrel$
      are replaced with~\nodeFSFalt, and
(3) nodes of type~\nodeWW with access modes~$\rlx$ and~$\rel$
      are respectively replaced with~\nodeWNTWalt and~\nodeWSFW.

To make the distinction between
mixed graphs composed of nodes
as defined in~\fref{fig:mnodes:intelext:app}
and
mixed graphs composed of nodes
as defined in~\fref{fig:mnodes:intelext:alt:app}
clear,
we call the later \textit{alternative mixed execution graphs}.
To state \rcelevenext-consistency for alternative mixed graphs,
it suffices to update the notation introduced in
\defref{def:mixed:intelext:consistent:app}:
\[\begin{array}{@{}l@{\hspace{2cm}}r@{}}
\begin{array}{@{}r@{\;}c@{\;\;}l@{}}
\WriteE
  &\eqdef&
    \textnormal\nodeWW    \,\cup\,
    \mathhl{\textnormal\nodeWSFWMF}\,\cup\,
    \textnormal\nodeWNTWalt
\\
\RMWExtE[]
  &\eqdef&
    \textnormal\nodeRMWRMWS \,\cup\,
    \textnormal\nodeRMWRMWF
\\
\ReadExtE[]
  &\eqdef&
    \textnormal\nodeRR
\\
\FenceExtE[]
  &\eqdef&
    \textnormal\nodeFMF   \,\cup\,
    \textnormal\nodeFSFalt   \,\cup\,
    \textnormal\nodeFB
\end{array}
&
\begin{array}{@{}r@{\;}c@{\;\;}l@{}}
\NTWriteExtE[]
  &\eqdef&
    \textnormal\nodeWNTWalt
\\
\SFenceExtE[]
  &\eqdef&
    \textnormal\nodeFSFalt
\\
\MFenceExtE[]
  &\eqdef&
    \textnormal\nodeFMF
\\
  &&
\end{array}
\end{array}\]

The statement of \rcelevenext-consistency thus corresponds
to \defref{def:rcelevenext:consistent:app}
when we replace the sets and relations to their corresponding
mixed-graph versions just introduced.
The statement of~\intelext-consistency however
needs more attention,
so we state it in a separate definition:
\begin{definition}[Alternative Mixed Execution Graph - \intelext-Consistency]
\label{def:alt:mixed:intelext:consistent:app}
An alternative mixed execution graph
$\tuple\graphM{\rf,\mo}$ is \emph{\intelext-consistent}
if the conditions from \defref{def:intelext:consistent:app} hold
when we replace~$\rf$,~$\mo$, and the sets of nodes with the ones
just introduced and the~$\ppo$ relation with the following one:
\[\begin{array}{@{}r@{\;}c@{\;\;}l@{}}
\ppo
  &\eqdef&
    \settorel{\EventE\setminus\nodeFB}\seq
    \po\seq
    \settorel{\RMWE\cup\MFence\cup\SFenceE\cup\nodeWSFWMF}
\\
  &\;\cup&
    \settorel{\ReadE\cup\RMWE\cup\MFence\cup\nodeWSFWMF}\seq
    \po\seq
    \settorel{\EventE\setminus\nodeFB}
\\
  &\;\cup&
    \settorel{\EventE\setminus\nodeFB}\seq
    \po\seq
    \settorel{\nodeWSFW}\seq
    \refC\po\seq
    \settorel{\EventE\setminus\nodeFB}
\\
  &\;\cup&
    \settorel{\SFenceE}\seq\po\seq\settorel{\EventE\setminus\nodeFB\setminus\ReadE}
\\
  &\;\cup&
    \settorel{\WriteMd{\neq\nt}}\seq\po\seq\settorel{\WriteMd{\neq\nt}}
\\
  &\;\cup&
    \settorel{\WriteE}\seq
    \perloc\po\seq
    \settorel{\WriteE}
\end{array}\]
\end{definition}

Finally, we state and prove the corresponding transfer principle
for alternative mixed graphs:
\begin{theorem}[Transfer Principle]
\label{thm:transfer:alt:app}
Let~$\tuple\graphM{\rf,\mo}$ be an alternative mixed execution graph.
The \rcelevenext-consistency conditions
hold of~$\graphM$ iff
they hold of~$\lproj\graphM$.
Analogously,
the \intelext-consistency conditions
hold of~$\graphM$ iff
they hold of~$\rproj\graphM$.
\end{theorem}
\begin{proof}
As in the proof of~\thmref{thm:transfer:app},
that~$\graphM$ is \rcelevenext-consistent iff~$\lproj\graphM$ is
\rcelevenext-consistent,
is straightforward,
because there is a one-to-one correspondence
between~$\graphM$ and~$\lproj\graphM$ and
because \rcelevenext-consistency conditions are
equivalently defined for both graphs.

We now prove
that~$\graphM$ is \intelext-consistent iff~$\rproj\graphM$ is
\intelext-consistent.
By studying the definition of~$\ppo$
from \defref{def:alt:mixed:intelext:consistent:app},
it is easy to see that every $\graphM.\ppo$ edge
is projected to a $(\rproj\graphM).\transC\ppo$ edge.
To give an example,
edges of type~%
\[\graphM.(
\settorel{\EventE\setminus\nodeFB}\seq
\po\seq
\settorel{\nodeWSFW}\seq
\refC\po\seq
\settorel{\EventE\setminus\nodeFB\setminus\ReadE}
)\]
are projected to edges of type~%
\[(\rproj\graphM).(
\underbrace{%
\po\seq\settorel{\SFence}%
}_{\subseteq\;\ppo}
\seq
\underbrace{%
\settorel{\SFence}\seq
\po\seq
\settorel\WriteE\seq
\refC\po\seq
\settorel{\EventE\setminus\ReadE}
}_{\subseteq\;\ppo}
),\]
who belong to~$(\rproj\graphM).\transC\ppo$.
Therefore, if the~$\intelext$-consistency
conditions hold of~$\rproj\graphM$
they must hold of~$\graphM$.

To show the converse,
it suffices to check that every edge of type
\[R\eqdef(\rproj\graphM).(
\settorel{\EventE\setminus\MFenceE\setminus\SFenceE}\seq
\transC\ppo\seq
\settorel{\EventE\setminus\MFenceE\setminus\SFenceE})\]
is the projection of an edge of type~$\graphM.\transC\ppo$.
(It is sound to restrict our attention to edges
that do not start or end in a fence,
because only this type of edge can be used to form
cycles that violate~\axiomExternal.)
Let~$\pair\evA\evB$ be an edge in~$R$.
The proof that~$\pair\evA\evB$ is the projection
of an edge of type~$\graphM.\transC\ppo$
goes by disjunction of cases on whether there
is a fence or a read-modify-write event
between~$\evA$ and~$\evB$.

In the negative case,
the~$(\rproj\graphM).\ppo$ edges between~$\evA$ and~$\evB$
are of type either
\[(\rproj\graphM).(
\settorel{\intelext.\WriteE}\seq
\po\seq
\settorel{\intelext.\WriteE})
\quad\text{or}\quad
(\rproj\graphM).(
\settorel{\intelext.(\WriteE\cup\NTWriteE)}\seq
\perloc\po\seq
\settorel{\intelext.(\WriteE\cup\NTWriteE)}).\]
It is easy to see that, in this case, events in~%
$\dom{(\rproj\graphM).(
\settorel{\WriteE\cup\NTWriteE}\seq
\ppo\seq\settorel{\evB})}$
come from the projection of mixed nodes of type either~%
$\nodeWW$,or~$\nodeWNTWalt$;
and that~$(\rproj\graphM).\ppo$ edges between these nodes
correspond to either~%
\[\graphM.(
\settorel{\WriteMd{\neq\nt}}\seq
\po\seq
\settorel{\WriteMd{\neq\nt}}
)
\quad\text{or}\quad
\graphM.(
\settorel{\WriteE}\seq
\perloc\po\seq
\settorel{\WriteE}
),\]
both of which are included in~$\graphM.\ppo$.

In the affirmative case,
there must be at least two~$(\rproj\graphM).\ppo$
edges between~$\evA$ and~$\evB$;
then the proof follows by induction on the number~$n$
of the remaining~$(\rproj\graphM).\ppo$ edges.
In the inductive case,
we can assume that there is only one fence or read-modify-write event
between~$\evA$ and~$\evB$ and that this event is the target
of the immediate~$(\rproj\graphM).\ppo$ coming out from~$\evA$,
because, otherwise, the edge~$\pair\evA\evB$ would fit
into a smaller number of~$(\rproj\graphM).\ppo$ edges
and the inductive hypothesis would be applicable.
\end{proof}

\begin{lemma}[\rcelevenext-Weaker-Than-\intelext-Alt]
\label{lemma:rcelevenext:weaker:intelext:alt}
Let~$\graphM$ be an alternative mixed execution graph
(that is, a mixed graph formed of nodes as specified in~%
\fref{fig:mnodes:intelext:alt:app}).
If~$\graphM$ is \intelext-consistent,
then~$\graphM$ is \rcelevenext-consistent.
\end{lemma}
\begin{proof}
\input{proof-compilation-intelext-alt}

\end{proof}

\subsection{Compiler Optimizations}

\subsubsection{Register Promotion.}
The following counterexample shows that
applying this optimization to inline-assembly
read-modify-writes is unsound:
\[\begin{array}{@{}ccc@{}}
\inarrII{
  \AsmNTWritePL\locx{1}                 \\
  \AsmRMWPL\regA\locz{0}{1}\outcomeI{0} \\
  \WritePL\rel\locy{1}
}{
  \ReadPL\acq\regB\locy\outcomeI{1} \\
  \ReadPL\rlx\regC\locx\outcomeI{1} \\
}
&
\rightsquigarrow
&
\inarrII{
  \AsmNTWritePL\locx{1}       \\
  \regD\kw{:=}\,0             \\
  \regA\kw{:=}\,\regD         \\
  \regD\kw{:=}\,1             \\
  \WritePL\rel\locy{1}
}{
  \ReadPL\acq\regB\locy\outcomeC{1} \\
  \ReadPL\rlx\regC\locx\outcomeC{1} \\
}
\end{array}\]

\begin{theorem}[Register Promotion]
\label{thm:register:promotion:app}
The transformation that promotes accesses to a register
a location used by only one thread and not via
inline-assembly read-modify-writes
is sound:
for every~$\progA$ and~$\progB$,
if~$\transf\progA\progB$,
then~$\interp\progB\subseteq\interp\progA$.
\end{theorem}
\begin{proof}
Let~$\locz$ be the location that is promoted to a register,
let~$\intj$ be the identifier of the thread
to which this location belongs,
and let~$\graphB$ be a %
execution graph associated with~$\progB$.
Let~$\graphA$ be a graph associated with~$\progA$,
obtained by extending~$\graphB.\EventE$ with
the missing~$\locz$ accesses
and by extending~$\graphB.\mo,\rf$
in such a way that~%
$\graphA.\atloc\mo\locz,\atloc\rb\locz,\atloc\rf\locz\subseteq\graphA.\po$
(this is possible because
the register instructions in~$\progB$,
to which the accesses to~$\locz$ were promoted,
are executed in order).
We show that, if~$\graphB$ is \rcelevenext-consistent,
then so is~$\graphA$,
and that, if~$\graphB$ is racy, then so is~$\graphA$.

The proof relies on the fact that~%
$\settorel{\graphB.\EventE}\seq
\graphA.(\transC\porceleven)\seq
\settorel{\graphB.\EventE}$
is included in~%
$\graphB.(\transC\porceleven)$.
Indeed,
let~$\pair\evA\evB$ be a pair in~$\graphB.\EventE$,
such that~%
$\pair\evA\evB\in\graphA.(\transC\porceleven)$.
We show that~%
$\pair\evA\evB\in\graphB.(\transC\porceleven)$.
If~$\evA\notin\WriteMd\nt$, then this is clearly the case.
Now, suppose that~$\evA\in\WriteMd\nt$.
Consequently,
for~$\pair\evA\evB$ to be included in~%
$\graphA.(\transC\porceleven)$,
there must be an event~$\evC\in\RMWMd\tso\disj\FenceAtLeast\stf$,
such that~$\pair\evA\evC,\pair\evC\evB\in\graphA.\porceleven$.
Such an event is not an access to~$\locz$,
because~$\locz$ is not accessed through inline-assembly
read-modify-writes
and because a fence is not an access to~$\locz$.
This event is thus included in~$\graphB.\EventE$.
It is then easy to see that~%
$\pair\evA\evB\in\graphB.(\transC\porceleven)$.

From this fact, it follows that~%
$\settorel{\graphB.\EventE}\seq
\graphA.\hb\seq
\settorel{\graphB.\EventE}$
is included in~$\graphB.\hb$.

We now prove that,
if~$\graphB$ is racy, then so is~$\graphA$.
If~$\graphB$ is racy,
then there are~$\evA,\evB\in\graphB.\EventE$,
such that~$\pair\evA\evB\in\raceS{\graphB.\hb}$.
To show that~$\pair\evA\evB\in\raceS{\graphA.\hb}$,
it suffices to show that~%
$\pair\evA\evB\notin\graphA.\hb\disj\graphA.\inv\hb$.
Suppose, by contradiction, and without loss of generality,
that~$\pair\evA\evB\in\graphA.\hb$.
Then, from the inclusion~%
$\settorel{\graphB.\EventE}\seq
\graphA.\hb\seq
\settorel{\graphB.\EventE}\subseteq\graphB.\hb$,
it follows that~$\pair\evA\evB\in\graphB.\hb$,
a contradiction to~$\pair\evA\evB\in\raceS{\graphB.\hb}$.

Now, suppose that~$\graphB$ is \rcelevenext-consistent.
We prove that so is~$\graphA$.
Suppose, by contradiction, this is not the case,
and let~$\cycle$ be a cycle violating
\rcelevenext-consistency in~$\graphA$.
This cycle must contain at least one access to~$\locz$,
because, if every event in~$\cycle$ is included in~$\graphB.\EventE$,
then~$\cycle$ is also a violation to \rcelevenext
in~$\graphB$.
The cycle is either contained in thread~$\intj$
or it spans over more than one thread.
If the cycle spans over more than one thread,
then every access to~$\locz$ in~$\cycle$
must be surrounded by two accesses to other locations
that are distinct from~$\locz$
(but not necessarily between themselves),
because~$\locz$ is not shared among threads.
The accesses to~$\locz$ in~$\cycle$ can thus be avoided,
thereby yielding a consistency-violating cycle in~$\graphB$,
a contradiction to its consistency.
If the cycle is contained in thread~$\intj$,
then there must be two accesses~%
$\evA,\evB\in\cycle$,
such that~$\pair\evA\evB\in\graph.\po$
and~$\pair\evB\evA\in\graph.(\rf\disj\mo\disj\rb)$.
These accesses cannot be to~$\locz$,
because~%
$\graphA.\atloc\mo\locz,\atloc\rb\locz,\atloc\rf\locz\subseteq\graphA.\po$.
Therefore, this
consistency-violating pair of events also belongs to~$\graphB$,
a contradiction to its consistency.
\end{proof}

\subsubsection{Sequentialization.}
The following counterexample shows that
sequentialization is unsound in~\intelext
(and therefore also unsound in~\rcelevenext
because of~\hyperlink{P4}{Property P4}):
\[\begin{array}{@{}ccc@{}}
\inarrIII{
  \WritePL{}\locx{1}
}{
  \ReadPL{}\regA\locx\outcomeI{1} \\
  \ReadPL{}\regB\locy\outcomeI{0}
}{
  \WritePL{}\locy{1}     \\
  \MFencePL              \\
  \ReadPL{}\regC\locx\outcomeI{0}
}
&
\rightsquigarrow
&
\inarrII{
  \WritePL{}\locx{1}              \\
  \ReadPL{}\regA\locx\outcomeC{1} \\
  \ReadPL{}\regB\locy\outcomeC{0}
}{
  \WritePL{}\locy{1}     \\
  \MFencePL              \\
  \ReadPL{}\regC\locx\outcomeC{0}
}
\\[10mm]
\begin{array}{@{}c@{}}
\centering
{\includegraphics{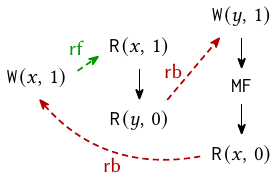}}
\end{array}
&
&
\begin{array}{@{}c@{}}
\centering
\includegraphics{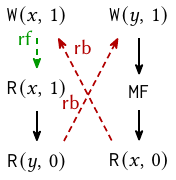}
\end{array}
\end{array}\]

\begin{theorem}[NITIA-Sequentialization]
\label{thm:nitia:sequentialization:app}
The transformation that merges two threads that satisfy NITIA is sound:
for every~$\progA$ and~$\progB$,
if~$\transf\progA\progB$,
then~$\interp\progB\subseteq\interp\progA$.
\end{theorem}
\begin{proof}
Let~$\graphB$ be an execution graph associated with~$\progB$,
and let~$\graphA$ be an execution graph associated with~$\progA$
such that~$\transf\graphA\graphB$.
We show that,~%
(1) if~$\graphB$ is consistent, then so is~$\graphA$;
and that,~%
(2) if~$\graphB$ is racy, then so is~$\graphA$.

The proof of~(1) follows by contradiction:
suppose that~$\graphA$ is \rcelevenext-inconsistent,
but~$\graphB$ is not.
Then there exists a problematic edge between
the two threads being merged.
We saw that such an edge must contain at least
one inline-assembly event.
This contradicts the NITIA condition,
which prevents the existence of~$\rfe$ edges
between such accesses.

To show~(2),
it suffices to notice that~%
$\graphA.\hb$ is included in $\graphB.\hb$.
Indeed, this inclusion holds
because the definition of~$\hb$
does not distinguish the internal
from the external components of~$\rf$
and because~$\graphA.\po$ is included in~$\graphB.\po$.
\end{proof}

\begin{theorem}[Fence-Sequentialization]
\label{thm:fence:sequentialization:app}
The transformation that merges two threads
by inserting a~$\sc$ fence between them
is sound:
for every~$\progA$ and~$\progB$,
if~$\transf\progA\progB$,
then~$\interp\progB\subseteq\interp\progA$.
\end{theorem}
\begin{proof}
It suffices to notice that
an external edge~$\pair\evA\evB\in\graphA.\rfe$
becomes part of both~$\graphB.\transC\porceleven$ and~%
$\graphB.\transC\ppoasm$,
because
a~$\sc$ fence is inserted between~$\evA$ and~$\evB$:
$\pair\evA\evB\in
\graphB.(\po\seq\settorel{\FenceMd\sc}\seq\po)\subseteq
\graphB.(\transC\porceleven\cap\transC\ppoasm)$.
\end{proof}

\subsubsection{Deordering.}

\input{figure-deordering-app}

Deordering transforms sequential composition
into parallel composition:
\[\transf{\Seq\cmdA\cmdB}{\inarrII\cmdA\cmdB}\]

Because, for simplicity,
\rcelevenextlang does not support
the dynamic allocation of threads,
we can only express this optimization
as a transformation that converts a
single two-instructions thread into
two one-instruction threads.
At the level of execution graphs, however,
we are able to express a more general transformation
that removes a~$\po$ edge between two given
events~$\evA$ and~$\evB$ but keeps every other~$\po$
edge from and to these events:
\[\begin{array}{@{}c@{\quad\quad}c@{\quad\quad}c@{}}
\begin{array}{@{}c@{}}
\includegraphics{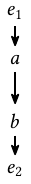}
\end{array}
&\rightsquigarrow&
\begin{array}{@{}c@{}}
\includegraphics{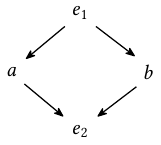}
\end{array}
\end{array}\]

\fref{fig:deordering:app}
shows the pairs
of \rceleven events~$\evA$ and~$\evB$ for which this
transformation is correct.
The Figure corresponds to Table 1 from \citet{rc11}.
Because events are restricted to plain \rceleven,
and because every~$\ppoasm$ edge contains at least one \asm event,
this transformation has no effect on~$\ppoasm$.
Moreover, no external edge is removed by deordering.
Therefore,
the restriction of~$\rf$ to~$\rfe$ is not problematic
and
cycles in~$\ppoasm\cup\eco$ cannot be undone
by deordering.
In sum, verifying the correctness of deordering is
limited to studying the effects of this transformation to
the~$\porceleven$ component of~$\hb$ and to the cycles
that violate the standard \rceleven conditions;
it can thus follow the
same arguments as those conveyed by \citet{rc11}.

\subsubsection{Merging.}

\input{figure-merging-app}

Merging transforms two consecutive instructions into one.
\fref{fig:merging:app} depicts the pairs of events
that can be merged.
These pairs impose constraints on the access modes of the events,
but, when performed after strengthening,
this transformation can be applied to any
pair of instructions whose accesses are at least the
ones so specified.
The remarks from the previous discussion
(about the correctness of deordering)
also apply here:
roughly speaking,
because inline-assembly events are left intact,
the correctness of merging follows essentially
the same arguments as those presented by \citet{rc11}.

%% file: figure-cpp-op-semantics-app.tex
\let\oldRefTirName\RefTirName
\let\RefTirName\TirName

\begin{figure}[H]

\raggedright

\judgment[Pool reduction]{\poolstep\pool\graph\pool\graph}

\begin{mathpar}
    \inferrule[\nameReadStep]{\begin{array}{@{}c@{}}
      \poolLkp>*>*>*[\Seq{\ReadPL\md\reg\expr}\cmd]<
        \quad
      \poolB =
        \poolA\left[\begin{array}{@{}l@{}}
          \inti :=
          \poolLkp[:=]
            >*[\phi[\reg:=\intn]]     %
            >*[\intj+1]               %
            >*<                       %
        \end{array}\right]
        \arrcr[.6cm]
         \loc=\interp\expr_\phi  \quad
         \evA=\pair\inti\intj   \quad
        \graphB =
          \graphUpt*{\{\evA\}}                   %
                   *{[\evA:=\Read\md\loc\intn]}< %
                   {}                            %
    \end{array}}{
      \poolstep\poolA\graphA\poolB\graphB
    }
\\
    \inferrule[\nameWriteStep]{\begin{array}{@{}c@{}}
      \poolLkp>*>*>*[\Seq{\WritePL\md\exprA\exprB}\cmd]<
        \quad
      \poolB =
        \poolA\left[\begin{array}{@{}l@{}}
          \inti :=
          \poolLkp[:=]
            >                         %
            >*[\intj+1]               %
            >*<                       %
        \end{array}\right]
        \arrcr[.4cm]
         \loc=\interp\exprA_\phi   \quad
         \intn=\interp\exprB_\phi  \quad
         \evA=\pair\inti\intj     \quad
        \graphB =
          \graphUpt*{\{\evA\}}                   %
                   *{[\evA:=\Write\md\loc\intn]}< %
                   {}                            %
    \end{array}}{
      \poolstep\poolA\graphA\poolB\graphB
    }
\\
    \inferrule[RMWSuccessStep]{\begin{array}{@{}c@{}}
      \poolLkp>*>*>*[\Seq{\RMWPL\md\reg\exprOne\exprTwo\exprThree}\cmd]<
        \quad
      \poolB =
        \poolA\left[\begin{array}{@{}l@{}}
          \inti :=
          \poolLkp[:=]
            >*[\phi[\reg:=\intn]]    %
            >*[\intj+1]              %
            >*[\cmd]<                %
        \end{array}\right]
      \arrcr[.6cm]
      \loc=\interp\exprOne_\phi
        \quad
      \intn=\interp\exprTwo_\phi
        \quad
      \intm = \interp\exprThree_\phi
        \quad
      \evA = \pair\inti\intj
        \quad
      \graphB =
        \graphUpt*{\evA}                                  %
                 *{[\evA:=\RMW[\md]\loc\intn\intm]}<      %
                 {}<                                      %
    \end{array}}{
      \poolstep\poolA\graphA\poolB\graphB
    }
\\
    \inferrule[RMWFailStep]{\begin{array}{@{}c@{}}
      \poolLkp>*>*>*[\Seq{\RMWPL\md\reg\exprOne\exprTwo\exprThree}\cmd]<
        \quad
      \poolB =
        \poolA\left[\begin{array}{@{}l@{}}
          \inti :=
          \poolLkp[:=]
            >*[\phi[\reg:=\intm]]    %
            >*[\intj+1]              %
            >*[\cmd]<                %
        \end{array}\right]
      \arrcr[.5cm]
        \loc=\interp\exprOne_\phi
          \quad
        \intn=\interp\exprTwo_\phi
          \quad
        \evA = \pair\inti\intj
          \quad
        \intm \neq \intn
      \quad
      \graphB =
        \graphUpt*{\evA}                                  %
                 *{[\evA:=\RMW[\md]\loc\intm\bot]}<       %
                 {}<                                      %
    \end{array}}{
      \poolstep\poolA\graphA\poolB\graphB
    }
\\
    \inferrule[FenceStep]{\begin{array}{@{}c@{}}
      \poolLkp>>*>*[\Seq{\FencePL\md}\cmd]<
        \quad
      \poolB =
        \poolA\left[\begin{array}{@{}l@{}}
          \inti :=
          \poolLkp[:=]
            >            %
            >*[\intj+1]  %
            >*[\cmd]<    %
        \end{array}\right]
        \arrcr[.3cm]
      \evA = \pair\inti\intj
        \quad
      \graphB =
        \graphUpt*{\{\evA\}}               %
                 *{[\evA:=\Fence\md]}<     %
                 {}                        %
    \end{array}}{
      \poolstep\poolA\graphA\poolB\graphB
    }
\end{mathpar}
\caption{Pool-reduction rules (Part 1 of 2).}
\Description{}
\label{fig:pool:reduction:partI:app}
\end{figure}

\begin{figure}[H]
\begin{mathpar}
    \inferrule[\nameIfStep]{
      \poolLkp>*>>*[\Seq{\IfThen\expr\cmdA}\cmdB]<
        \and
      \intn=\interp\expr_\phi
        \\\\
      \begin{array}{@{}l@{}}
        \poolB =
          \poolA\left[\begin{array}{@{}l@{}}
            \inti :=
            \poolLkp[:=]
              >                                                    %
              >                                                    %
              >*[\metaIfThenElse{\intn\neq0}{\Seq\cmdA\cmdB}\cmdB]< %
          \end{array}\right]
      \end{array}
    }{
      \poolstep\poolA\graphA\poolB\graphA
    }
\and
    \inferrule[SkipStep]{
      \poolA[\inti].\nextcmd = \Seq\Skip{\cmd}
        \\\\
      \poolB =
        \poolA[\inti].\nextcmd := \cmd
    }{
      \poolstep\poolA\graphA\poolB\graphA
    }
\and
    \inferrule[WhileStep]{
      \poolA[\inti].\nextcmd = \Seq{\While\expr{\cmdA}}{\cmdB}
        \\\\
      \poolB =
      \poolA[\inti].\nextcmd :=
        \Seq{\IfThen\expr{\Seq\cmdA{\While\expr\cmdA}}}\cmdB
    }{
      \poolstep\poolA\graphA\poolB\graphA
    }
 \and
    \inferrule[SeqStep]{
      \poolA[\inti].\nextcmd = \Seq{(\Seq{\cmd_1}{\cmd_2})}{\cmd_3}
        \\\\
      \poolB =
        \poolA[\inti].\nextcmd := \Seq{\cmd_1}{(\Seq{\cmd_2}{\cmd_3})}
    }{
      \poolstep\poolA\graphA\poolB\graphA
    }
\\
    \inferrule[\nameTerminateStep]{
      \poolA[\inti].\nextcmd = \Skip
        \and
      \poolB = 
        \lambda\intj\in\dom\poolA\setminus\{\inti\}.\;\poolA[\intj]
    }{
      \poolstep\poolA\graphA\poolB\graphA
    }
\end{mathpar}

\caption{Pool-reduction rules (Part 2 of 2).}
\Description{}
\label{fig:pool:reduction:partII:app}
\end{figure}

\let\RefTirName\oldRefTirName

%% file: figure-cpp-syntax-app.tex
\begin{figure}[H]

\raggedright

\begin{minipage}{.75\textwidth}
\[\begin{array}{@{}r@{\;}r@{\;}l@{}}
  \typeExpr\ni\expr & ::=
               & \intn\;(\in\NN)
          \mid   \reg\;(\in\typeReg)
          \mid   \loc\;(\in\typeLoc\triangleq\NN)
          \mid   \Plus\expr\expr
          \mid   \Sub\expr\expr
          \mid   \Times\expr\expr

\\[.5mm]

  \typeCmd\ni\cmd & ::=
               & \ReadPL\md\reg\expr
         \mid    \WritePL\md\expr\expr
         \mid    \RMWPL\md\reg\expr\expr\expr
\\
        &\mid &  \FencePL\md
         \mid    \IfThen\expr\cmd
         \mid    \While\expr\cmd
         \mid    \Seq\cmd\cmd
         \mid    \Skip

\\[.5mm]

  \typeMode\ni\md & ::=
            & \na
         \mid \rlx
         \mid \rel
         \mid \acq
         \mid \acqrel
         \mid \sc
\end{array}\]
\end{minipage}
\begin{minipage}{.15\textwidth}
\centering
\hspace{-1.1cm}
\includegraphics{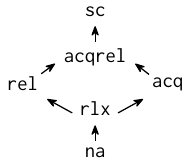}
\end{minipage}
\vspace{-10pt}
\caption{Syntax of \rcelevenlang.}
\vspace{-10pt}
\Description{}
\label{fig:cpp:syntax:app}
\end{figure}

%% file: rc11-consistency-app.tex
An execution graph~$\tuple\graphA{\rf,\mo}$
is \emph{\rceleven-consistent} if the conditions
\begin{itemize}
\item \axiomCoherenceDef\hfill\textnormal{(\TirName\nameCoherence)}
\item \axiomSCDef\hfill\textnormal{(\TirName\nameSC)}
\item \axiomAtomicityDef\hfill\textnormal{(\TirName\nameAtomicity)}
\item \axiomNTADef\hfill\textnormal{(\TirName\nameNTA)}
\end{itemize}
hold, where the relations
$\hb$,
$\sw$,
$\eco$,
$\rb$,
$\psc$,
and $\scb$
are defined as follows:

\[\begin{array}{@{}l@{\hspace{1cm}}r@{}}
\begin{array}{@{}r@{\;}c@{\;\;}l@{}}
  \rb &\eqdef&
    (\inv\rf\seq\mo)\setminus\settorel\EventE
\\[.4mm]
  \hb &\eqdef&
    \transC{(\union\po\sw)}
\\[.4mm]
  \eco &\eqdef&
    \transC{(\union\rf{\union\mo\rb})}
\end{array}
&
\begin{array}{@{}l@{}}
  \sw\eqdef
    \left\{\begin{array}{@{}l@{}}
    \settorel{\EventAtLeast\rel}\seq
     \refC{(
       \settorel\FenceE\seq
       \po
     )}\seq
    \settorel{\WriteAtLeast\rlx}\seq\\
    \quad\transC\rf\seq\\
    \settorel{\ReadAtLeast\rlx}\seq
    \refC{(\po\seq{\settorel\FenceE})}\seq
    \settorel{\EventAtLeast\acq}
    \end{array}\right.
\end{array}
\end{array}\]

\begin{minipage}{.2\textwidth}
\[\begin{array}{r@{\;}c@{\;\;}l}
  \psc &\eqdef& \union\pscBase\pscFence
\end{array}\]
\end{minipage}
\begin{minipage}{.8\textwidth}
\[\begin{array}{r@{\;}c@{\;\;}l}
  \scb &\eqdef&
    \union\po{
    \union{{\neqloc\po}\seq{\hb\seq{\neqloc\po}}}{
    \union{\perloc\hb}{
    \union\mo\rb
    }}}
\\[.4mm]
  \pscBase &\eqdef&
    {(
      \union{\settorel{\EventMd\sc}}
            {{\settorel{\FenceMd\sc}}\seq{\refC\hb}}
    )}\seq{
      \scb\seq{(
        \union{\settorel{\EventMd\sc}}
              {{\refC\hb}\seq{\settorel{\FenceMd\sc}}}
      )}
    }
\\[.4mm]
  \pscFence &\eqdef&
    {\settorel{\FenceMd\sc}}\seq{
      {
        (\union\hb{{\hb\seq\eco}\seq\hb})
      }\seq{
        \settorel{\FenceMd\sc}
      }}
\end{array}\]
\end{minipage}

%% file: figure-cpp-intelext-diagram-app.tex
\begin{figure}[H]
\[\begin{array}{c}
\includegraphics{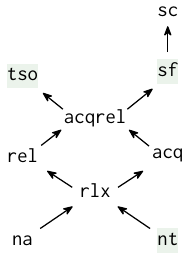}
\end{array}\]
\caption{Diagram of access modes.}
\Description{}
\label{fig:cpp:x86:diagram:app}
\end{figure}

%% file: rc11-ex86-consistent-app.tex
An execution graph~$\tuple\graphA{\rf,\mo}$ is
\emph{\rcelevenext-consistent}
if the conditions
\begin{itemize}
\item $\irr{\mathhl\hb\seq\refC{\mathhl\eco}}$
      \hfill\textnormal{(\TirName\nameCoherenceI)}
\item \greenhl{\axiomCoherenceIIDef}\hfill
      \greenhl{\textnormal{(\TirName\nameCoherenceII)}}
\item \greenhl{\axiomCoherenceIIIDef}\hfill
      \greenhl{\textnormal{(\TirName\nameCoherenceIII)}}
\item \axiomAtomicityDef\hfill\textnormal{(\TirName\nameAtomicity)}
\item \axiomSCDef\hfill\textnormal{(\TirName\nameSC)}
\item \axiomNTADef\hfill\textnormal{(\TirName\nameNTA)}
\end{itemize}
hold, where the relations
$\hb$,
$\sw$,
$\eco$,
$\rb$,
$\psc$,
$\scb$,
$\porceleven$,
and $\ppoasm$
are defined as follows:

\[\begin{array}{@{}l@{\;\;\;}r@{}}
\begin{array}{@{}r@{\,}c@{\;}l@{}}
  \rb &\triangleq&
    (\inv\rf\seq\mo)\setminus\settorel\EventE
\\[.4mm]
  \hb &\triangleq&
    \transC{(\mathhl\porceleven\disj\sw)}
\\[.4mm]
  \eco &\triangleq&
    \transC{(
      \mathhl\rfe\disj\mo\disj\rb
    )}
\end{array}
&
\begin{array}{@{}r@{\;}l@{}}
  \sw &\triangleq
    \left\{\begin{array}{@{}l@{}}
    \settorel{\EventAtLeast\rel}\seq
     \refC{(
       \settorel\FenceE\seq
       \po
     )}\seq
    \settorel{\WriteAtLeast\rlx}\seq\\
    \quad\transC\rf\seq\\
    \settorel{\ReadAtLeast\rlx}\seq
    \refC{(\po\seq{\settorel\FenceE})}\seq
    \settorel{\EventAtLeast\acq}
    \end{array}\right.
\end{array}
\end{array}
\]
\[\begin{array}{@{}l@{}r@{}}
\begin{array}{r@{}l}
  \mathhl{\porceleven}
         &\mathhl{\;\triangleq
    \settorel{\EventE\setminus\WriteMd\nt}\seq
    \po}
\\
         &\mathhl{\;\cup\;
    \po\seq
    \settorel{\RMWMd\tso\cup\FenceAtLeast\stf}}
\\
         &\mathhl{\;\cup\;
    \perloc\po\seq\settorel\WriteE}
\end{array}
&
\begin{array}{r@{}l}
  \mathhl{\ppoasm}
         &\mathhl{\;\triangleq
    \po\seq\settorel{\RMWMd\tso\cup\FenceAtLeast\stf}}
\\
         &\mathhl{\;\cup\;
    \settorel{\ReadMd\tso\cup\RMWMd\tso\cup\FenceMd\sc}\seq\po}
\\

         &\mathhl{\;\cup\;
    \settorel{\FenceAtLeast\stf}\seq\po\seq\settorel{\EventE\setminus\ReadE}}
\\
         &\mathhl{\;\cup\;
    \settorel{\WriteMd\tso}\seq\po\seq\settorel{\EventE\setminus\ReadE\setminus\WriteMd\nt}}
\\
         &\mathhl{\;\cup\;
    \settorel{\EventE\setminus\ReadE\setminus\WriteMd\nt}\seq\po\seq
    \settorel{\WriteMd\tso}}
\end{array}
\end{array}\]

\[\begin{array}{r@{\;}c@{\;\;}l}
  \psc &\eqdef&
    \pscBase\disj\pscFence
\\[.4mm]
  \scb &\eqdef&
    \po\disj
    {{\neqloc\po}\seq{\hb\seq{\neqloc\po}}}\disj
    \perloc\hb\disj
    \mo\disj
    \rb
\\[.4mm]
  \pscBase &\eqdef&
    {(
      {\settorel{\EventMd\sc}}\disj
      {{\settorel{\FenceMd\sc}}\seq{\refC\hb}}
    )}\seq{
      \scb\seq{(
        {\settorel{\EventMd\sc}}\disj
        {{\refC\hb}\seq{\settorel{\FenceMd\sc}}}
      )}
    }
\\[.4mm]
  \pscFence &\eqdef&
    {\settorel{\FenceMd\sc}}\seq{
      {(
        \hb\disj{{\hb\seq\eco}\seq\hb}
      )}\seq{
        \settorel{\FenceMd\sc}
      }}
\end{array}\]

%% file: counterexample-rc11-non-racy-non-SC.tex
\[\scalebox{.90}{
\includegraphics{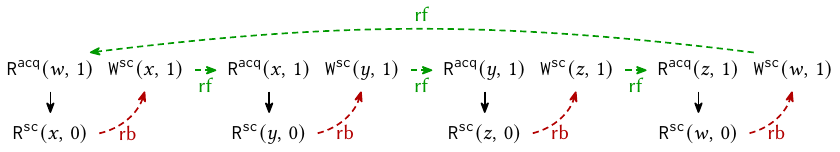}}\]

%% file: figure-mixed-nodes-intelext-app.tex
\begin{figure}[t]

\raggedright

\begin{mathpar}
    \inferrule[\nameWWMF]{}{\tikz{\mnodeWWMF}}
\and
    \inferrule[\nameWW]{}{\tikz{\mnodeWW[][][\rc/\asm]}}
\and
    \inferrule[\nameWNTW]{}{\tikz{\mnodeWNTW}}
\\
    \inferrule[\nameRMWRMWS]{}{\tikz{\mnodeRMWRMWS[][][\rc/\asm]}}
\and
    \inferrule[\nameRMWRMWF]{}{\tikz{\mnodeRMWRMWF[][][\rc/\asm]}}
\\
    \inferrule[\nameFMF]{}{\tikz{\mnodeFMF}}
\and
    \inferrule[\nameFSF]{}{\tikz{\mnodeFSF}}
\and
    \inferrule[\nameFB]{}{\tikz{\mnodeFB}}
\and
    \inferrule[\nameRR]{}{\tikz{\mnodeRR[][][\rc/\asm]}}
\end{mathpar}

\caption{
Set of mixed nodes reflecting the
compilation scheme from
\defref{def:scheme:rcelevenext:app}.
}
\Description{}
\label{fig:mnodes:intelext:app}
\end{figure}

%% file: figure-projections-intelext-app.tex
\begin{figure}[t]
\begin{center}
\includegraphics{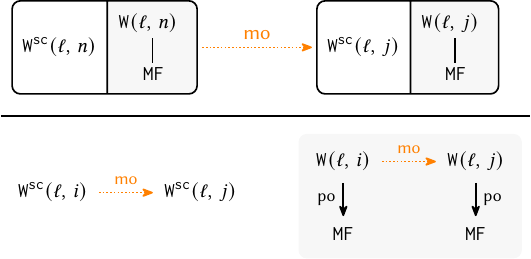}\end{center}

\vspace{10pt}

\begin{mathpar}
\includegraphics{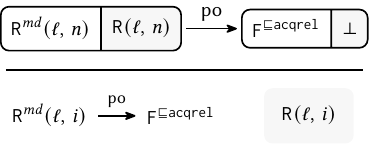}
\and
\includegraphics{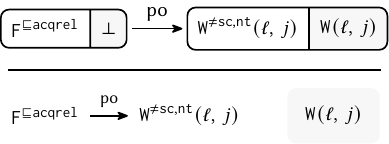}\end{mathpar}

\caption{Selection of projection rules.}
\Description{}
\label{fig:proj:intelext:app}
\end{figure}

%% file: proof-compilation-intelext.tex
Suppose that $\graphM$ is \intelext-consistent
but not \rcelevenext-consistent.
Then at least one of \rcelevenext consistency conditions must
not hold of $\graphM$.
We show that the violation of any of them
leads to a contradiction:
\begin{enumerate}
\item \axiomCoherenceI.\\
The violation of \axiomCoherenceI implies that (at least)
one of the following assertions holds:
\begin{enumerate}
\item Assertion: $\notirr\hb$
\[\begin{array}{@{}r@{\;}c@{\;}l@{}}
\notirr\hb
  &\implies&
    \cyc{\po\disj\rfe}
\\
  &\implies&
    \cyc{(\underbrace{\settorel{\codom\rfe}\seq\po}_{\subseteq\;\ppo})\disj\rfe}
\\
  &\implies&
    \cyc{\ppo\disj\rfe\disj\moe\disj\rbe}
\end{array}\]

\item Assertion: $\notirr{\hb\seq\rfe}$
\[\begin{array}{@{}r@{\;}c@{\;}l@{}}
\notirr{\hb\seq\rfe}
  &\implies&
    \cyc{\po\disj\rfe}
\\
  &\implies&
    \cyc{\ppo\disj\rfe\disj\moe\disj\rbe}
\end{array}\]

\item Assertion: $\notirr{\hb\seq(\mo\disj\rb)}$
\[\begin{array}{@{~~~~~~~}l@{}}
\begin{array}{@{}l@{\;}c@{\;}l@{}}
\notirr{\hb\seq(\moi\disj\rbi)}
  &\implies&
    \notirr{\transC{(\po\disj\rfe)}\seq(\moi\disj\rbi)}
\\
  &\implies&
    \notirr{\transC{(\po\disj\rfe)}\seq\po}
\\
  &\implies&
    \cyc{\po\disj\rfe}
\\
  &\implies&
    \cyc{\ppo\disj\rfe\disj\moe\disj\rbe}
\end{array}
\\
\hrulefill
\\[1mm]
\begin{array}{@{}l@{}l@{}}
  &\notirr{\hb\seq(\moe\disj\rbe)}
\\
  &\implies
    \notirr{\transC{(\porceleven\disj\sw)}\seq(\moe\disj\rbe)}
    \implies
\\
  &
    \notirr*{%
      \begin{array}{l}
        \overbrace{
          \settorel{\codom{\mo\cup\rb}}\seq
          \porceleven\seq
          \settorel{\EventAtLeast\rel}\seq
           \refC{(
             \settorel\FenceE\seq
             \po
           )}\seq
          \settorel{\WriteAtLeast\rlx}
        }^{\begin{array}{l}\subseteq\;
           \left(\begin{array}{l}
            \settorel{\WriteMd{\neq\nt}}\seq
            \po\seq
            \settorel{\WriteMd{\neq\nt}}
            \;\cup
          \\
            \settorel{\WriteMd\nt}\seq
            \perloc\po\seq
            \settorel{\WriteE}
            \;\cup
          \\
            \settorel{\WriteMd\nt}\seq
            \po\seq
            \settorel{\RMWMd\tso\cup\FenceAtLeast\stf}\seq
            \po
          \end{array}\right)
          \;\subseteq\;\ppo
          \end{array}}
        \seq
      \\
        \transC{(\rfe\seq\refC\ppo)}\seq
        (\moe\disj\rbe)
      \end{array}
    }
\\
  &\implies
    \notirr{%
        \ppo\seq\transC{(\rfe\seq\refC\ppo)}\seq
        (\moe\disj\rbe)
    }
\\
  &\implies
    \cyc{\ppo\disj\rfe\disj\moe\disj\rbe}
\end{array}
\end{array}\]

\item Assertion: $\notirr{\hb\seq(\mo\disj\rb)\seq\rfe}$
\[\begin{array}{@{~~~~}l@{}}
\begin{array}{@{}l@{\;}c@{\;}l@{}}
\notirr{\hb\seq(\moi\disj\rbi)\seq\rfe}
  &\implies&
    \notirr{\hb\seq\porceleven\seq\rfe}
\\
  &\implies&
    \notirr{\hb\seq\rfe}
\end{array}
\\
\hrulefill
\\
\begin{array}{@{}l@{}l@{}}
  &\notirr{\hb\seq(\moe\disj\rbe)\seq\rfe}
\\
  &\implies
    \notirr{\reftransC{(\po\cup\rfe)}\seq(\moe\disj\rbe)\seq\rfe}
\\
  &\implies
    \notirr{
      \settorel{\codom\rf}\seq
      \po\seq
      \reftransC{(\rfe\seq\refC\po)}\seq
      (\moe\disj\rbe)\seq\rfe
    }
\\
  &\implies
    \notirr{
      \ppo\seq
      \reftransC{(\rfe\seq\refC\ppo)}\seq
      (\moe\disj\rbe)\seq\rfe
    }
\\
  &\implies
    \cyc{\ppo\disj\rfe\disj\moe\disj\rbe}
\end{array}
\end{array}\]

\end{enumerate}

\item \axiomCoherenceII.
\[\begin{array}{@{}r@{\;\;}c@{\;\;}l@{}}
\cyc{\ppoasm\disj\eco}
  &\implies&
    \cyc{\ppoasm\disj\rfe\disj\mo\disj\rb}
\\
  &\implies&
    \cyc{(\ppoasm\disj\moi\disj\rbi)\disj\rfe\disj\moe\disj\rbe}
\\
  &\implies&
    \cyc{\ppo\disj\rfe\disj\moe\disj\rbe}
\end{array}\]

\item \axiomCoherenceIII. (Immediate by \axiomInternal.)

\item \axiomAtomicity.
\[\begin{array}{@{}r@{\;\;}c@{\;\;}l@{}}
\notirr{\rb\seq\mo}
  &\implies&
    \notirr{\rbi\seq\moi}
      \;\lor\;
    \notirr{\rbe\seq\moe}
\\
  &\implies&
    \notirr{\po}\;\lor\;
    \cyc{\ppo\disj\rfe\disj\moe\disj\rbe}
\end{array}\]

\item \axiomSC. \\
We prove that $\psc\subseteq\transC\ob$,
therefore~$\cyc\psc\implies\cyc\ob$.
\[\begin{array}{@{}l@{}}
\begin{array}{@{}r@{\;\;}c@{\;\;}l@{}}
\psc
  &\eqdef&
    \pscFence \disj
    \pscBase
\\[2mm]
\pscFence
  &\eqdef&
    \settorel{\FenceMd\sc}\seq(\hb\disj\hb\seq\eco\seq\hb)\seq\settorel{\FenceMd\sc}
\\
  &\;=&
    \settorel{\FenceMd\sc}\seq\hb\seq\settorel{\FenceMd\sc}
      \disj
    \settorel{\FenceMd\sc}\seq\hb\seq\eco\seq\hb\seq\settorel{\FenceMd\sc}
\end{array}
\\
\hrulefill
\\[2mm]
\begin{array}{@{}r@{\;\;}c@{\;\;}l@{}}
\settorel{\FenceMd\sc}\seq\hb\seq\settorel{\FenceMd\sc}
  &\subseteq&
    \underbrace{
      \settorel{\FenceMd\sc}\seq\po\seq
    }_{\subseteq\;\ppo}
    \reftransC{(\rfe\disj\ppo)}
   \subseteq
    \transC{(\ppo\disj\rfe)}\subseteq
    \transC\ob
\end{array}
\\
\hrulefill
\\[2mm]
\begin{array}{@{}l@{}}
  \settorel{\FenceMd\sc}\seq\hb\seq\eco\seq\hb\seq\settorel{\FenceMd\sc}
\\
\quad
  \begin{array}{@{}l@{}}
    \subseteq\;
      \begin{array}{@{}l@{}}
        \overbrace{
          \settorel{\FenceMd\sc}\seq\po
        }^{\subseteq\;\ppo}
       \seq
       \reftransC{(\rfe\seq\refC\ppo)}\seq
       \overbrace\eco^{\subseteq\;\transC\ob}\seq
      \\
        \underbrace{
          \settorel{\codom\eco}\seq
          \refC\porceleven\seq
          \refC{(
            \settorel{\EventAtLeast\rel}\seq
            \refC{(
              \settorel\FenceE\seq
              \po
            )}\seq
            \settorel{\WriteAtLeast\rlx}\seq
            \reftransC{(\rfe\seq\refC\ppo)}
          )}\seq
          \settorel{\FenceMd\sc}
        }_{
        \begin{array}{l}\subseteq\;
          \left(\begin{array}{l}
            \settorel{\ReadE\cup\RMWE}\seq
            \po\seq
            \reftransC{(\rfe\seq\refC\ppo)}
          \;\cup\\
            \settorel{\WriteMd{\neq\nt}}\seq
            \po\seq
            \settorel{\WriteMd{\neq\nt}\cup\FenceMd\sc}\seq
            \reftransC{(\rfe\seq\refC\ppo)}
          \;\cup\\
            \settorel{\WriteMd\nt}\seq
            \perloc\po\seq
            \settorel{\EventE\setminus\ReadE}\seq
            \reftransC{(\rfe\seq\refC\ppo)}
          \;\cup\\
            \settorel{\WriteMd\nt}\seq
            \po\seq
            \settorel{\RMWMd\tso\cup\FenceAtLeast\stf}\seq
            \po\seq
            \settorel{\EventE\setminus\ReadE}\seq
            \reftransC{(\rfe\seq\refC\ppo)}
          \end{array}\right)^{?}
        \end{array}
        }
      \end{array}
  \\
    \subseteq\;
      \ppo\seq
      \reftransC{(\rfe\seq\refC\ppo)}\seq
      \transC\ob\seq
      \refC\ppo\seq
      \reftransC{(\rfe\seq\refC\ppo)}
    \;\subseteq\;\transC\ob
  \end{array}

\end{array}
\\
\hrulefill
\\[2mm]
\begin{array}{@{}l@{}}
\pscBase \eqdef
  \begin{array}{@{\;}l@{}}
  (\settorel{\EventMd\sc}\disj\settorel{\FenceMd\sc}\seq\refC\hb)\seq\\
  \quad(\po\disj\neqloc\po\seq\hb\seq\neqloc\po\disj\perloc\hb\disj\mo\disj\rb)\seq\\
  (\refC\hb\seq\settorel{\FenceMd\sc}\disj\settorel{\EventMd\sc})
  \end{array}
\\
\begin{array}{@{}r@{\;\;}ll@{}}
  \subseteq&
    \settorel{\EventMd\sc}\seq
    \transC{(\po\disj\rfe)}
&
    (\subseteq
      \settorel{\EventMd\sc}\seq\po\seq\reftransC{(\rfe\seq\refC\ppo)}
    \subseteq
      \transC\ob
    )
\\
  \cup&
    \settorel{\FenceMd\sc}\seq
    \hb\seq
    \eco\seq
    \hb\seq
    \settorel{\FenceMd\sc}
  &(\subseteq\transC\ob)
\\
  \cup&
    \settorel{\FenceMd\sc}\seq
    \hb\seq
    (\mo\disj\rb)
&(\subseteq
  \settorel{\FenceMd\sc}\seq
  \po\seq
  \reftransC{(\rfe\seq\refC\ppo)}\seq
  \eco
  \subseteq
  \transC\ob
)
\\
  \cup&
    (\mo\disj\rb)\seq
    \hb\seq
    \settorel{\FenceMd\sc}
&(\text{See proof of}\;
\settorel{\FenceMd\sc}\seq\hb\seq\eco\seq\hb\seq\settorel{\FenceMd\sc}
\subseteq
\transC\ob
)
\\
  \cup&
    (\mo\disj\rb)
&(\subseteq
(\moi
\disj
\rbi)
\disj
\moe\disj\rbe
\subseteq
\ob
)
\end{array}
\end{array}
\end{array}\]

\item \axiomNTA.
\[\begin{array}{@{}r@{\;\;}c@{\;\;}l@{}}
\cyc{\po\disj\rfe}
  &\implies&
    \cyc{\underbrace{(\settorel{\codom\rf}\seq\po)}_{\subseteq\;\ppo}\disj\rfe}
\\
  &\implies&
    \cyc{\ppo\disj\rfe\disj\moe\disj\rbe}
\end{array}\]
\end{enumerate}

%% file: figure-mixed-nodes-intelext-alt-app.tex
\begin{figure}[t]

\raggedright

\begin{center}
\begin{minipage}{.35\textwidth}
\begin{mathpar}
    \inferrule[\nameWSFWMF]{}{\tikz{\mnodeWSFWMF}}
\end{mathpar}
\end{minipage}
\begin{minipage}{.55\textwidth}
\begin{mathpar}
    \inferrule[\nameWSFW]{}{\tikz{\mnodeWSFW}}
\and
    \inferrule[\nameWNTWalt]{}{\tikz{\mnodeWNTW[][][\rlx,\nt]}}
\and
    \inferrule[\nameFSFalt]{}{\tikz{\mnodeFSF[][][\stf,\rel,\acqrel]}}
\end{mathpar}
\end{minipage}
\end{center}
\caption{
Selected set of mixed nodes reflecting the
alternative compilation scheme from
\defref{def:alt:scheme:rcelevenext:app}.
}
\Description{}
\label{fig:mnodes:intelext:alt:app}
\end{figure}

%% file: proof-compilation-intelext-alt.tex
The proof is analogous to the proof of~%
\lemmaref{lemma:rcelevenext:weaker:intelext}.
The main difference is how we show
that the~$\po$ prefix of a~$\hb$ edge
is included in~$\ppo$.
To give an illustration,
we include here the proof that,
in the alternative mixed graph~$\graphM$,
the violation of the condition~%
$\irr{\hb\seq(\moe\disj\rbe)}$
(ensured by~\axiomCoherenceI)
implies the violation of \axiomExternal.
The idea is to exploit the fact that,
thanks to the alternative compilation scheme
(which is encoded in the structure of~$\graphM$),
a $\sw$ edge always starts with a store fence
or a stronger barrier:
\[\begin{array}{@{}l@{}l@{}}
  &\notirr{\hb\seq(\moe\disj\rbe)}
\\
  &\implies
    \notirr{\transC{(\porceleven\disj\sw)}\seq(\moe\disj\rbe)}
    \implies
\\
  &
    \notirr*{%
      \begin{array}{l}
        \overbrace{
          \settorel{\codom{\mo\cup\rb}}\seq
          \refC\porceleven\seq
          \settorel{\EventAtLeast\rel}\seq
          \refC{(\settorel{\FenceE}\seq\po)}\seq
          \settorel{\WriteAtLeast\rlx}
        }^{\begin{array}{l}\subseteq\;
             \left(\begin{array}{@{}l@{}}
               \settorel{\WriteMd\nt}\seq
               \perloc\po\seq
               \settorel{\WriteE}
               \;\cup
             \\
               \settorel{\WriteE}\seq
               \po\seq
               \left[\begin{array}{@{}l@{}}
                 \MFenceE\cup
                 \SFenceE\cup
                 \RMWE\,\cup\\
                 \nodeWSFW\,\cup\\
                 \nodeWSFWMF
               \end{array}\right]\seq
               \refC\po\seq\settorel{\WriteE}
             \end{array}\right)^{?}
             \;\subseteq\;\refC\ppo
          \end{array}}
        \seq
      \\
        \transC{(\rfe\seq\refC\ppo)}\seq
        (\moe\disj\rbe)
      \end{array}
    }
\\
  &\implies
    \notirr{%
        \refC\ppo\seq
        \transC{(\rfe\seq\refC\ppo)}\seq
        (\moe\disj\rbe)
    }
\\
  &\implies
    \cyc{\ppo\disj\rfe\disj\moe\disj\rbe}
\end{array}\]

%% file: figure-deordering-app.tex
\begin{figure}[t]
\renewcommand{\arraystretch}{1.5}
\begin{tabular}{@{}|c||c|c|c|c|@{}}
\cline{2-5}
\multicolumn{1}{c|}{}
&$\ReadE^{\mdB}_{\locy}$
&$\WriteE^{\mdB}_{\locy}$
&$\RMWE^{\mdB}_{\locy}$
&$\FenceE^{\mdB}$
\\
\hhline{-|=|=|=|=|}
$\ReadE^{\mdA}_{\locx}$
&$\mdA\sqsubseteq\rlx$
&$\begin{array}{c}
\mdA,\mdB\sqsubseteq\rlx\\[-5pt]
\na\in\{\mdA,\mdB\}
\end{array}$
&$\begin{array}{c}
\mdA=\na\\[-5pt]
\mdB\sqsubseteq\acq
\end{array}$
&$\begin{array}{c}
\mdA\neq\rlx\\[-5pt]
\mdB=\acq
\end{array}$
\\
\hline
$\WriteE^{\mdA}_{\locx}$
&$\begin{array}{c}
\mdA\neq\sc\\[-5pt]
\mdB\neq\sc
\end{array}$
&$\mdB\sqsubseteq\rlx$
&$\mdB\sqsubseteq\acq$
&$\mdB=\acq$
\\
\hline
$\RMWE^{\mdA}_{\locx}$
&$\mdA\sqsubseteq\rel$
&$\begin{array}{c}
\mdA\sqsubseteq\rel\\[-5pt]
\mdB=\na
\end{array}$
&---
&$\begin{array}{c}
\mdA\sqsupseteq\acq\\[-5pt]
\mdB=\acq
\end{array}$
\\
\hline
$\FenceE^{\mdA}$
&$\mdA=\rel$
&$\begin{array}{c}
\mdA=\rel\\[-5pt]
\mdB\neq\rlx
\end{array}$
&$\begin{array}{c}
\mdA=\rel\\[-5pt]
\mdB\sqsupseteq\rel
\end{array}$
&$\begin{array}{c}
\mdA=\rel\\[-5pt]
\mdB=\acq
\end{array}$
\\
\hline
\end{tabular}
\caption{Deorderable pairs.
Events belong to \rceleven.
The variables $\locx$ and $\locy$ denote distinct locations.}
\Description{}
\label{fig:deordering:app}
\end{figure}

%% file: figure-merging-app.tex
\begin{figure}[t]
\[\begin{array}{@{}c@{}}
\begin{array}{@{}c@{\quad\quad\quad\quad}c@{\quad\quad\quad\quad}c@{}}
\begin{array}{@{}r@{\;}c@{\;}l@{}}
  \Seq{\WriteMd\md}{\WriteMd\md}
    &\rightsquigarrow&
  \WriteMd{\md}
\\
  \Seq{\ReadMd\md}{\ReadMd\md}
    &\rightsquigarrow&
  \ReadMd\md
\end{array}
&
\begin{array}{@{}r@{\;}c@{\;}l@{}}
  \Seq{\WriteMd\sc}{\ReadMd\sc}
    &\rightsquigarrow&
  \WriteMd{\sc}
\\
  \Seq{\WriteMd\md}{\ReadMd\acq}
    &\rightsquigarrow&
  \WriteMd{\md}
\end{array}
&
\begin{array}{@{}r@{\;}c@{\;}l@{}}
  \Seq{\RMWMd\md}{\RMWMd\md}
    &\rightsquigarrow&
  \RMWMd\md
\\
  \Seq{\FenceMd\md}{\FenceMd\md}
    &\rightsquigarrow&
  \FenceMd\md
\end{array}
\end{array}
\\[5mm]
\begin{array}{c}
\begin{array}{@{}r@{\;}c@{\;}ll@{}}
  \Seq{\WriteMd{\mdA}}{\RMWMd\mdB}
    &\rightsquigarrow&
  \WriteMd\mdA
    &
  \mdA = \maximum{\md\mid\rlx\sqsubseteq\md\sqsubseteq\mdB}
\\
  \Seq{\RMWMd\mdA}{\ReadMd\mdB}
    &\rightsquigarrow&
  \RMWMd\mdB
    &
  \mdB = \maximum{\md\mid\rlx\sqsubseteq\md\sqsubseteq\mdA}
\end{array}
\end{array}
\end{array}\]
\caption{Mergeable pairs. Events belong to \rceleven. Events in a pair access the same location.}
\Description{}
\label{fig:merging:app}
\end{figure}

%% file: main.bbl

\begin{thebibliography}{31}


\ifx \showCODEN    \undefined \def \showCODEN     #1{\unskip}     \fi
\ifx \showDOI      \undefined \def \showDOI       #1{#1}\fi
\ifx \showISBNx    \undefined \def \showISBNx     #1{\unskip}     \fi
\ifx \showISBNxiii \undefined \def \showISBNxiii  #1{\unskip}     \fi
\ifx \showISSN     \undefined \def \showISSN      #1{\unskip}     \fi
\ifx \showLCCN     \undefined \def \showLCCN      #1{\unskip}     \fi
\ifx \shownote     \undefined \def \shownote      #1{#1}          \fi
\ifx \showarticletitle \undefined \def \showarticletitle #1{#1}   \fi
\ifx \showURL      \undefined \def \showURL       {\relax}        \fi
\providecommand\bibfield[2]{#2}
\providecommand\bibinfo[2]{#2}
\providecommand\natexlab[1]{#1}
\providecommand\showeprint[2][]{arXiv:#2}

\bibitem[Alglave et~al\mbox{.}(2024)]%
        {alglave-24}
\bibfield{author}{\bibinfo{person}{Jade Alglave}, \bibinfo{person}{Richard
  Grisenthwaite}, \bibinfo{person}{Artem Khyzha}, \bibinfo{person}{Luc
  Maranget}, {and} \bibinfo{person}{Nikos Nikoleris}.}
  \bibinfo{year}{2024}\natexlab{}.
\newblock \showarticletitle{Puss In Boots: on formalizing Arm’s Virtual
  Memory System Architecture}.
\newblock \bibinfo{journal}{\emph{IEEE Micro}} (\bibinfo{date}{July}
  \bibinfo{year}{2024}), \bibinfo{pages}{1--9}.
\newblock
\urldef\tempurl%
\url{https://doi.org/10.1109/MM.2024.3422668}
\showURL{%
\tempurl}


\bibitem[Alglave et~al\mbox{.}(2014)]%
        {alglave-14}
\bibfield{author}{\bibinfo{person}{Jade Alglave}, \bibinfo{person}{Luc
  Maranget}, {and} \bibinfo{person}{Michael Tautschnig}.}
  \bibinfo{year}{2014}\natexlab{}.
\newblock \showarticletitle{Herding Cats - Modelling, simulation, testing, and
  data-mining for weak memory}.
\newblock \bibinfo{journal}{\emph{ACM Transactions on Programming Languages and
  Systems}} \bibinfo{volume}{36}, \bibinfo{number}{2} (\bibinfo{year}{2014}).
\newblock
\urldef\tempurl%
\url{https://doi.org/10.1145/2627752}
\showURL{%
\tempurl}


\bibitem[Batty et~al\mbox{.}(2012)]%
        {batty-al-12}
\bibfield{author}{\bibinfo{person}{Mark Batty}, \bibinfo{person}{Kayvan
  Memarian}, \bibinfo{person}{Scott Owens}, \bibinfo{person}{Susmit Sarkar},
  {and} \bibinfo{person}{Peter Sewell}.} \bibinfo{year}{2012}\natexlab{}.
\newblock \showarticletitle{Clarifying and Compiling C/C++ Concurrency: From
  C++11 to POWER} \emph{(\bibinfo{series}{Principles of Programming Languages
  ({POPL})})}. \bibinfo{pages}{509–--520}.
\newblock
\urldef\tempurl%
\url{https://doi.org/10.1145/2103656.2103717}
\showURL{%
\tempurl}


\bibitem[Batty et~al\mbox{.}(2011)]%
        {c11mm}
\bibfield{author}{\bibinfo{person}{Mark Batty}, \bibinfo{person}{Scott Owens},
  \bibinfo{person}{Susmit Sarkar}, \bibinfo{person}{Peter Sewell}, {and}
  \bibinfo{person}{Tjark Weber}.} \bibinfo{year}{2011}\natexlab{}.
\newblock \showarticletitle{Mathematizing {C++} concurrency}. In
  \bibinfo{booktitle}{\emph{Principles of Programming Languages ({POPL})}}.
  \bibinfo{publisher}{ACM Press}, \bibinfo{pages}{55--66}.
\newblock
\urldef\tempurl%
\url{https://www.cl.cam.ac.uk/~pes20/cpp/popl085ap-sewell.pdf}
\showURL{%
\tempurl}


\bibitem[{Clang Project}(2007)]%
        {clang}
\bibfield{author}{\bibinfo{person}{{Clang Project}}.}
  \bibinfo{year}{2007}\natexlab{}.
\newblock \bibinfo{title}{Clang: a C language family frontend for LLVM}.
\newblock
\newblock
\urldef\tempurl%
\url{https://clang.llvm.org/}
\showURL{%
\tempurl}


\bibitem[{Cppreference Community}(2019)]%
        {cppref-memorder}
\bibfield{author}{\bibinfo{person}{{Cppreference Community}}.}
  \bibinfo{year}{2019}\natexlab{}.
\newblock \bibinfo{title}{Cppreference - Memory Order}.
\newblock
\newblock
\urldef\tempurl%
\url{https://en.cppreference.com/w/cpp/atomic/memory_order}
\showURL{%
\tempurl}


\bibitem[Floyd(1967)]%
        {floyd-nondet-67}
\bibfield{author}{\bibinfo{person}{Robert~W. Floyd}.}
  \bibinfo{year}{1967}\natexlab{}.
\newblock \showarticletitle{Nondeterministic Algorithms}.
\newblock \bibinfo{journal}{\emph{Journal of the {ACM}}} \bibinfo{volume}{14},
  \bibinfo{number}{4} (\bibinfo{date}{Oct.} \bibinfo{year}{1967}),
  \bibinfo{pages}{636--644}.
\newblock
\urldef\tempurl%
\url{https://doi.org/10.1145/321420.321422}
\showURL{%
\tempurl}


\bibitem[Flynn(1972)]%
        {flynn-72}
\bibfield{author}{\bibinfo{person}{Michael~J. Flynn}.}
  \bibinfo{year}{1972}\natexlab{}.
\newblock \showarticletitle{Some Computer Organizations and Their
  Effectiveness}.
\newblock \bibinfo{journal}{\emph{{IEEE} Trans. Computers}}
  \bibinfo{volume}{C-21} (\bibinfo{date}{Nov.} \bibinfo{year}{1972}).
\newblock
\urldef\tempurl%
\url{https://ieeexplore.ieee.org/document/5009071}
\showURL{%
\tempurl}


\bibitem[{{GNU} Project}(1987)]%
        {gcc}
\bibfield{author}{\bibinfo{person}{{{GNU} Project}}.}
  \bibinfo{year}{1987}\natexlab{}.
\newblock \bibinfo{title}{GNU Compiler Collection}.
\newblock
\newblock
\urldef\tempurl%
\url{https://gcc.gnu.org/git/gcc.git}
\showURL{%
\tempurl}


\bibitem[Goens et~al\mbox{.}(2023)]%
        {goens-al-23}
\bibfield{author}{\bibinfo{person}{Andr\'{e}s Goens}, \bibinfo{person}{Soham
  Chakraborty}, \bibinfo{person}{Susmit Sarkar}, \bibinfo{person}{Sukarn
  Agarwal}, \bibinfo{person}{Nicolai Oswald}, {and} \bibinfo{person}{Vijay
  Nagarajan}.} \bibinfo{year}{2023}\natexlab{}.
\newblock \showarticletitle{Compound Memory Models}, Vol.~\bibinfo{volume}{7}.
  \bibinfo{publisher}{ACM Press}, \bibinfo{pages}{153:1--153:24}.
\newblock
\urldef\tempurl%
\url{https://doi.org/10.1145/3591267}
\showURL{%
\tempurl}


\bibitem[{Intel}(2024)]%
        {intel-manual}
\bibfield{author}{\bibinfo{person}{{Intel}}.} \bibinfo{year}{2024}\natexlab{}.
\newblock \bibinfo{title}{Intel 64 and IA-32 Architectures Software Developer's
  Manual (Combined Volumes)}.
\newblock
\newblock
\urldef\tempurl%
\url{https://software.intel.com/content/www/us/en/develop/download/intel-64-and-ia-32-architectures-sdm-combined-volumes-1-2a-2b-2c-2d-3a-3b-3c-3d-and-4.html}
\showURL{%
\tempurl}
\newblock
\shownote{Order Number: 325462-083US}.


\bibitem[{ISO}(2011)]%
        {cpp11}
\bibfield{author}{\bibinfo{person}{{ISO}}.} \bibinfo{year}{2011}\natexlab{}.
\newblock \bibinfo{booktitle}{\emph{{ISO} International Standard {ISO/IEC}
  14882:2011(E) – Programming Language C++}}.
\newblock \bibinfo{publisher}{International Organization for Standardization
  (ISO)}.
\newblock
\urldef\tempurl%
\url{https://www.iso.org/standard/50372.html}
\showURL{%
\tempurl}


\bibitem[Kang et~al\mbox{.}(2017)]%
        {promising-semantics}
\bibfield{author}{\bibinfo{person}{Jeehoon Kang}, \bibinfo{person}{Chung-Kil
  Hur}, \bibinfo{person}{Ori Lahav}, \bibinfo{person}{Viktor Vafeiadis}, {and}
  \bibinfo{person}{Derek Dreyer}.} \bibinfo{year}{2017}\natexlab{}.
\newblock \showarticletitle{A promising semantics for relaxed-memory
  concurrency}. In \bibinfo{booktitle}{\emph{Principles of Programming
  Languages ({POPL})}}. \bibinfo{pages}{175--189}.
\newblock
\urldef\tempurl%
\url{https://www.cs.tau.ac.il/~orilahav/papers/popl17.pdf}
\showURL{%
\tempurl}


\bibitem[Kokologiannakis et~al\mbox{.}(2023)]%
        {kokologiannakis-al-23}
\bibfield{author}{\bibinfo{person}{Michalis Kokologiannakis},
  \bibinfo{person}{Ori Lahav}, {and} \bibinfo{person}{Viktor Vafeiadis}.}
  \bibinfo{year}{2023}\natexlab{}.
\newblock \showarticletitle{Kater: Automating Weak Memory Model Metatheory and
  Consistency Checking}. In \bibinfo{booktitle}{\emph{Principles of Programming
  Languages ({POPL})}}, Vol.~\bibinfo{volume}{7}.
\newblock
\urldef\tempurl%
\url{https://doi.org/10.1145/3571212}
\showURL{%
\tempurl}


\bibitem[Lahav et~al\mbox{.}(2017)]%
        {rc11}
\bibfield{author}{\bibinfo{person}{Ori Lahav}, \bibinfo{person}{Viktor
  Vafeiadis}, \bibinfo{person}{Jeehoon Kang}, \bibinfo{person}{Chung-Kil Hur},
  {and} \bibinfo{person}{Derek Dreyer}.} \bibinfo{year}{2017}\natexlab{}.
\newblock \showarticletitle{Repairing sequential consistency in {C/C++11}}. In
  \bibinfo{booktitle}{\emph{{Programming Language Design and Implementation
  (PLDI)}}}. \bibinfo{publisher}{ACM Press}, \bibinfo{pages}{618--632}.
\newblock
\urldef\tempurl%
\url{https://plv.mpi-sws.org/scfix/paper.pdf}
\showURL{%
\tempurl}


\bibitem[Leroy(2021)]%
        {compcert}
\bibfield{author}{\bibinfo{person}{Xavier Leroy}.}
  \bibinfo{year}{2021}\natexlab{}.
\newblock \bibinfo{title}{The {CompCert C} verified compiler}.
\newblock \bibinfo{howpublished}{\url{http://compcert.org/man}}.
\newblock


\bibitem[{Linux Kernel Community}(2007)]%
        {linux-kvm}
\bibfield{author}{\bibinfo{person}{{Linux Kernel Community}}.}
  \bibinfo{year}{2007}\natexlab{}.
\newblock \bibinfo{title}{Linux Kernel-Based Virtual Machine}.
\newblock
\newblock
\urldef\tempurl%
\url{https://git.kernel.org/pub/scm/virt/kvm/kvm.git}
\showURL{%
\tempurl}


\bibitem[{Linux Kernel Mailing List}(1999)]%
        {spinunlock-opt}
\bibfield{author}{\bibinfo{person}{{Linux Kernel Mailing List}}.}
  \bibinfo{year}{1999}\natexlab{}.
\newblock \bibinfo{title}{spin\_unlock optimization(i386)}.
\newblock
\newblock
\urldef\tempurl%
\url{https://lists.archive.carbon60.com/linux/kernel/105412}
\showURL{%
\tempurl}


\bibitem[Margalit and Lahav(2021)]%
        {margalit-lahav-21}
\bibfield{author}{\bibinfo{person}{Roy Margalit} {and} \bibinfo{person}{Ori
  Lahav}.} \bibinfo{year}{2021}\natexlab{}.
\newblock \showarticletitle{Verifying Observational Robustness Against a
  C11-Style Memory Model}.
\newblock \bibinfo{journal}{\emph{Proceedings of the ACM on Programming
  Languages}} \bibinfo{volume}{5}, \bibinfo{number}{POPL} (\bibinfo{date}{Jan.}
  \bibinfo{year}{2021}).
\newblock
\urldef\tempurl%
\url{https://doi.org/10.1145/3434285}
\showURL{%
\tempurl}


\bibitem[{Microsoft Learn}(2021)]%
        {cpp-assembler-msvc}
\bibfield{author}{\bibinfo{person}{{Microsoft Learn}}.}
  \bibinfo{year}{2021}\natexlab{}.
\newblock \bibinfo{title}{Advantages of Inline Assembly}.
\newblock
\newblock
\urldef\tempurl%
\url{https://learn.microsoft.com/en-us/cpp/assembler/inline/advantages-of-inline-assembly}
\showURL{%
\tempurl}


\bibitem[Nipkow et~al\mbox{.}(2002)]%
        {isabelle}
\bibfield{author}{\bibinfo{person}{Tobias Nipkow}, \bibinfo{person}{Lawrence~C.
  Paulson}, {and} \bibinfo{person}{Markus Wenzel}.}
  \bibinfo{year}{2002}\natexlab{}.
\newblock \bibinfo{booktitle}{\emph{Isabelle/HOL --- A Proof Assistant for
  Higher-Order Logic}}. \bibinfo{series}{Lecture Notes in Computer Science},
  Vol.~\bibinfo{volume}{2283}.
\newblock \bibinfo{publisher}{Springer}.
\newblock
\urldef\tempurl%
\url{https://doi.org/10.1007/3-540-45949-9}
\showURL{%
\tempurl}


\bibitem[Podkopaev et~al\mbox{.}(2019)]%
        {imm-19}
\bibfield{author}{\bibinfo{person}{Anton Podkopaev}, \bibinfo{person}{Ori
  Lahav}, {and} \bibinfo{person}{Viktor Vafeiadis}.}
  \bibinfo{year}{2019}\natexlab{}.
\newblock \showarticletitle{Bridging the Gap between Programming Languages and
  Hardware Weak Memory Models}. In \bibinfo{booktitle}{\emph{Principles of
  Programming Languages ({POPL})}}, Vol.~\bibinfo{volume}{3}.
  \bibinfo{publisher}{ACM Press}, \bibinfo{pages}{69:1--69:31}.
\newblock
\urldef\tempurl%
\url{https://doi.org/10.1145/3290382}
\showURL{%
\tempurl}


\bibitem[Preshing(2012)]%
        {preshing-memory-ordering-at-compile-time}
\bibfield{author}{\bibinfo{person}{Jeff Preshing}.}
  \bibinfo{year}{2012}\natexlab{}.
\newblock \bibinfo{title}{Memory Ordering at Compile Time}.
\newblock
\newblock
\urldef\tempurl%
\url{https://preshing.com/20120625/memory-ordering-at-compile-time}
\showURL{%
\tempurl}


\bibitem[Pulte et~al\mbox{.}(2017)]%
        {pulte-al-17}
\bibfield{author}{\bibinfo{person}{Christopher Pulte}, \bibinfo{person}{Shaked
  Flur}, \bibinfo{person}{Will Deacon}, \bibinfo{person}{Jon French},
  \bibinfo{person}{Susmit Sarkar}, {and} \bibinfo{person}{Peter Sewell}.}
  \bibinfo{year}{2017}\natexlab{}.
\newblock \showarticletitle{Simplifying {ARM} Concurrency: Multicopy-Atomic
  Axiomatic and Operational Models for {ARMv8}}, Vol.~\bibinfo{volume}{2}.
  \bibinfo{publisher}{ACM Press}, \bibinfo{pages}{19:1--19:29}.
\newblock
\urldef\tempurl%
\url{https://doi.org/10.1145/3158107}
\showURL{%
\tempurl}


\bibitem[Raad et~al\mbox{.}(2022)]%
        {raad-22}
\bibfield{author}{\bibinfo{person}{Azalea Raad}, \bibinfo{person}{Luc
  Maranget}, {and} \bibinfo{person}{Viktor Vafeiadis}.}
  \bibinfo{year}{2022}\natexlab{}.
\newblock \showarticletitle{Extending Intel-X86 Consistency and Persistency:
  Formalising the Semantics of Intel-X86 Memory Types and Non-Temporal Stores}.
\newblock \bibinfo{journal}{\emph{Proceedings of the ACM on Programming
  Languages}} \bibinfo{volume}{6}, \bibinfo{number}{{POPL}}
  (\bibinfo{date}{Jan.} \bibinfo{year}{2022}), \bibinfo{pages}{22:1--22:31}.
\newblock
\urldef\tempurl%
\url{https://doi.org/10.1145/3498683}
\showURL{%
\tempurl}


\bibitem[Sammler et~al\mbox{.}(2023)]%
        {sammler-al-23}
\bibfield{author}{\bibinfo{person}{Michael Sammler}, \bibinfo{person}{Simon
  Spies}, \bibinfo{person}{Youngju Song}, \bibinfo{person}{Emanuele D'Osualdo},
  \bibinfo{person}{Robbert Krebbers}, \bibinfo{person}{Deepak Garg}, {and}
  \bibinfo{person}{Derek Dreyer}.} \bibinfo{year}{2023}\natexlab{}.
\newblock \showarticletitle{DimSum: A Decentralized Approach to Multi-Language
  Semantics and Verification}, Vol.~\bibinfo{volume}{7}.
  \bibinfo{pages}{27:1--27:31}.
\newblock
\urldef\tempurl%
\url{https://doi.org/10.1145/3571220}
\showURL{%
\tempurl}


\bibitem[Sarkar et~al\mbox{.}(2012)]%
        {sarkar-al-12}
\bibfield{author}{\bibinfo{person}{Susmit Sarkar}, \bibinfo{person}{Kayvan
  Memarian}, \bibinfo{person}{Scott Owens}, \bibinfo{person}{Mark Batty},
  \bibinfo{person}{Peter Sewell}, \bibinfo{person}{Luc Maranget},
  \bibinfo{person}{Jade Alglave}, {and} \bibinfo{person}{Derek Williams}.}
  \bibinfo{year}{2012}\natexlab{}.
\newblock \showarticletitle{Synchronising C/C++ and POWER}
  \emph{(\bibinfo{series}{{Programming Language Design and Implementation
  (PLDI)}})}. \bibinfo{pages}{311–--322}.
\newblock
\urldef\tempurl%
\url{https://doi.org/10.1145/2254064.2254102}
\showURL{%
\tempurl}


\bibitem[Sewell et~al\mbox{.}(2010)]%
        {sewell-al-10}
\bibfield{author}{\bibinfo{person}{Peter Sewell}, \bibinfo{person}{Susmit
  Sarkar}, \bibinfo{person}{Scott Owens}, \bibinfo{person}{Francesco~Zappa
  Nardelli}, {and} \bibinfo{person}{Magnus~O. Myreen}.}
  \bibinfo{year}{2010}\natexlab{}.
\newblock \showarticletitle{X86-TSO: A Rigorous and Usable Programmer's Model
  for X86 Multiprocessors}.
\newblock \bibinfo{journal}{\emph{Commun. ACM}} \bibinfo{volume}{53},
  \bibinfo{number}{7} (\bibinfo{date}{July} \bibinfo{year}{2010}),
  \bibinfo{pages}{89–--97}.
\newblock
\urldef\tempurl%
\url{https://doi.org/10.1145/1785414.1785443}
\showURL{%
\tempurl}


\bibitem[Simner et~al\mbox{.}(2022)]%
        {simner-al-22}
\bibfield{author}{\bibinfo{person}{Ben Simner}, \bibinfo{person}{Alasdair
  Armstrong}, \bibinfo{person}{Jean Pichon-Pharabod},
  \bibinfo{person}{Christopher Pulte}, \bibinfo{person}{Richard Grisenthwaite},
  {and} \bibinfo{person}{Peter Sewell}.} \bibinfo{year}{2022}\natexlab{}.
\newblock \showarticletitle{Relaxed Virtual Memory in Armv8-A}. In
  \bibinfo{booktitle}{\emph{European Symposium on Programming (ESOP)}}
  \emph{(\bibinfo{series}{Lecture Notes in Computer Science},
  Vol.~\bibinfo{volume}{13240})}. \bibinfo{publisher}{Springer},
  \bibinfo{pages}{143--173}.
\newblock
\urldef\tempurl%
\url{https://doi.org/10.1007/978-3-030-99336-8_6}
\showURL{%
\tempurl}


\bibitem[Sindhu et~al\mbox{.}(1992)]%
        {sindhu-92}
\bibfield{author}{\bibinfo{person}{Pradeep~S. Sindhu},
  \bibinfo{person}{Jean-Marc Frailong}, {and} \bibinfo{person}{Michel
  Cekleov}.} \bibinfo{year}{1992}\natexlab{}.
\newblock \bibinfo{booktitle}{\emph{Formal Specification of Memory Models}}.
\newblock \bibinfo{publisher}{Springer}, \bibinfo{pages}{25--41}.
\newblock
\urldef\tempurl%
\url{https://doi.org/10.1007/978-1-4615-3604-8_2}
\showURL{%
\tempurl}


\bibitem[Vafeiadis et~al\mbox{.}(2015)]%
        {vafeiadis-et-al-15}
\bibfield{author}{\bibinfo{person}{Viktor Vafeiadis}, \bibinfo{person}{Thibault
  Balabonski}, \bibinfo{person}{Soham Chakraborty}, \bibinfo{person}{Robin
  Morisset}, {and} \bibinfo{person}{Francesco~Zappa Nardelli}.}
  \bibinfo{year}{2015}\natexlab{}.
\newblock \showarticletitle{Common Compiler Optimisations are Invalid in the
  C11 Memory Model and What We Can Do About It}. In
  \bibinfo{booktitle}{\emph{Principles of Programming Languages ({POPL})}}.
  \bibinfo{publisher}{ACM Press}, \bibinfo{pages}{209--220}.
\newblock
\urldef\tempurl%
\url{https://dl.acm.org/doi/10.1145/2676726.2676995}
\showURL{%
\tempurl}


\end{thebibliography}
